\documentclass[english]{article}

\usepackage[T1]{fontenc}
\usepackage{geometry}
\usepackage{setspace}
\usepackage[latin9]{inputenc}
\usepackage{color}
\usepackage{babel}
\usepackage{verbatim}
\usepackage{float}
\usepackage{amsmath}
\usepackage{amsthm}
\usepackage{amssymb}
\usepackage{graphicx}
\usepackage{esint}
\usepackage[authoryear]{natbib}
\usepackage{bm}
\usepackage[colorlinks=true,pdftitle={Unbiased intractable MCMC}]{hyperref}

\newcommand{\lyxdot}{.}

\floatstyle{ruled}
\newfloat{algorithm}{tbp}{loa}
\providecommand{\algorithmname}{Algorithm}
\floatname{algorithm}{\protect\algorithmname}

\theoremstyle{plain}
\newtheorem{assumption}{\protect\assumptionname}
\theoremstyle{plain}
\newtheorem{thm}{\protect\theoremname}
\theoremstyle{plain}
\newtheorem{prop}{\protect\propositionname}
\theoremstyle{plain}
\newtheorem{lem}{\protect\lemmaname}
\newtheorem{remark}{Remark}

\usepackage{dsfont}

\usepackage{subfig}
\makeatother

\providecommand{\assumptionname}{Assumption}
\providecommand{\lemmaname}{Lemma}
\providecommand{\propositionname}{Proposition}
\providecommand{\theoremname}{Theorem}

\begin{document}

\title{Unbiased Markov chain Monte Carlo for intractable target distributions}
\author{Lawrence Middleton\thanks{Department of Statistics, University of Oxford, UK.}
, George Deligiannidis$^{*}$, Arnaud Doucet$^{*}$, Pierre E. Jacob\thanks{Department of Statistics, Harvard University, USA.}}
\maketitle

\begin{abstract}
Performing numerical integration when the integrand itself cannot
be evaluated point-wise is a challenging task that arises in statistical
analysis, notably in Bayesian inference for models with intractable
likelihood functions. Markov chain Monte Carlo (MCMC) algorithms have
been proposed for this setting, such as the pseudo-marginal method
for latent variable models and the exchange algorithm for a class
of undirected graphical models. As with any MCMC algorithm, the resulting
estimators are justified asymptotically in the limit of the number
of iterations, but exhibit a bias for any fixed number of iterations
due to the Markov chains starting outside of stationarity. This ``burn-in''
bias is known to complicate the use of parallel processors for MCMC
computations. We show how to use coupling techniques to generate unbiased
estimators in finite time, building on recent advances for generic
MCMC algorithms. We establish the theoretical validity of some of
these procedures, by extending existing results to cover the case
of polynomially ergodic Markov chains. The efficiency of the proposed
estimators is compared with that of standard MCMC estimators, with
theoretical arguments and numerical experiments including state space
models and Ising models.
\end{abstract}

\section{Introduction\label{sec:Introduction}}

\subsection{Context \label{subsec:Context}}

For various statistical models the likelihood function cannot be computed
point-wise, which prevents the use of standard Markov chain Monte
Carlo (MCMC) algorithms such as Metropolis\textendash Hastings (MH)
for Bayesian inference. For example, the likelihood of latent variable
models typically involves an intractable integral over the latent
space. Classically, one can address this problem by designing MCMC
algorithms on the joint space of parameters and latent variables.
However, these samplers can mix poorly when latent variables and parameters
are strongly correlated under the joint posterior distribution. Furthermore
these schemes cannot be implemented if we can only simulate the latent
variables and not evaluate their probability density function \citep[Section 2.3]{andrieu2010particle}.
Similarly, in the context of undirected graphical models, the likelihood
function might involve an intractable integral over the observation
space; see \citet{moller2006efficient} with examples from spatial
statistics.

Pseudo-marginal methods have been proposed for these situations \citep{Lin_2000,Beaumont_2003,andrieu2009pseudo},
whereby unbiased Monte Carlo estimators of the likelihood are used
within an MH acceptance mechanism while still producing chains that
are ergodic with respect to the exact posterior distribution of interest,
denoted by $\pi$. Pseudo-marginal algorithms and their extensions
\citep{deligiannidis2015correlated,tran2016block} are particularly
adapted to latent variable models, such as random effects models and
state space models, where the likelihood can be estimated without bias
using importance sampling or particle filters \citep{Beaumont_2003,andrieu2009pseudo,andrieu2010particle}. Related schemes
include the exchange algorithm \citep{murray2006mcmc,andrieu2018utility},
which applies to scenarios where the likelihood involves an intractable,
parameter-dependent normalizing constant. Exchange algorithms rely
on simulation of synthetic observations to cancel out intractable
terms in the MH acceptance ratio. As with any MCMC algorithm, the
computation of each iteration requires the completion of the previous
ones, which hinders the potential for parallel computation. Running
independent chains in parallel is always possible, and averaging over
independent chains leads to a linear decrease of the resulting variance.
However, the inherent bias that comes from starting the chains outside
of stationarity, also called the ``burn-in bias'', remains \citep{rosenthal2000parallel}.

This burn-in bias has motivated various methodological developments
in the MCMC literature; among these, some rely on coupling techniques,
such as the circularly-coupled Markov chains of \citet{neal2017circularly},
regeneration techniques described in \citet{mykland1995regeneration,brockwell2005identification},
and ``coupling from the past'' as proposed in \citet{propp1996exact}.
Coupling methods have also been proposed for diagnosing convergence
in \citet{johnson1996studying,johnson1998coupling} and as a means to assess the approximation error for approximate MCMC kernels in  \citet{nicholls2012coupled}. Recently, a method
has been proposed to completely remove the bias of Markov chain ergodic
averages \citep{glynn2014exact}. An extension of this approach using coupling ideas
was proposed by \citet{jacob2017unbiased} and applied to a variety of MCMC algorithms.
This methodology involves the construction of a pair of Markov chains,
which are simulated until an event occurs. At this point, a certain
function of the chains is returned, with the guarantee that its expectation
is exactly the integral of interest. The output is thus an unbiased
estimator of that integral. Averaging over i.i.d. copies of such estimators
we obtain consistent estimators in the limit of the number of copies,
which can be generated independently in parallel. Relevant limit theorems
have been established in \citet{glynn1990bias,glynn1992asymptotic},
enabling the construction of valid confidence intervals. The methodology
has already been demonstrated for various MCMC algorithms \citep{jacob2017unbiased,heng2017unbiased,jacob2017smoothing},
which were instances of geometrically ergodic Markov chain samplers
under typical conditions. However, in the case of intractable likelihoods
and pseudo-marginal samplers, in realistic situations the associated Markov chains can often be sub-geometrically ergodic, see e.g.\ \citep{andrieu2015convergence}.

We show here that unbiased estimators of $\pi(h)$, with finite variance
and finite computational cost, can also be derived from polynomially ergodic
Markov chains such as those generated by pseudo-marginal methods.
We provide results on the associated efficiency in comparison with
standard MCMC estimators.
We apply the methodology to particle MCMC algorithms
for inference in generic state space models, with an application to
a time series of neuron activation counts. We also consider a variant
of the pseudo-marginal approach known as the block pseudo-marginal approach \citep{tran2016block} as well as the exchange algorithm \citep{murray2006mcmc}.

Accompanying code used for simulations and to generate the figures are provided at \url{https://github.com/lolmid/unbiased_intractable_targets}.

\subsection{Unbiased estimators from coupled Markov chains\label{subsec:Unbiased-estimators-from}}

Let $\pi$ be a probability measure on a topological space $\mathcal{Z}$
equipped with the Borel $\sigma$-algebra $\mathcal{B}(\mathcal{Z})$.
In this section we recall how two coupled chains that are marginally
converging to $\pi$ can be used to produce unbiased estimators of
expectations $\pi(h):=\int h\left(z\right)\pi(dz)$ for any $\pi$-integrable
test function $h:\mathcal{Z}\to\mathbb{R}$. Following \citet{glynn2014exact,jacob2017unbiased},
we consider the following coupling of two Markov chains $(Z_{n})_{n\ge0}$
and $(\tilde{Z}_{n})_{n\ge0}$. First, $Z_{0},\tilde{Z}_{0}$ are
drawn independently from an initial distribution $\pi_{0}$. Then,
$Z_{1}$ is drawn from a Markov kernel $P$ given $Z_{0}$, which
is denoted $Z_{1}|Z_{0}\sim P(Z_{0},\cdot)$. Subsequently, at step
$n\geq1$, a pair $(Z_{n+1},\tilde{Z}_{n})$ is drawn from a Markov
kernel $\bar{P}$ given $(Z_{n},\tilde{Z}_{n-1})$, which is denoted
\textsf{$(Z_{n+1},\tilde{Z}_{n})|(Z_{n},\tilde{Z}_{n-1})\sim\bar{P}((Z_{n},\tilde{Z}_{n-1}),\cdot)$}.
The kernel $\bar{P}$ is such that, marginally, $Z_{n+1}|(Z_{n}, \tilde{Z}_{n-1}) \sim P(Z_{n},\cdot)$
and $\tilde{Z}_{n}|(Z_n, \tilde{Z}_{n-1})\sim P(\tilde{Z}_{n-1},\cdot)$.
This implies that, marginally for all $n\geq0$, $Z_{n}$ and $\tilde{Z}_{n}$
have the same distribution. Furthermore, the kernel $\bar{P}$ is
constructed so that there exists a random variable $\tau$ termed
the meeting time, such that for all $n\ge\tau$, $Z_{n}=\tilde{Z}_{n-1}$
almost surely (a.s.). Then, for any integer $k$, the following informal
telescoping sum argument informally suggests an unbiased estimator
of $\pi(h)$. We start from $\pi(h)=\lim_{n\to\infty}\mathds{E}[h(Z_{n})]$
and write
\begin{align*}
\pi\left(h\right) & =\mathds{E}[h(Z_{k})]+\sum_{n=k+1}^{\infty}\mathds{E}[h(Z_{n})]-\mathds{E}[h(\tilde{Z}_{n-1})] && \text{(write as telescoping sum)},\\
 & =\mathds{E}[h(Z_{k})+\sum_{n=k+1}^{\infty}h(Z_{n})-h(\tilde{Z}_{n-1})\text{]} && \text{(swap expectation \& limit)},\\
 & =\mathds{E}[h(Z_{k})+\sum_{n=k+1}^{\tau-1}h(Z_{n})-h(\tilde{Z}_{n-1})] && \text{(\ensuremath{Z_{n}=\tilde{Z}_{n-1}} for \ensuremath{n\geq\tau})}.
\end{align*}
The sum $\sum_{n=k+1}^{\tau-1}$ is treated as zero if $\tau-1<k+1$.
The suggested estimator is thus defined as
\begin{equation}
H_{k}(Z,\tilde{Z})=h(Z_{k})+\sum_{n=k+1}^{\tau-1}\{h(Z_{n})-h(\tilde{Z}_{n-1})\},\label{eq:defHk}
\end{equation}
with $Z$ and $\tilde{Z}$ denoting the chains $(Z_{n})_{n\ge0}$
and $(\tilde{Z}_{n})_{n\ge0}$ respectively. As in \citet{jacob2017unbiased},
we average $H_{l}(Z,\tilde{Z})$ over a range of values of $l$, $l\in\{k,k+1,...,m\}$
for an integer $m\geq k$, resulting in the estimator
\begin{equation}
H_{k:m}(Z,\tilde{Z})=\frac{1}{m-k+1}\sum_{l=k}^{m}h(Z_{l})+\sum_{n=k+1}^{\tau-1}\min\left(1,\frac{n-k}{m-k+1}\right)(h(Z_{n})-h(\tilde{Z}_{n-1})).\label{eq:bc}
\end{equation}
Intuitively, $H_{k:m}$ can be understood as a standard Markov chain
average after $m$ steps using a burn-in period of $k-1$ steps (which
would be in general biased for $\pi(h)$), plus a second term that
can be shown to remove the burn-in bias. That ``bias correction''
term is a weighted sum of differences of the chains between step
$k$ and the meeting time $\tau=\inf\{n\geq1:\;Z_{n}=\tilde{Z}_{n-1}\}$.
In the following, we will write $H_{k:m}:=H_{k:m}(Z,\tilde{Z})$ for
brevity. The construction of $H_{k:m}$ is summarized in Algorithm
\ref{alg:Hkm}, where the initial distribution of the chains is denoted
by $\pi_{0}$, and the Markov kernels by $P$ and $\bar{P}$ as above.
Standard MCMC estimators require the specification of $\pi_{0}$ and
$P$, while the proposed method requires the additional specification
of the coupled kernel $\bar{P}$. We will propose coupled kernels
for the setting of intractable likelihoods, and study the estimator
$H_{k:m}$ under conditions which cover pseudo-marginal methods.

\begin{algorithm}[H]
\caption{\textbf{Unbiased MCMC estimator }\textsf{$H_{k:m}$}\textbf{ for any
choice of $k$ and $m$ with $0\protect\leq k\protect\leq m$.} \label{alg:Hkm}}

\begin{enumerate}
\item \textsf{Initialization:}
\begin{enumerate}
\item \textsf{Sample }$Z_{0},\tilde{Z}_{0}\sim\pi_{0}$$\left(\cdot\right)$.
\item \textsf{Sample }$Z_{1}|\{Z_{0}=z_{0}\}\sim P(z_{0},\cdot)$.
\item \textsf{Set $n=1$ and $\tau=\infty$.}
\end{enumerate}
\item \textsf{While $n<\max(m,\tau)$}:
\begin{enumerate}
\item \textsf{Sample $(Z_{n+1},\tilde{Z}_{n})|\{Z_{n}=z_{n},\tilde{Z}_{n-1}=\tilde{z}_{n-1}\}\sim\bar{P}\left(\left(z_{n},\tilde{z}_{n-1}\right),\cdot\right)$.}
\item \textsf{If $Z_{n+1}=\tilde{Z}_{n}$ and $\tau=\infty$, set $\tau=n$.}
\item \textsf{Increment} $n$ \textsf{by} $1$.
\end{enumerate}
\item \textsf{Return $H_{k:m}$ as described in Equation (\ref{eq:bc})}.
\end{enumerate}
\end{algorithm}

To see how coupled kernels can be constructed, we first recall a
construction for simple MH kernels.
Focusing, for now, on the typical Euclidean space case $\mathcal{Z}\subseteq\mathds{R}^d$, we assume that $\pi$ admits a density, which with a slight abuse of notation we also denote with $\pi$.
Then the standard MH algorithm relies
on a proposal distribution $q(dz'|z)$, for instance chosen
as a Gaussian distribution centered at $z$. At iteration $n-1$,
a proposal $Z'\sim q(\cdot|Z_{n-1})$ is accepted as the new state
$Z_{n}$ with probability $\alpha_{\mathrm{MH}}(Z_{n-1},Z'):=\min\left(1,\pi(Z')q(Z_{n-1}|Z')/\pi(Z_{n-1})q(Z'|Z_{n-1})\right)$, known as the MH acceptance probability. If $Z'$ is rejected, then $Z_{n}$
is assigned the value of $Z_{n-1}$. This defines the kernel $P$.
To construct $\bar{P}$, following \citet{jacob2017unbiased} we can
consider a maximal coupling of the proposal distributions. This is
described in Algorithm \ref{alg:coupledMH} for completeness; see
also \citet{johnson1998coupling} and \citet{jacob2017unbiased} for a consideration of the cost of sampling from a maximal coupling. 
Here $\mathcal{U}[a,b]$ refers
to the uniform distribution on the interval $[a,b]$. The algorithm
relies on draws from a maximal coupling (or $\gamma$-coupling)
of the two proposal distributions $q\left(\cdot|Z_{n}\right)$ and
$q(\cdot|\tilde{Z}{}_{n-1})$ at step $n\geq1$. Draws $(Z',\tilde{Z}')$
from maximal couplings are such that the probability of the event
$\{Z'=\tilde{Z}'\}$ is maximal over all couplings of $Z'\sim q(\cdot|Z_{n})$
and $\tilde{Z}'\sim q(\cdot|\tilde{Z}_{n-1})$. Sampling from maximal
couplings can be done with rejection sampling techniques as described
in \citet{jacob2017unbiased}, in Section 4.5 of Chapter 1 of \citet{thorisson2000coupling}
and in \citet{johnson1998coupling}. On the event $\{Z'=\tilde{Z}'\}$,
the two chains are given identical proposals, which are then accepted
or not based on $\alpha_{\mathrm{MH}}(Z_{n},Z')$ and $\alpha_{\mathrm{MH}}(\tilde{Z}_{n-1},\tilde{Z}')$
using a common uniform random number. In the event that both proposals
are identical and accepted, then the chains meet: $Z_{n+1}=\tilde{Z}_{n}$.
One can then check that the chains remain identical from that iteration
onwards.

\begin{algorithm}[H]
\caption{\textbf{Sampling from the coupled MH kernel given $(Z_{n},\tilde{Z}_{n-1})$.}
\label{alg:coupledMH}}

\begin{enumerate}
\item \textsf{Sample $Z'$ and $\tilde{Z}'$ from a maximal coupling of
$q\left(\cdot|Z_{n}\right)$ and }$q(\cdot|\tilde{Z}{}_{n-1})$.
\item \textsf{Sample $\mathsf{\mathit{\mathfrak{u}}}\sim\mathcal{U}\left[0,1\right]$.}
\item \textsf{If $\mathit{\mathfrak{u}}<\alpha_{\mathrm{MH}}\left(Z_{n},Z'\right)$
set $Z_{n+1}=Z'$. Otherwise set $Z_{n+1}=Z_{n}$. }
\item \textsf{If $\mathit{\mathfrak{u}}<\alpha_{\mathrm{MH}}(\tilde{Z}{}_{n-1},\tilde{Z}')$
set $\tilde{Z}{}_{n}=\tilde{Z}'$. Otherwise set $\tilde{Z}{}_{n}=\tilde{Z}{}_{n-1}$. }
\item \textsf{Return $(Z_{n+1},\tilde{Z}_{n})$.}
\end{enumerate}
\end{algorithm}

The unbiased property of $H_{k:m}$ has an important consequence for
parallel computation. Consider $R$ independent copies, denoted by
$(H_{k:m}^{(r)})$ for $r=1,\ldots,R$, and the average $\bar{H}_{k:m}^{R}=R^{-1}\sum_{r=1}^{R}H_{k:m}^{(r)}$.
Then $\bar{H}_{k:m}^{R}$ is a consistent estimator of $\pi(h)$ as
$R\to\infty$, for any fixed $(k,m)$, and a central limit theorem
holds provided that $\mathds{V}[H_{k:m}]<\infty$; sufficient conditions
are given in Section \ref{sec:theory}. Since $\tau$ is a random
variable, the cost of generating $H_{k:m}$ is random. Neglecting
the cost of drawing from $\pi_{0}$, the cost amounts to that of one
draw from the kernel $P$, $\tau-1$ draws from the kernel $\bar{P}$,
and then $(m-\tau)$ draws from $P$ if $\tau<m$. Overall that leads
to a cost of $T_{m}:=2(\tau-1)+\max(1,m-\tau+1)$ units, where each unit
is the cost of drawing from $P$, and assuming that one sample from
$\bar{P}$ costs two units. Theoretical considerations on variance
and cost will be useful to guide the choice of the parameters $k$
and $m$ as discussed in Section \ref{subsec:Efficiency-of-unbiased}.


\subsection{Theoretical validity under polynomial tails \label{sec:theory}}

We provide here sufficient conditions under which the estimator $H_{k:m}$ is unbiased,
has finite expected cost and finite variance. Below, Assumptions \ref{assumption:marginaldistributions}
and \ref{assumption:sticktogether} are identical to Assumptions 2.1
and 2.3 in \citet{jacob2017unbiased} whereas Assumption \ref{assumption:meetingtime} is a polynomial
tail assumption on the meeting time weaker than the geometric tail assumption, namely, $\mathds{P}(\tau>n)\leq K\rho^{n}$ for all $n\geq1$,
for some constants $K<\infty$ and $\rho\in(0,1)$, used in \citet{jacob2017unbiased}.
Relaxing this assumption is useful in our context as the pseudo-marginal algorithm
is polynomially ergodic under realistic assumptions \citep{andrieu2015convergence} and, as demonstrated in
Section \ref{subsec:Conditions-for-meeting}, this allows the verification of the polynomial tail assumption.


%
\begin{assumption}
\label{assumption:marginaldistributions}Each of the two chains marginally
starts from a distribution $\pi_{0}$, evolves according to a transition
kernel $P$ and is such that $\mathds{E}[h(Z_{n})]\to\pi(h)$ as
$n\to\infty$ for a real-valued function $h$. Furthermore, there
exists constants $\eta>0$ and $D<\infty$ such that $\mathds{E}[\left|h(Z_{n})\right|{}^{2+\eta}]<D$
for all $n\geq0$.
\end{assumption}
\begin{assumption}
\label{assumption:meetingtime}The two chains are such that there
exists an almost surely finite meeting time $\tau=\inf\{n\geq1:\ Z_{n}=\tilde{Z}_{n-1}\}$
such that $\mathds{P}(\tau>n)\leq Kn^{-\kappa}$ for some constants
$0<K<\infty$ and $\kappa>2\left(2\eta^{-1}+1\right)$, where $\eta$
is as in Assumption \ref{assumption:marginaldistributions}.
\end{assumption}
\begin{assumption}
\label{assumption:sticktogether}The chains stay together after meeting,
i.e. $Z_{n}=\tilde{Z}_{n-1}$ for all $n\geq\tau$.
\end{assumption}
Under Assumption \ref{assumption:meetingtime}, $\mathds{E}[\tau^{p}]\leq Kp\sum_{n\geq0}n^{-\kappa+p-1}$
for all $p\geq1$ and thus $\mathds{E}[\tau^{p}]<\infty$ if $\kappa>p$.
As it is assumed that $\kappa>2\left(2\eta^{-1}+1\right),$ this implies
that $\mathds{E}[\tau^{p}]<\infty$ for $p<2(2\eta^{-1}+1)$. In particular,
one has $\mathds{E}[\tau]<\infty$ and thus the computational cost
associated with $H_{k:m}$ has a finite expectation. It also implies
that $\tau$ has a finite second moment.

The following result states that $H_{k:m}$ has not only a finite expected cost but also
has a finite variance and that its expectation is indeed $\pi(h)$
under the above assumptions. The proof is provided in Appendix \ref{subsec:Proof-of-Theorem-Validity}.
\begin{thm}
\label{thm:validity}Under Assumptions \ref{assumption:marginaldistributions}-\ref{assumption:meetingtime}-\ref{assumption:sticktogether},
for all $k\geq0$ and $m\geq k$, the estimator $H_{k:m}$ defined
in (\ref{eq:bc}) has expectation $\pi(h)$, has a finite expected
computing time and admits a finite variance.
\end{thm}

\subsection{Conditions for polynomial tails\label{subsec:Conditions-for-meeting} }

We now proceed to establishing conditions that imply Assumption \ref{assumption:meetingtime}.
To state the main result, we put assumptions on the probability of
meeting at each iteration. We write $\mathcal{D}$ for the diagonal
of the joint space $\mathcal{Z}\times\mathcal{Z}$, that is $\mathcal{D}:=\{(z,\tilde{z})\in\mathcal{Z}\times\mathcal{Z}:z=\tilde{z}\}$
and introduce the measure $\pi_{\mathcal{D}}(dz,d\tilde{z}):=\pi(dz)\delta_{z}(d\tilde{z})$.
In this case, we identify the meeting time $\tau$ with the hitting time of the diagonal, $\tau=\tau_{\mathcal{D}}:=\inf\left\{ n\geq1:\left(Z_{n},\tilde{Z}_{n-1}\right)\in\mathcal{D}\right\} $.
The first assumption is on the ability of the pair of chains to hit
the diagonal when it enters a certain subset of $\mathcal{Z}\times\mathcal{Z}$.
\begin{assumption}
\label{ass:pbarproperties} The kernel $\bar{P}$ is $\pi_{\mathcal{D}}$-irreducible:
for any set $A\subset\mathcal{D}$ such that $\pi_{\mathcal{D}}(A)>0$
and all $(z,\tilde{z})\in\mathcal{Z}\times\mathcal{Z}$ there exists
some $n\geq0$ such that $\bar{P}^{n}\left((z,\tilde{z}),A\right)>0$.
The kernel $\bar{P}$ is also aperiodic. Finally, there exist $\epsilon\in(0,1)$, $n_0\geq 0$
and a set $C\subset\mathcal{Z}$ such that
\begin{equation}
\inf_{(z,\tilde{z})\in C\times C}\bar{P}^{n_0}\left(\left(z,\tilde{z}\right),\mathcal{D}\right)\geq\epsilon.\label{ass:joint_minor}
\end{equation}
\end{assumption}
Next we will assume that the marginal kernel $P$ admits a polynomial
drift condition and a small set $C$; we will later consider that
small set to be the same set $C$ as in Assumption \ref{ass:pbarproperties}.
Intuitively, the polynomial drift condition on $C$ will ensure regular
entries of the pair of chains in the set $C\times C$, from which
the diagonal can be hit in $n_0$ steps under Assumption \ref{ass:pbarproperties}.

\begin{assumption}
\label{ass:minordrift} There exist $\epsilon_{0}>0$, a probability
measure $\nu$ on $\mathcal{Z}$ and a set $C\subset\mathcal{Z}$
such that
\begin{equation}
\inf_{z\in C}P\left(z,\mathcal{\cdot}\right)\geq\epsilon_{0}\nu(\cdot).\label{ass:marginal_minor}
\end{equation}
In addition, there exist a measurable function $V:\mathcal{Z}\to[1,\infty)$,
constants $b_{V},c_{V}>0$, $\epsilon_{b}\in(0,1)$, and a value $\alpha\in(0,1)$,
such that, defining $\phi(x):=dx^{\alpha}$ for a constant $d>0$
and all $x\in[1,\infty)$, then for any $z\in\mathcal{Z}$,
\begin{align}
PV(z) & \leq V(z)-\phi\circ V(z)+b_{V}\mathds{1}_{C}\left(z\right),\label{eq:drift-2}\\
\sup_{z\in C}V(z) & \leq c_{V},\label{eq:drift1-2}\\
\inf_{z\notin C}\phi\circ V(z) & \geq b_{V}(1-\epsilon_{b})^{-1}.\label{eq:drift2-2}
\end{align}
\end{assumption}
\textbf{}%

The following result states that Assumptions \ref{ass:pbarproperties}
and \ref{ass:minordrift} guarantee that the tail probabilities of
the meeting time are polynomially bounded. The proof is provided in
Appendix \ref{subsec:polyproof}.
\begin{thm}
\label{thm:pmmeetingtimes}Suppose that Assumptions \ref{ass:pbarproperties}
and \ref{ass:minordrift} hold for the same set $C\subset\mathcal{Z}$,
and that $\pi_{0}$ admits a density with respect to $\pi$ and is
supported on a compact set. Then we have that for all $n\geq1$ and
some constant $K>0$,
\[
\mathds{P}(\tau\geq n)\leq Kn^{-\kappa},
\]
where $\kappa=1/(1-\alpha)$, with $\alpha$ defined as in Assumption
\ref{ass:minordrift}.
\end{thm}
We note the direct relation between the exponent $\alpha$ in the
polynomial drift condition and the exponent $\kappa$ in the bound
on the tail probability $\mathds{P}(\tau\geq n)$. In turn this
relates to the existence of finite moments for $\tau$, as discussed
after Assumption \ref{assumption:meetingtime}. In particular, if
we can take large values of $\eta$ in Assumption \ref{assumption:marginaldistributions},
then we require in Assumption \ref{assumption:meetingtime} that $\kappa$
is just above $2$, which is implied by $\alpha>1/2$ according to
Theorem \ref{thm:pmmeetingtimes}. However, if we consider $\eta=1$
in Assumption \ref{assumption:marginaldistributions}, for instance,
then we require in Assumption \ref{assumption:meetingtime} that $\kappa$
is just above $6$, which is implied by $\alpha>5/6$ according to
Theorem \ref{thm:pmmeetingtimes}. The condition $\alpha>5/6$ will appear again in the next section.

\subsection{Efficiency under polynomial tails\label{subsec:Efficiency-of-unbiased}}
\label{sec:effpoly}
In removing the bias from MCMC estimators, we expect that $H_{k:m}$
will have an increased variance compared to an MCMC estimator with
equivalent cost. In this section we study the overall efficiency of
$H_{k:m}$ in comparison to standard MCMC estimators. This mirrors
Proposition 3.3 in \citet{jacob2017unbiased} in the case of geometrically
ergodic chains.%

We can define the inefficiency of the estimator $H_{k:m}$ as the
product of its variance and of its expected computational cost via $\text{IF}[H_{k:m}]:=\mathds{E}[T_{m}]\mathds{V}[H_{k:m}]$, with $T_{m}$ denoting the computational cost. This
quantity appears in the study of estimators with random
computing costs, since seminal works such as \citet{glynn1990bias}
and \citet{glynn1992asymptotic}. The inefficiency can be understood
as the asymptotic variance of the proposed estimator as the computing
budget goes to infinity. The following provides a precise comparison
between this inefficiency and the inefficiency of the standard ``serial''
algorithm. Since the cost $T_{m}$ is measured in units equal to the
cost of sampling from $P$, the cost of obtaining a serial MCMC estimator
based on $m$ iterations is equal to $m$ such units. The mean squared
error associated with an MCMC estimator based on $(Z_{n})_{n\geq0}$
is denoted by $\mathrm{MSE}_{b:m}:=\mathds{E}\left[\left(\mathrm{MCMC}_{b:m}-\pi(h)\right)^{2}\right]$,
where $\text{MCMC}_{b:m}:=(m-b+1)^{-1}\sum_{l=b}^{m}h(Z_{l})$ and
where $b-1$ denotes the number of discarded iterations. We are particularly
interested in the comparison between $\text{IF}[H_{k:m}]$, the inefficiency
of the proposed estimator with parameters $k,m$, and $\lim_{m\to\infty}m\times\mathrm{MSE}_{b:m}$,
the asymptotic inefficiency of the serial MCMC algorithm. Both correspond
to asymptotic variances when the computing budget goes to infinity.

We first express the estimator $H_{k:m}$, for $m\geq k\geq0$ as
$\mathrm{MCMC}_{k:m}+\mathrm{BC}_{k:m}$, where the bias correction
term is
\begin{align}
\mathrm{BC}_{k:m} & :=\sum_{n=k+1}^{\tau-1}\min\left(1,\frac{n-k}{m-k+1}\right)\left(h(Z_{n})-h(\tilde{Z}_{n-1})\right).
\end{align}
Then Cauchy-Schwarz provides a relationship between the variance of
$H_{k:m}$, the MCMC mean squared error, and the second moment of
the bias-correction term:
\begin{equation}
\mathds{V}[H_{k:m}]\le\mathrm{MSE}_{k:m}+2\sqrt{\mathrm{MSE}_{k:m}\mathds{E}\left[\mathrm{BC}_{k:m}^{2}\right]}+\mathds{E}\left[\mathrm{BC}_{k:m}^{2}\right].\label{eq:vhkmineq}
\end{equation}
This relationship motivates the study of the second moment of $\mathrm{BC}_{k:m}$.
The following result shows that if the Markov chains are mixing well
enough, in the sense of the exponent $\alpha$ in the polynomial drift
condition of Assumption \ref{ass:minordrift} being close enough to
one, then we can obtain a bound on $\mathds{E}\left[\mathrm{BC}_{k:m}^{2}\right]$
which is explicit in $k$ and $m$. The proof can be found in Appendix
\ref{subsec:proofprop33}.
\begin{prop}
\label{prop:george33}Suppose that the marginal chain evolving according
to $P$ is $\psi$-irreducible and that the assumptions of Theorem \ref{thm:pmmeetingtimes}
hold for $5/6<\alpha\leq1$ and some measurable function $V:\mathcal{Z}\to[1,\infty)$,
such that $S_{V}:=\{z:V(z)<\infty\}\neq\emptyset$. In addition assume
that there exists a $\gamma\in(1-\alpha,1)$ such that $\pi(V^{4\gamma})<\infty$.
Then for any measurable function $h:\mathcal{Z}\to\mathds{R}$ such
that $\sup_{z\in\mathcal{Z}}V(z)^{-\alpha-\gamma+1}|h(z)|<\infty$,
and any integers $m\geq k\geq0$ we have that, for $\kappa:=1/(1-\alpha)$,
and a constant $B<\infty$,
\begin{equation}
\mathds{E}\left[\mathrm{BC}_{k:m}^{2}\right]\leq B\left[\frac{1}{m^{\kappa/2-1}}+\frac{1}{(m-k+1)^{2}}\frac{1}{k^{\kappa/2-3}}\right].\label{eq:bcvar}
\end{equation}
\end{prop}
The fact that a restriction on the exponent $\alpha$ has to be specified
to control the second moment of $\text{BC}_{k:m}$ is to be expected:
we have already seen in the previous section that such a restriction
is also necessary to apply Theorem \ref{thm:pmmeetingtimes} to verify
Assumption \ref{assumption:meetingtime} with an adequate exponent
$\kappa$, which, in turn, leads to a finite variance for $H_{k:m}$
through Theorem \ref{thm:validity}. The specific condition $5/6<\alpha\leq1$
could perhaps be relaxed with a more refined technical analysis, thus we
interpret the condition qualitatively: the chains are allowed to satisfy
only a polynomial drift condition but it needs to be ``close'' enough
to a geometric drift condition.

It follows from \eqref{eq:vhkmineq} and \eqref{eq:bcvar} that under
the assumptions of Proposition \ref{prop:george33}, we have
\begin{align}
\mathds{V}[H_{k:m}]\le\mathrm{MSE}_{k:m}+2\sqrt{B\mathrm{MSE}_{k:m}}\sqrt{\frac{1}{m^{\kappa/2-1}}+\frac{1}{(m-k+1)^{2}}\frac{1}{k^{\kappa/2-3}}} \nonumber \\
+B\left[\frac{1}{m^{\kappa/2-1}}+\frac{1}{(m-k+1)^{2}}\frac{1}{k^{\kappa/2-3}}\right].\label{eq:boundonvariance}
\end{align}
The variance of $H_{k:m}$ is thus bounded by the mean squared error
of an MCMC estimator, and additive terms that vanish polynomially
when $k$, $m-k$ and $m$ increase. To compare the efficiency of
$H_{k:m}$ to that of MCMC estimators, we add simplifying assumptions
as in \citet{jacob2017unbiased}. As $k$ increases and for $m\geq k$,
we expect $(m-k+1)\mathrm{MSE}_{k:m}$ to converge to $\mathds{V}[(m-k+1)^{-1/2}\sum_{t=k}^{m}h\left(Z_{t}\right)]:=V_{k,m}$
as $m\rightarrow\infty$, where $Z_{k}\sim\pi$. We will make the
simplifying assumption that $\mathrm{MSE}_{k:m}\approx V_{k,m}/(m-k+1)$
for $k$ large enough. As the condition $5/6<\alpha$ is equivalent
to $\kappa>6$, $\mathds{E}\left[\mathrm{BC}_{k:m}^{2}\right]$ will
be negligible compared to the two other terms appearing on the right
hand side of (\ref{eq:boundonvariance}), so
we obtain the approximate inequality
\begin{align*}
\mathds{E}[2(\tau-1)+&\max(1,m-\tau+1)]\mathds{V}[H_{k:m}]\lessapprox\frac{m}{m-k+1}V_{k,m}\\&+2m\sqrt{BV_{k,m}}\sqrt{\frac{1}{\left(m-k+1\right)m^{\kappa/2-1}}+\frac{1}{(m-k+1)^{3}}\frac{1}{k^{\kappa/2-3}}},
\end{align*}
where the cost of $H_{k:m}$ is approximated by the cost of $m$ calls to $P$.
For the left-hand side to be comparable to $V_{k,m}$, we can select
$m$ as a large multiple of $k$ such that $m/(m-k+1)$ is close to one.
The second term on the right-hand side is then negligible as $k$ increases, and we see that the polynomial index determining the rate of decay is monotonic in $\kappa$.

\section{Unbiased pseudo-marginal MCMC \label{sec:Unbiased-pseudo-marginal-method}}

\subsection{Pseudo-marginal Metropolis\textendash Hastings}

The pseudo-marginal approach \citep{Lin_2000,Beaumont_2003,andrieu2009pseudo}
generates Markov chains that target a distribution of interest, while
using only non-negative unbiased estimators of target density evaluations.
For concreteness we focus on target distributions that are posterior
distributions in a standard Bayesian framework. The likelihood function
associated to data $y\in\mathcal{Y}$ is denoted by $\theta\mapsto p(y|\theta)$,
and a prior density $\theta\mapsto p\left(\theta\right)$ w.r.t. the Lebesgue
measure is assigned to an unknown parameter $\theta\in\Theta\subseteq\mathbb{R}^{D}$.
We assume that we can compute a non-negative unbiased estimator of
$p(y|\theta)$, for all $\theta$, denoted by $\hat{p}(y|\theta,U)$
where $U\in\mathcal{U}\subset\mathbb{R}^{M}$ are random variables
such that $U\sim m_{\theta}(du)$, where for any $\theta\in\Theta$,
$m_{\theta}$ denotes a Borel probability measure on $\mathcal{U}$.
We assume that $m_{\theta}(du)$ admits a density with respect to
the Lebesgue measure denoted by $u\mapsto m_{\theta}(u)$. The random
variables $U$ represent variables required in the construction of
the unbiased estimator of $p(y|\theta)$. The pseudo-marginal algorithm
targets a distribution with density
\begin{equation}
(\theta,u)\mapsto\pi(\theta,u)=p(\theta\mid y)\frac{\widehat{p}(y\mid\theta,u)}{p(y\mid\theta)}m_{\theta}\left(u\right).\label{eq:pdpost}
\end{equation}
The generated Markov chain $(Z_{n})_{n\geq0}$ takes values in $\mathcal{Z}=\Theta\times\mathcal{U}$.
Since $\int\widehat{p}(y\mid\theta,u)m_{\theta}\left(u\right)du=p(y|\theta)$
for all $\theta$, marginally $\pi(\theta)=\int\pi(\theta,u)du=p(\theta\mid y)$,
corresponding to the target of interest for the $\theta$ component
of $(Z_{n})_{n\geq0}$. Sampling from $\pi(d\theta,du)$ is achieved
with an MH scheme, with proposal $q\left(d\theta'|\theta\right)m_{\theta'}\left(du'\right)$.
This results in an acceptance probability that simplifies to
\begin{equation}
\alpha_{\mathrm{PM}}\left\{ \left(\theta,\widehat{p}(y\mid\theta,u)\right),(\theta',\widehat{p}(y\mid\theta',u'))\right\} :=\min\left\{ 1,\frac{\widehat{p}(y\mid\theta',u')p(\theta')q\left(\theta|\theta'\right)}{\widehat{p}(y\mid\theta,u)p(\theta)q\left(\theta'|\theta\right)}\right\} ,\label{eq:PM:acceptance}
\end{equation}
which does not involve any evaluation of $u\mapsto m_{\theta}(u)$.
Thus the algorithm proceeds exactly as a standard MH algorithm with
proposal density $q(\theta'|\theta)$, with the difference that likelihood
evaluations $p(y|\theta)$ are replaced by estimators $\hat{p}(y|\theta,U)$
with $U\sim m_{\theta}(\cdot)$. The performance of the pseudo-marginal
algorithm depends on the likelihood estimator: lower variance estimators
typically yield ergodic averages with lower asymptotic variance, but the cost of producing lower variance
estimators tends to be higher which leads to a trade-off analyzed
in detail in \citet{doucet2015biometrika,schmon2018large}.

In the following we will generically denote by $g_{\theta}$ the distribution
of $\widehat{p}(y\mid\theta,U)$ when $U\sim m_{\theta}(\cdot)$,
and for notational simplicity, we might write $\widehat{p}(y\mid\theta)$
instead of $\widehat{p}(y\mid\theta,U)$. The above description defines
a Markov kernel $P$ and we next proceed to defining a coupled kernel
$\bar{P}$, to be used for unbiased estimation as in Algorithm \ref{alg:Hkm}.

\subsection{Coupled pseudo-marginal Metropolis\textendash Hastings}

To define a kernel $\bar{P}$ that is marginally identical to $P$
but jointly allows the chains to meet, we proceed as follows, mimicking
the coupled MH kernel in Algorithm \ref{alg:coupledMH}. First, the
proposed parameters are sampled from a maximal coupling of the two
proposal distributions. If the two proposed parameters $\theta'$
and $\tilde{\theta}'$ are identical, we sample a unique likelihood
estimator $\widehat{p}(y\mid\theta')\sim g{}_{\theta'}\left(\cdot\right)$
and we use it in the acceptance step of both chains. Otherwise, we
sample two estimators, $\widehat{p}(y\mid\theta')\sim g{}_{\theta'}\left(\cdot\right)$
and $\widehat{p}(y\mid\tilde{\theta}')\sim g{}_{\tilde{\theta}'}\left(\cdot\right)$.
Denoting the two states of the chains at step $n\geq1$ by $(\theta_{n},\widehat{p}(y\mid\theta_{n}))$
and $(\tilde{\theta}_{n-1},\widehat{p}(y\mid\tilde{\theta}_{n-1}))$,
Algorithm \ref{alg:coupledPMMH} describes how to obtain $(\theta_{n+1},\widehat{p}(y\mid\theta_{n+1}))$
and $(\tilde{\theta}_{n},\widehat{p}(y\mid\tilde{\theta}_{n}))$;
thereby describing a kernel $\bar{P}$.

\begin{algorithm}[H]
\caption{\textbf{Sampling from the coupled pseudo-marginal MH kernel given
$\left\{ (\theta_{n},\widehat{p}(y\mid\theta_{n})),(\tilde{\theta}_{n-1},\widehat{p}(y\mid\tilde{\theta}_{n-1}))\right\} $.}
\label{alg:coupledPMMH}}

\begin{enumerate}
\item \textsf{Sample $\theta'$ and $\tilde{\theta}'$ from a maximal coupling
of $q\left(\cdot|\theta_{n}\right)$ and }$q(\cdot|\tilde{\theta}{}_{n-1})$.
\item \textsf{If $\theta'=\tilde{\theta}'$, then sample $\widehat{p}(y\mid\theta')\sim g{}_{\theta'}\left(\cdot\right)$
and set }$\widehat{p}(y\mid\tilde{\theta}')=\widehat{p}(y\mid\theta')$.\textsf{}\\
\textsf{Otherwise sample $\widehat{p}(y\mid\theta')\sim g{}_{\theta'}\left(\cdot\right)$
and }$\widehat{p}(y\mid\tilde{\theta}')\sim g{}_{\tilde{\theta}'}\left(\cdot\right)$.
\item \textsf{Sample $\mathsf{\mathit{\mathfrak{u}}}\sim\mathcal{U}\left[0,1\right]$.}
\item \textsf{If $\mathit{\mathfrak{u}}<\alpha_{\mathrm{PM}}\left\{ \left(\theta_{n},\widehat{p}(y\mid\,\theta_{n})\right),(\theta',\widehat{p}(y\mid\theta'))\right\} $
then set $(\theta_{n+1},\widehat{p}(y\mid\theta_{n+1}))=\left(\theta',\widehat{p}(y\mid\theta')\right).$
}\\
\textsf{Otherwise, set }$(\theta_{n+1},\widehat{p}(y\mid\theta_{n+1}))=(\theta_{n},\widehat{p}(y\mid\theta_{n})).$
\item \textsf{If $\mathit{\mathfrak{u}}<\alpha_{\mathrm{PM}}\left\{ (\tilde{\theta}_{n-1},\widehat{p}(y\mid\,\tilde{\theta}_{n-1}),(\tilde{\theta}',\widehat{p}(y\mid\tilde{\theta}'))\right\} $
then set $(\tilde{\theta}_{n},\widehat{p}(y\mid\tilde{\theta}_{n}))=(\tilde{\theta}',\widehat{p}(y\mid\tilde{\theta}')).$
}\\
\textsf{Otherwise, set }$(\tilde{\theta}_{n},\widehat{p}(y\mid\tilde{\theta}_{n}))=(\tilde{\theta}_{n-1},\widehat{p}(y\mid\tilde{\theta}_{n-1})).$
\item \textsf{Return $\left\{ (\theta_{n+1},\widehat{p}(y\mid\theta_{n+1})),(\tilde{\theta}_{n},\widehat{p}(y\mid\tilde{\theta}_{n}))\right\}$.}
\end{enumerate}
\end{algorithm}

In step 2. of Algorithm \ref{alg:coupledPMMH} the two likelihood
estimators $\widehat{p}(y\mid\theta')$ and $\widehat{p}(y\mid\tilde{\theta}')$
can be generated independently, as we will do below for simplicity.
They can also be sampled together in a way that induces positive correlations,
for instance using common random numbers and other methods described
in \citet{deligiannidis2015correlated,jacob2017smoothing}. We leave
the exploration of possible gains in correlating likelihood estimators
in that step as a future avenue of research. An appealing aspect of
Algorithm \ref{alg:coupledPMMH}, particularly when using independent
estimators in step 2., is that existing implementation of likelihood
estimators can be readily used. In Section \ref{subsec:Neuroscience-experiment}
we will exploit this by demonstrating the use of controlled sequential
Monte Carlo \citep{heng2017controlled} in the proposed framework.
Likewise, one could explore the use of other advanced particle filters
such as sequential quasi Monte Carlo \citep{gerber2015sequential}.
To summarize, given an existing implementation of a pseudo-marginal
kernel, Algorithm \ref{alg:coupledPMMH} involves only small modifications
and the extra implementation of a maximal coupling which itself is
relatively simple following, for example, \citet{jacob2017unbiased}.
\begin{remark}
It is worth remarking that the proposed coupling based on maximally coupling the proposals may be sub-optimal, especially in high-dimensional problems where the overlap of the proposals may be quite small. In such cases one may consider more sophisticated couplings, for example reflection couplings, see e.g.\ \cite{bou2018coupling} for an application to Hamiltonian Monte Carlo; see also \cite{heng2017unbiased} and references therein.
\end{remark}
\subsection{Theoretical guarantees}

%
We provide sufficient conditions to ensure that the coupled pseudo-marginal
algorithm returns unbiased estimators with finite variance and finite expected
computation time, i.e. sufficient conditions to satisfy the requirements of  Theorem
\ref{thm:pmmeetingtimes} are provided.
By introducing the
parameterization $w=\hat{p}(y|\theta,u)/p(y|\theta)$ and using the
notation $w\sim\bar{g}_{\theta}\left(\cdot\right)$ when $u\sim m_{\theta}\left(\cdot\right)$,
we can rewrite the pseudo-marginal kernel
\begin{align*}
P\left(\left(\theta,w\right),(d\theta',dw'\right))=q\left(\theta,\theta'\right)&\bar{g}_{\theta'}\left(w'\right)\alpha_{\mathrm{PM}}\left\{ \left(\theta,w\right),(\theta',w')\right\} d\theta'dw'\\&+\varrho_{\mathrm{PM}}\left(\theta,w\right)\delta_{\left(\theta,w\right)}\left(d\theta',dw'\right),
\end{align*}
where, in this parameterization, we write
\[
\alpha_{\mathrm{PM}}\left\{ \left(\theta,w\right),(\theta',w')\right\} =\min\left\{ 1,\frac{\pi\left(\theta'\right)}{\pi\left(\theta\right)}\frac{q(\theta',\theta)}{q(\theta,\theta')}\frac{w'}{w}\right\} ,
\]
and $\varrho_{\mathrm{PM}}\left(\theta,w\right)$ is the corresponding
rejection probability.
We first make assumptions about the target and proposal densities.
\begin{assumption}
\label{assu:The-posterior-density}The target posterior density $\theta \mapsto \pi\left(\theta\right)$
is strictly positive everywhere and continuously differentiable\textup{.
}\textup{\emph{Its tails are super-exponentially decaying and have
regular contours, that is,
\[
\lim_{\left|\theta\right|\rightarrow\infty}\frac{\theta}{\left|\theta\right|}.\nabla\log\pi\left(\theta\right)=-\infty,\qquad\limsup_{\left|\theta\right|\rightarrow\infty}\frac{\theta}{\left|\theta\right|}.\frac{\nabla\pi\left(\theta\right)}{\left|\nabla\pi\left(\theta\right)\right|}<0,
\]
where $\left|\theta\right|$ denotes the Euclidean norm of $\theta$.
Moreover, the proposal distribution satisfies $q\left(\theta,A\right)=\int_{A}q\left(\theta'-x\right)d\theta'$
with a bounded, symmetric density $q$ that is bounded away from zero
on all compact sets. }}
\end{assumption}
We then make assumptions about the moments of the noise.
\begin{assumption}
\label{assu:There-exist-constants}There exist constants $a'>0$ and
$b'>1$ such that
\[
M_{W}:=\textrm{ess\,sup}_{\theta\in\Theta}\int_{\mathds{R^{+}}}\max\left(w^{-a'},w^{b'}\right)\bar{g}_{\theta}(dw)<\infty,
\]
where the essential supremum is taken with respect to the Lebesgue
measure. Additionally the family of distributions defined
by the densities $\bar{g}_{\theta}$ is continuous with respect to
$\theta$ in the topology of weak convergence.
\end{assumption}
Both assumptions are used in \citep{andrieu2015convergence} to establish a drift condition for the pseudo-marginal algorithm.
Assumption \ref{assu:The-posterior-density} can be understood as a condition on the `ideal' algorithm, i.e. if the likelihood could be evaluated exactly, and Assumption \ref{assu:There-exist-constants} ensures the likelihood estimate has neither too much mass around zero nor in the tails.
The following proposition follows from establishing minorization conditions for both the pseudo-marginal and coupled pseudo-marginal kernels along with \citep[Theorem 38]{andrieu2015convergence}.

\begin{prop}
\label{prop:resultsforpseudo}Under Assumptions \ref{assu:The-posterior-density}
and \ref{assu:There-exist-constants}  then Equations \eqref{eq:drift-2}, \eqref{eq:drift1-2}
and \eqref{eq:drift2-2} hold for any $\chi\in$$(0,\min\left(1,a'\right))$,
$a\in\left(\chi,a'\right]$ and $b\in(0,b'-\chi)$ for the drift function defined as
\[
V\left(\theta,w\right):=\left\{ \sup_{\theta\in\mathbb{R^{d}}}\pi\left(\theta\right)\right\} ^{\chi}\pi^{-\chi}\left(\theta\right)\max(w^{-a},w^{b}),
\]
where $\alpha=1-1/b$ and $C=\left\{ \left(\theta,w\right)\in\Theta\times\mathbb{R}^{+}:|\theta|\leq M,w\in[\underline{w},\overline{w}]\right\}$ for
some constants $M\geq1$, $\underline{w}\in(0,1]$ and $\overline{w}>\underline{w}$.
Additionally the minorization conditions \eqref{ass:joint_minor} and \eqref{ass:marginal_minor} hold
for the same $C$ and $n_0=1$. Finally if $\varrho_{\mathrm{PM}}\left(\theta,w\right)<1$
for all $\theta,w$ and if for some $\theta\in B\left(0,M\right)$
we have $\int_{\underline{w}}^{\overline{w}}\bar{g}_{\theta}(w)wdw>0$
then Assumptions \ref{ass:pbarproperties} and \ref{ass:minordrift} hold with the same $C$ for the kernel $\bar{P}$ induced by Algorithm \ref{alg:coupledPMMH}.
\end{prop}

If the assumptions of Proposition \ref{prop:resultsforpseudo} are satisfied for $a',b'$ such that $b'-\min\left(1,a'\right) > 6$ then, by application of Theorem \ref{thm:pmmeetingtimes}, the coupling times exhibit the required tail bounds of Assumption \ref{assumption:meetingtime} with $\alpha>5/6$ - provided also $\pi_{0}$ admits a density with respect to $\pi$ and is
supported on a compact set. We note that the uniform moments bounds of Assumption \ref{assu:There-exist-constants} might not be satisfied in many non-compact parameter spaces. A weaker assumption allowing to satisfy the polynomial drift condition is provided in \citet[Condition 44]{andrieu2015convergence} and could be alternatively used here.

\section{Experiments with coupled pseudo-marginal kernel}
We next present two examples where we are able to verify the conditions guaranteeing the validity of the estimators.

\subsection{Tails of meeting times in a toy experiment\label{subsec:pm-meetings-experiment}}

We provide numerical experiments on the tails of the
meeting time $\tau$ in a toy example, to illustrate the transition from geometric to polynomial tails.
The target $\pi$ is a bivariate Normal
distribution $\mathcal{N}({\mu},{I})$, with ${\mu}=(1,2)\in\mathbb{R}^{2}$ and identity covariance matrix;
the initial distribution $\pi_{0}$ is uniform over the unit square.
Although we can evaluate $\theta\mapsto\pi(\theta)$, in order to emulate the pseudo-marginal setting, we assume instead we have access for each $\theta$ to an unbiased estimator $\hat{\pi}(\theta,W)$ of $\pi(\theta)$, of the form $\hat{\pi}(\theta,W)=\pi(\theta)\times W$ where $W$ is a log-Normal variable; that is $\log W \sim\mathcal{N}(-\sigma^{2}/2,\sigma^{2})$ with $\sigma$ calibrating the precision of $\hat{\pi}(\theta,W)$
of $\pi(\theta)$. We consider a pseudo-marginal Metropolis--Hastings
algorithm with proposal distribution $q(d\theta'|\theta)=\mathcal{N}(d\theta';\theta,I)$,
and a coupled version following Algorithm \ref{alg:coupledPMMH}.
Indeed, in this simplified setting we are able to verify Assumptions \ref{assu:The-posterior-density} and \ref{assu:There-exist-constants} directly.
We note that in the case $\sigma=0$, we recover the standard MCMC setting.

We draw $R=10^5$ independent realizations of the meeting time for
$\sigma$ in a grid of values $\{0,0.5,1,1.5,2\}$. We then approximate tail probabilities
$\mathds{P}(\tau>n)$ by empirical counterparts, for $n$
between $1$ and the $99.9\%$ quantile of the meeting times for each
$\sigma$. The resulting estimates of $\mathds{P}(\tau>n)$ are plotted
against $n$ in Figure \ref{fig:pm-meeting:survprobs}, where the $y$-axis
is in log-scale. First note that in the case $\sigma=0$, $\log\mathds{P}(\tau>n)$
seems to be bounded by a linear function of $n$, which would correspond to
$\mathds{P}(\tau>n)\leq K\rho^{n}$ for some constants $K<\infty$
and $\rho\in(0,1)$. This is indeed the expected behavior in the case
of geometrically ergodic Markov chains \citep{jacob2017unbiased}.

As $\sigma$ increases, $\mathds{P}(\tau>n)$ decreases less rapidly
as a function of $n$. To verify whether $\mathds{P}(\tau>n)$
might be bounded by $Kn^{-\kappa}$ (as our theoretical considerations
suggest), we plot $\mathds{P}(\tau>n)$ against $n$ with both axes in log-scale
in Figure \ref{fig:pm-meeting:survprobs:loglog}, with a focus on the tails, with $n\geq20$.
The figure confirms that $\log \mathds{P}(\tau>n)$ might indeed by upper bounded by $\kappa \log n$, up to a constant offset,
for large enough values of $n$.
The figure suggests also that in this case $\kappa$ decreases with $\sigma$.

\begin{figure}
\begin{centering}
\subfloat[\label{fig:pm-meeting:survprobs}]{\includegraphics[width=0.4\textwidth]{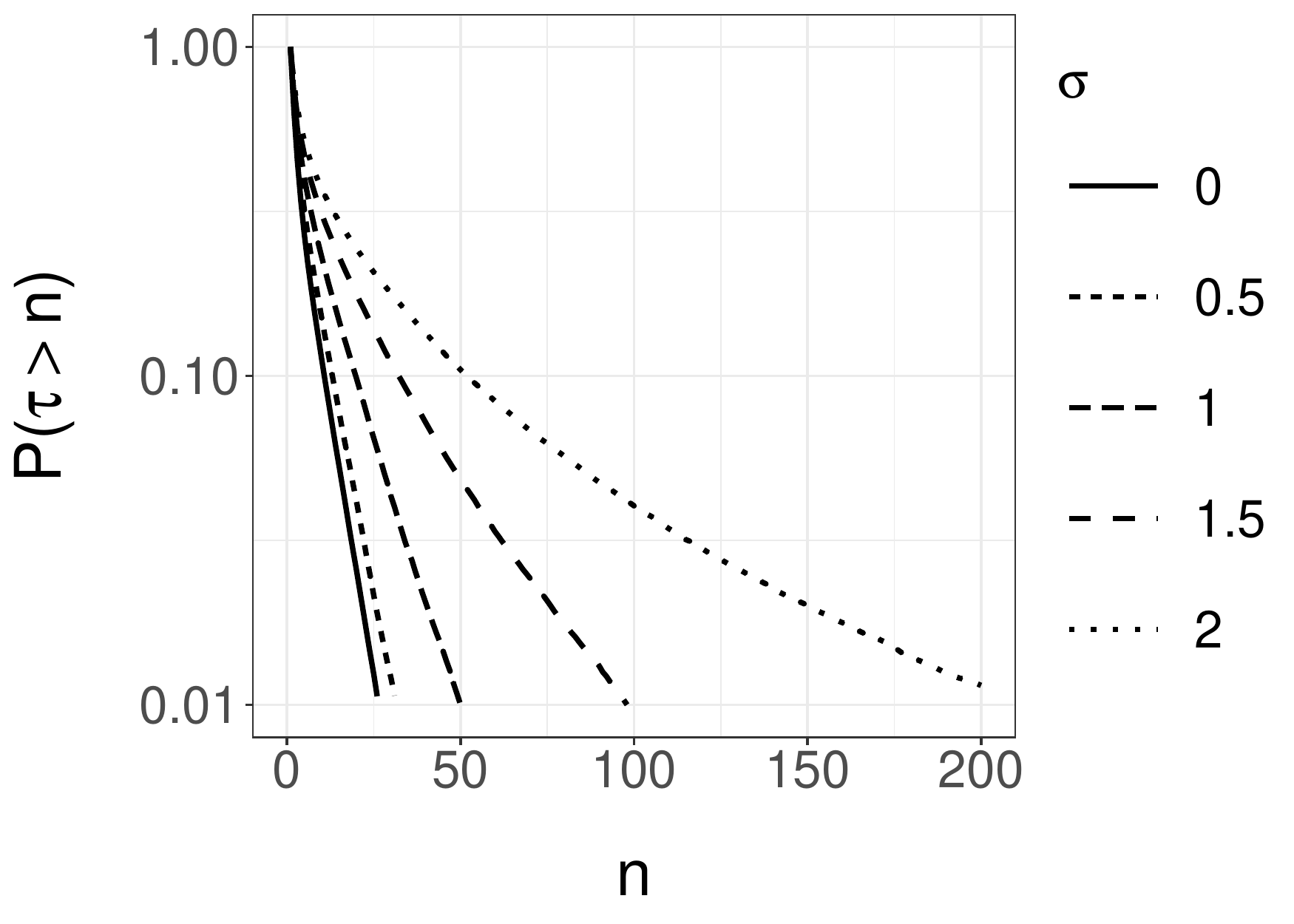}

}\hspace{1cm}\subfloat[\label{fig:pm-meeting:survprobs:loglog}]{\includegraphics[width=0.4\textwidth]{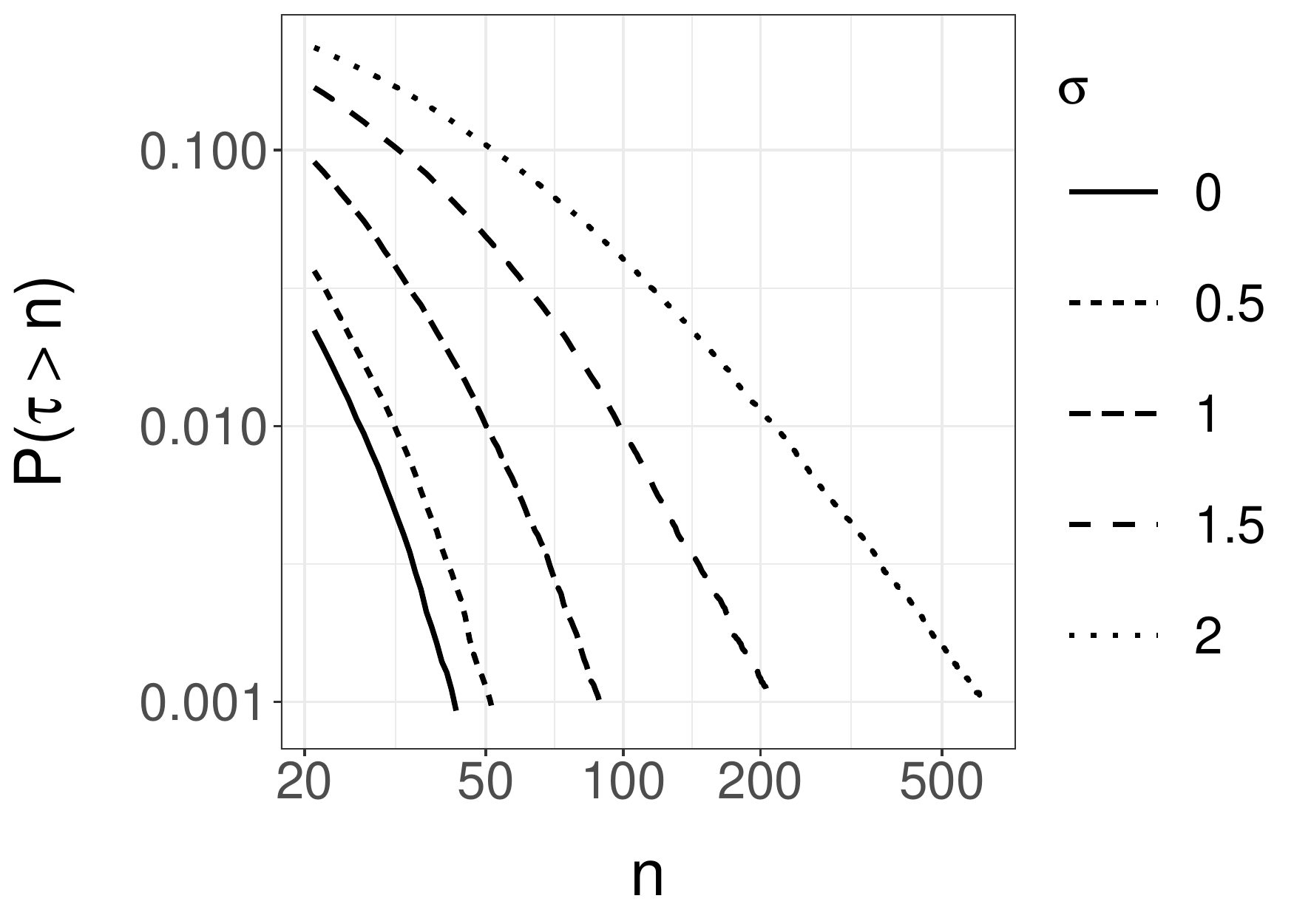}}
\par\end{centering}
\caption{Survival probabilities of the meeting time $\mathds{P}(\tau>n)$ along
$n$, approximated with $10,000$ copies of the meeting times in the pseudo-marginal
toy example of Section \ref{subsec:pm-meetings-experiment}.  Left:
y-axis in log-scale and x-axis in natural scale. Right: log-scale
for both axes, and restriction to $n\geq 20$, in order to focus on the tails. Each
line corresponds to a different value of $\sigma$, which calibrates the amount of noise
in the estimators of target density evaluations.
\label{fig:pm-meeting:bla} }
\end{figure}

\subsection{Beta-Bernoulli model}
\subsubsection{Model description}
\label{sec:randomeffects}
We consider here a random effect model such that, for $t=1,\ldots,T$,
\begin{equation}
X_{t}\stackrel{i.i.d.}{\sim}f_{\theta}(\cdot),\quad\quad\quad\quad Y_{t}|\{X_{t}=x\}\sim g_{\theta}(\cdot|x).\label{eq:SSM-1}
\end{equation}
The likelihood of data $y=(y_{1},...,y_{T})$ is
of the form $p(y|\theta)=\prod_{t=1}^{T}p(y_{t}|\theta)$ where $p(y_{t}|\theta)=\int f_{\theta}(dx)g_{\theta}(y_{t}|x)$
and the likelihood estimator is given by $\hat{p}(y|\theta)=\prod_{t=1}^{T}\hat{p}(y_{t}|\theta)$,
where $\{\hat{p}(y_{t}|\theta)\}_{t=1,\ldots,T}$ are $T$ independent
non-negative unbiased likelihood estimators of $\{p(y_{t}|\theta)\}_{t=1,\ldots,T}$.
These are importance sampling
estimators using a proposal $q_{\theta}(x|y)$ detailed below.

We focus on a Beta-Bernoulli model in which the likelihood is tractable; the latent states
$x_t\in\mathsf{X}=[0,1]$ and observations $y_t\in\{0,1\}$ are such that
\[
f_{\theta}(x_t)=\text{Beta}(x_t;\alpha,\beta),\quad g(y_{t}|x_t)=x_t^{y_{t}}(1-x_t)^{1-y_{t}},
\]
where $\text{Beta}(x;\alpha,\beta)=\text{B}(\alpha,\beta)^{-1}x^{\alpha-1}(1-x)^{\beta-1}$
and $\text{B}(\alpha,\beta)$ denotes the Beta function.

The marginal likelihood of a single observation is given by
$p(y_{t}|\theta)=\alpha^{y_t} \beta^{1-y_t}/(\alpha+\beta)$, and therefore the full marginal likelihood is 
\begin{equation*}
p(y_1, \dots, y_T|\theta) = \frac{\alpha^{T'} \beta^{T-T'}}{(\alpha +\beta)^T}, \qquad 
T'=\sum_{t=1}^T \mathds{1}[y_t =1].
\end{equation*}
Since the likelihood is uniquely determined by the ratio $\beta/\alpha$, we fix
$\alpha>0$ and thus our parameter is given by $\theta=\beta$. We allow $\beta$ to vary in the interval
$\beta\in\Theta = [\underline{\beta},\overline{\beta}]$ bounded away from 0 and
$\infty$. 

We consider likelihood estimator employing the following importance proposal,
\begin{align*}
q_{\theta}(x_t|y_t) & =\begin{cases}
\text{Beta}\left(x_t;1+\alpha,\beta(1+\epsilon)\right) & \quad\text{if}\quad y_t=1,\\
\text{Beta}\left(x_t;\alpha(1+\epsilon),1+\beta\right) & \quad\text{if}\quad y_t=0.
\end{cases}
\end{align*}

Recall that Assumption~\ref{assu:The-posterior-density} was introduced in
\cite{jarner2000geometric} where it was shown to imply geometric ergodicity of
random walk Metropolis. In the present scenario, the state space $\Theta$ of
the marginal algorithm is compact, whence we easily obtain that the marginal
random walk Metropolis algorithm is even \emph{uniformly ergodic}, see for
example \cite[Example~15.3.2]{douc2018markov}.

To establish Assumption~\ref{assu:There-exist-constants}, we need to bound moments of $w=\hat{p}(y_t|\theta)/p(y_t|\theta)$ where
\begin{align*}
\mathds{E}\left[w^{c}\right]=\prod_{t=1}^T \mathds{E}\left[\left(\frac{\hat{p}(y_t|\theta)}{p(y_t|\theta)}\right)^{c}\right],&\quad \hat{p}(y_t|\theta)=\frac{1}{N} \sum_{i=1}^N  \omega(X^{i}_t,y_{t}),\quad \\ \omega(x_t,y_{t}) & =\frac{g(y_{t}|x_t)f_{\theta}(x_t)}{q_{\theta}(x_t|y_{t})}
\end{align*}
for $c>0$ with $X_t^{i}\stackrel{i.i.d.}{\sim} q_{\theta}(\cdot|y_t)$ for $i=1,\ldots,N$.
We have $p(y_{t}=1|\theta)=\alpha/(\alpha+\beta)$ and $p(y_{t}=0|\theta)=\beta/(\alpha+\beta)$, 
thus with $\bar{\omega}(x,y_t):=\omega(x,y_t)/p(y_t |\theta )$ we obtain
\begin{equation}
\bar{\omega}(x,y_{t}=1) \propto (1-x)^{-\varepsilon \beta}, \qquad \bar{\omega}(x,y_{t}=0)
\propto x^{-\alpha \varepsilon}.
\label{eq:incrementweight}
\end{equation}
We see that $\sup_{x\in\mathsf{X}} \bar{\omega}(x,y_{1})=\infty$
suggesting that the associated pseudo-marginal algorithm is not geometrically
ergodic; see \citet[Remark 34]{andrieu2015convergence}. Despite this, we have
$\lim_{\epsilon\rightarrow0}\bar{\omega}(x,y_{t})=1$ for any $\alpha,\beta>0$ and $x\in(0,1)$.
The next proposition, proven in Section~\ref{sec:proofofmoments} in the appendices,  verifies Assumption~\ref{assu:There-exist-constants}.
\begin{prop}
\label{prop:propmoments} For any $\epsilon>0$ and $y\in \{0,1\}$, there exists
$1<b'<1+\epsilon^{-1}$ such that
\[
\sup_{\theta\in\Theta}\mathds{E}_{q_{\theta}}\left[\bar{\omega}(X,y)^{b'}\right]<\infty\quad\text{and}\quad\sup_{\theta\in\Theta}\mathds{E}_{q_{\theta}}\left[\bar{\omega}(X, y)^{-a'}\right]<\infty,
\]
for any $a'>0$. 
Moreover, for any $b'>1$, there
exists $\epsilon$ sufficiently small such that
\[
\sup_{\theta\in\Theta}\mathds{E}_{q_{\theta}}\left[\bar{\omega}(X,y_{t})^{b'}\right]<\infty.
\]
\end{prop}
Through inspection of Proposition~\ref{prop:resultsforpseudo} we see that for
any $\chi\in(0,1)$ we obtain $\kappa=(1-\alpha)^{-1}\in(0,b'-\chi)$. 
Essentially higher, uniformly bounded, moments of the weights
translate to higher moments for the meeting time, and therefore tighter
polynomial bounds for the tail of $\tau$. As a result
we understand the latter part of the proposition qualitatively, in
that the better the proposal the more moments of the meeting time
are bounded and as such the lighter the tail of the meeting time.

\subsubsection{Experiments}
We simulated $T=100$ observations with $\alpha=1$ and $\beta=2$.
We set a uniform prior on $\beta$ on the interval $[0.1,10.0]$.

We ran $100,000$ independent coupled pseudo-marginal algorithms with a random walk proposal with standard deviation 2, employing the maximal coupling between proposals, as in Algorithm \ref{alg:coupledPMMH}.
Figure \ref{fig:post} shows the plot of the (unnormalised) posterior
distribution and contrasts this to the prior.
The distribution of the meeting times was examined
for $N=10$ and $\epsilon\in \{2^{-1},2^{-2},2^{-3},0\}$,  with $\epsilon=0$
corresponding to the exact algorithm where the likelihood is
evaluated exactly. The variance of the log-likelihood estimator
for $\theta=\{\beta\}$ at its true value was estimated to be $\{1.9,0.4,0.1,0\}$
for each of these values respectively, from 1,000 independent likelihood estimators.

The resulting tail probability $\mathds{P}(\tau>n)$ was examined
for the coupling algorithm and is displayed on a log-log scale in Figure \ref{fig:survprob}. In addition to plotting the tail probabilities
in Figure \ref{fig:survprob}, we also plot polynomials of
the form $Cn^{-\kappa'}$ which appear to bound each of the experiments in an attempt to estimate the the true index of the tail $\mathds{P}(\tau>n)$.
For the value of $\epsilon=2^{-3}$, corresponding to the green line,
the meeting times appear to be bounded by $C=2\cdot10^{6}$ and $\kappa'=6$,
therefore guaranteeing that the resulting estimators have finite variance, as per Proposition \ref{prop:george33}.
The remaining polynomials for $\epsilon\in\{2^{-1},2^{-2}\}$
had values $80n^{-2}$ and $2\cdot10^{3}n^{-3.5}$ respectively. In the case 
In all cases, the exponent is smaller in absolute value than $1+\epsilon^{-1}$, the bound predicted by Proposition \ref{prop:propmoments}, noting that $\kappa<b'<1+\epsilon^{-1}$.

\begin{figure}
\begin{centering}
\subfloat[]{\includegraphics[width=0.4\textwidth]{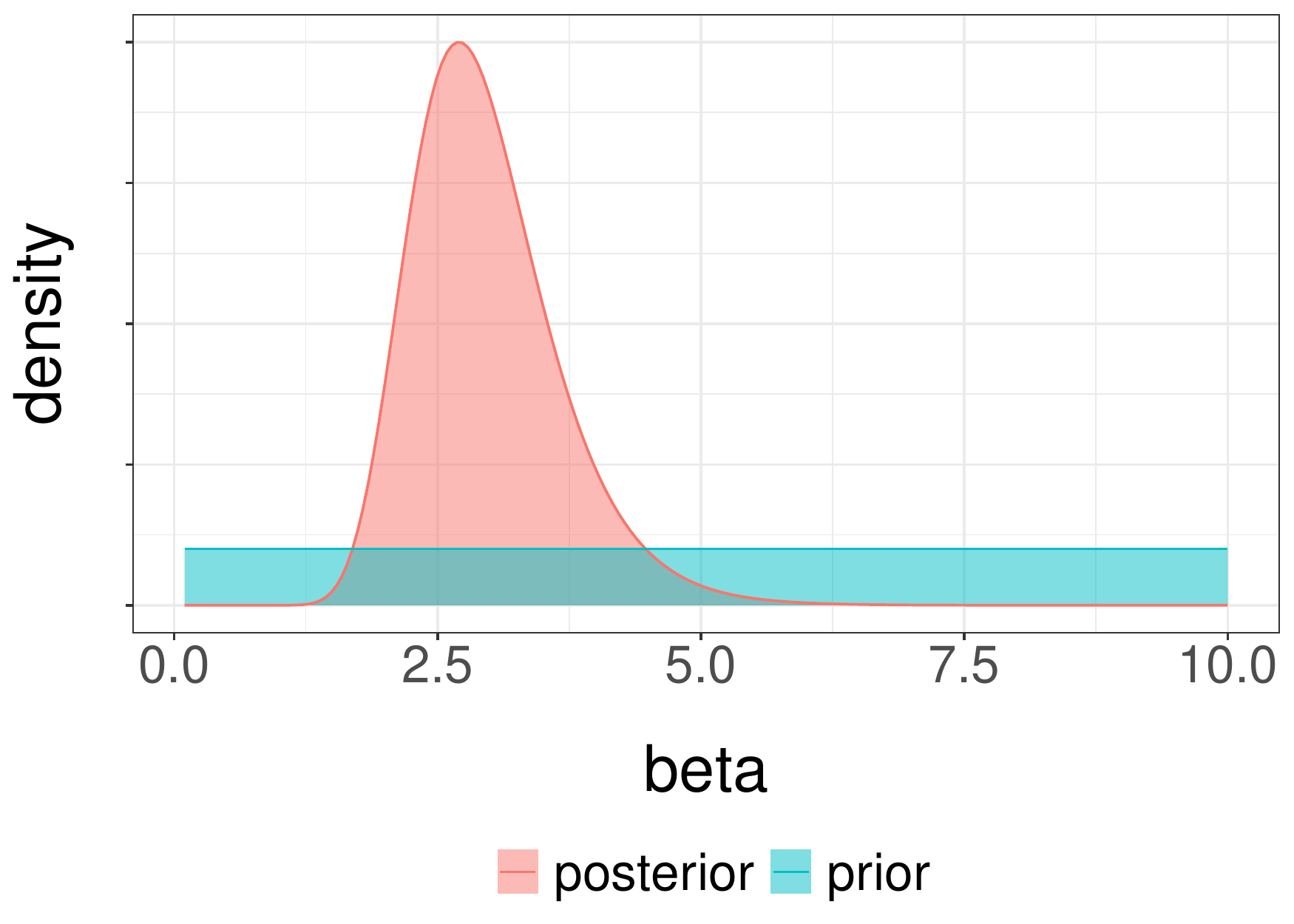}\label{fig:post}}\hspace{1cm}
\subfloat[]{\includegraphics[width=0.4\textwidth]{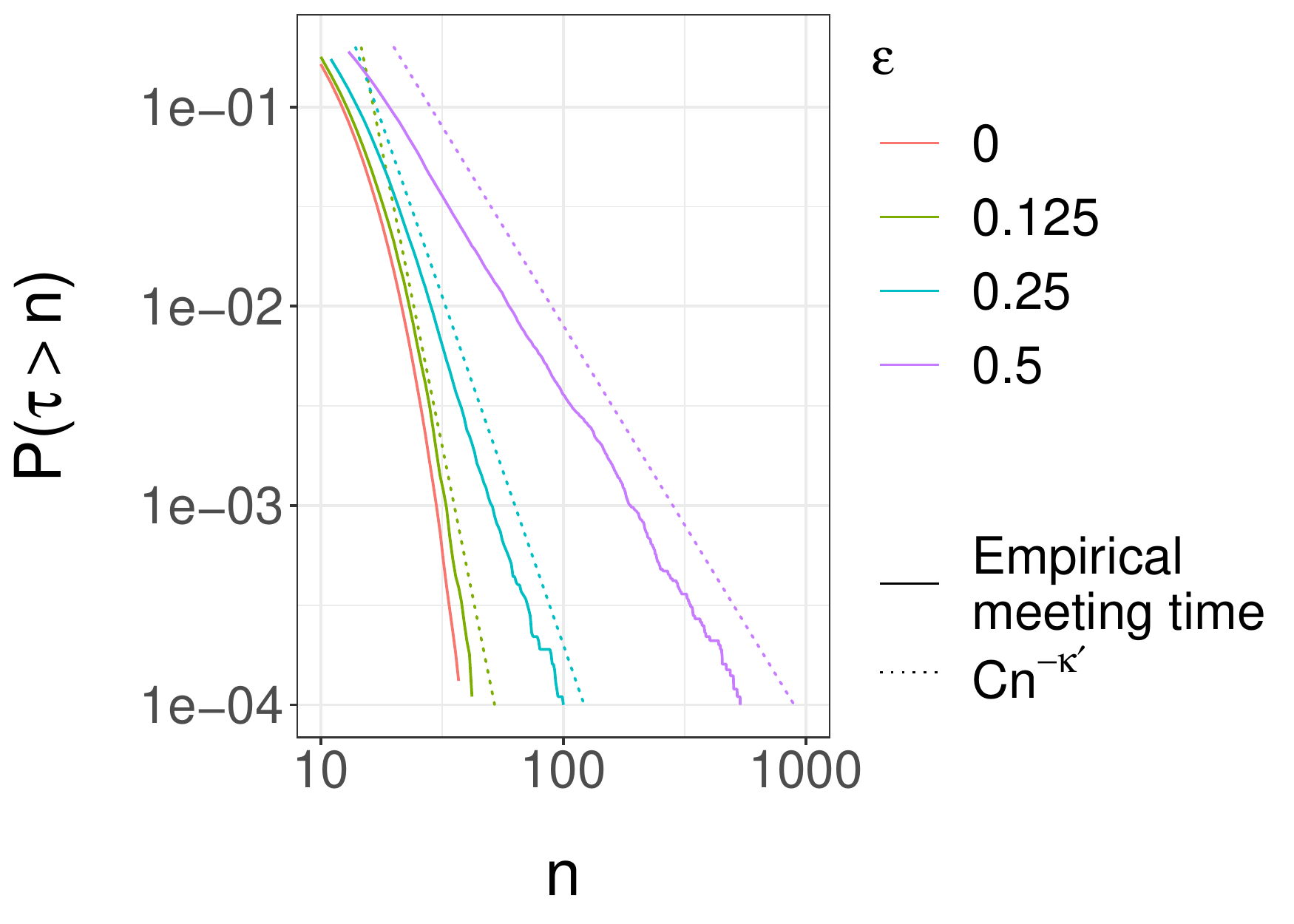}\label{fig:survprob}}

\caption{Beta-Bernoulli model.
Left: Plots of the prior and posterior distribution of paramter $\beta$.
Right: Plots of the tail probability $\mathds{P}(\tau>n)$ for a range of values of $\epsilon$.
Dotted lines show bounding polynomials of the form $Cn^{-\kappa'}$ for each of the values of $\epsilon$.}
\end{centering}
\end{figure}


\section{Experiments in state space models\label{sec:numerics}}

State space models are a popular class of time series models. These
latent variable models are defined by an unobserved Markov process
$(X_{t})_{t\ge0}$ and an observation process $(Y_{t})_{t\ge1}$ where
the observations are conditionally independent given $(X_{t})_{t\ge0}$
with
\begin{equation}
X_{0}\sim\mu_{\theta}(\cdot),\quad\quad X_{t}|\{X_{t-1}=x\}\sim f_{\theta}(\cdot|x),\quad\quad Y_{t}|\{X_{t}=x\}\sim g_{\theta}(\cdot|x),\label{eq:SSM}
\end{equation}
where $\theta$ parameterizes the distributions $\mu_{\theta}$, $f_{\theta}$
and $g_{\theta}$ (termed the `initial', `transition' and `observation'
distribution respectively). Given a realization of the observations
$Y_{1:T}=y_{1:T},$ we are interested in performing Bayesian inference
on the parameter $\theta$ to which we assign a prior density $p(\theta)$.
The posterior density of interest is thus $\pi\left(\theta\right)\propto p(\theta)p(y_{1:T}|\theta)$
where the likelihood $p(y_{1:T}|\theta)=\int\mu_{\theta}(dx_{0})\prod_{t=1}^{T}f_{\theta}(dx_{t}|x_{t-1})g_{\theta}(y_{t}|x_{t})$
is usually intractable. It is possible to obtain a non-negative unbiased
estimator $\widehat{p}(y|\theta,u)$ of $p(y|\theta)$ using particle
filtering where here $u$ represents all the random variables simulated
during the run of a particle filter. The resulting pseudo-marginal
algorithm is known as the particle marginal MH algorithm (PMMH) \citep{andrieu2010particle}.
This algorithm can also be easily modified to perform unbiased smoothing for state inference and is an alternative to existing methods in \citet{jacob2017smoothing}. Guidelines on the selection of the number of particle in this context are provided in \citet{middleton2018unbiasedpimh}.
For state-space models, it is unfortunately extremely difficult to check that Assumptions \ref{assu:The-posterior-density}
and \ref{assu:There-exist-constants} are verified.

\subsection{Linear Gaussian state space model}

The following experiments explore the proposed unbiased estimators
in a linear Gaussian state space model where the likelihood can be
evaluated exactly. This allows a comparison between the pseudo-marginal
kernels, that use bootstrap particle filters \citep{gordon1993novel}
with $N$ particles to estimate the likelihood, and the ideal kernels
that use exact likelihood evaluations obtained with Kalman filters.
We assume $X_{0}\sim\mathcal{N}(0,1),~X_{t}|\{X_{t-1}=x\}\sim\mathcal{N}(ax,\sigma_{X}^{2})$
and $Y_{t}|\{X_{t}=x\}\sim\mathcal{N}(x,1)$ where $a$ and $\sigma_{X}$
are assigned prior distributions, $a\sim\mathcal{U}[0,1]$ and $\sigma_{X}\sim\Gamma(2,2)$.

\subsubsection{Effect of the number of particles }

\label{subsec:effectN}

A dataset of $T=100$ observations was generated from the model with
parameters $a=0.5$ and $\sigma_{X}=1$. We study how the meeting
times and the efficiency vary as a function of $N$, the number of particles. We set the initial
distribution to $\mathcal{U}[0,1]$ over $a$ and $\mathcal{U}[0,5]$
over $\sigma_{X}$, and the proposal covariance of the Normal random
walk proposals to $0.2^{2}I$, corresponding to acceptance rate for the exact algorithm of approximately 36.6\%. In the following we consider a grid
of values for the number of particles, varying $N$ between 50 and 250.

We estimate large quantiles of the distribution of the meeting time over
20,000 repetitions of coupled PMMH, with the results shown in Figure \ref{fig:lgssmmt}.
As expected, increasing $N$ generally reduces the meeting time at the cost of more computation
per iteration.

We examine $\text{IF}[H_{k:m}]$, as defined in section \ref{sec:effpoly}, for the proposed unbiased estimators with
$h:x\mapsto x_{1}+x_{2}+x_{1}^2+x_{2}^2$, for each of these values of $N$ and consider three cases for $k$ and $m$, in particular
\[
(k,m)\in\{(250,500),(250,1000),(750,1000)\}
\]
corresponding to the following: (1) a smaller value of $m-k$, (2) a larger value of $m-k$ and (3) a smaller value of $m-k$ with a more conservative choice of $k$.
Estimates of $\text{IF}[H_{k:m}]$ were obtained using
20,000 repetitions of coupled PMMH where for each value of $(k,m)$ estimators were obtained using a single realisation of the largest value of $m=1,000$ using 30 cores of an Intel Xeon CPU E5-4657L 2.40GHz, taking approximately 60 hours in total.

The results are plotted in Figure \ref{fig:lgssmineff} where we plot also the inefficiency of estimators obtained using coupled
Metropolis-Hastings (horizontal line) for $(k,m)$ as in case (2).
We see first of all that the inefficiency is reduced by increasing $N$ in all cases, and that the inefficiency of estimators obtained using coupled PMMH asymptotes over this range of $N$ towards the inefficiency of estimators obtained using coupled Metropolis-Hastings for $N$ increasing.
We also see that for case (3) that the larger value of $k$ can ameliorate the efficiency of the estimators for small numbers of particles.

We also examine the inefficiency weighted by the cost of obtaining each estimator, i.e. $N\text{IF}[H_{k:m}]$, and compare this to the inefficiency of the serial algorithm using $NV_{\text{as}}$, with
the notation of Section \ref{sec:Introduction}.
Here, $V_{\text{as}}$
was estimated using the \texttt{spectrum0.ar } function in R's \texttt{CODA}
package \citep{plummer2006coda}, averaging over 10 estimators obtained through running the serial algorithm for 500,000
iterations and discarding the first 10\% as burn-in.
Figure \ref{fig:lgssm-ineffN-and-zoom} shows the results of this procedure, showing $\pm2$ sample standard errors for the inefficiency estimates.
Figure \ref{fig:lgssm-ineffN} demonstrates that despite the lower cost of obtaining unbiased estimators for lower values of $N$, the initial decline in inefficiency is still significant.
In Figure \ref{fig:lgssm-ineffNzoom} we show the same results though with a focus around the optimum inefficiency.
Here, we see that the optimum is attained at $N=100$ with value $NV_{\text{as}}=640$ for the serial algorithm and at $N=150$ with $N\text{IF}[H_{k:m}]=980$ for case (2).
Therefore, we see that the increase in inefficiency is estimated to be under 55\% relative to a well-tuned serial algorithm for the values considered.
Indeed, for this particular batch of $N=150$ and $m=1,000$, the parallel execution time to obtain the estimators $H_{k:m}$ on the stated machine was under 14 hours, which we compare to approximately 12 days of serial execution time if performed all on a single core (the mean time to obtain an estimator was 53 seconds) or 8 days after accounting for the increase in inefficiency of 55\%.

\begin{figure}[H]
\centering{}\subfloat[]{\begin{centering}
\includegraphics[width=0.4\textwidth]{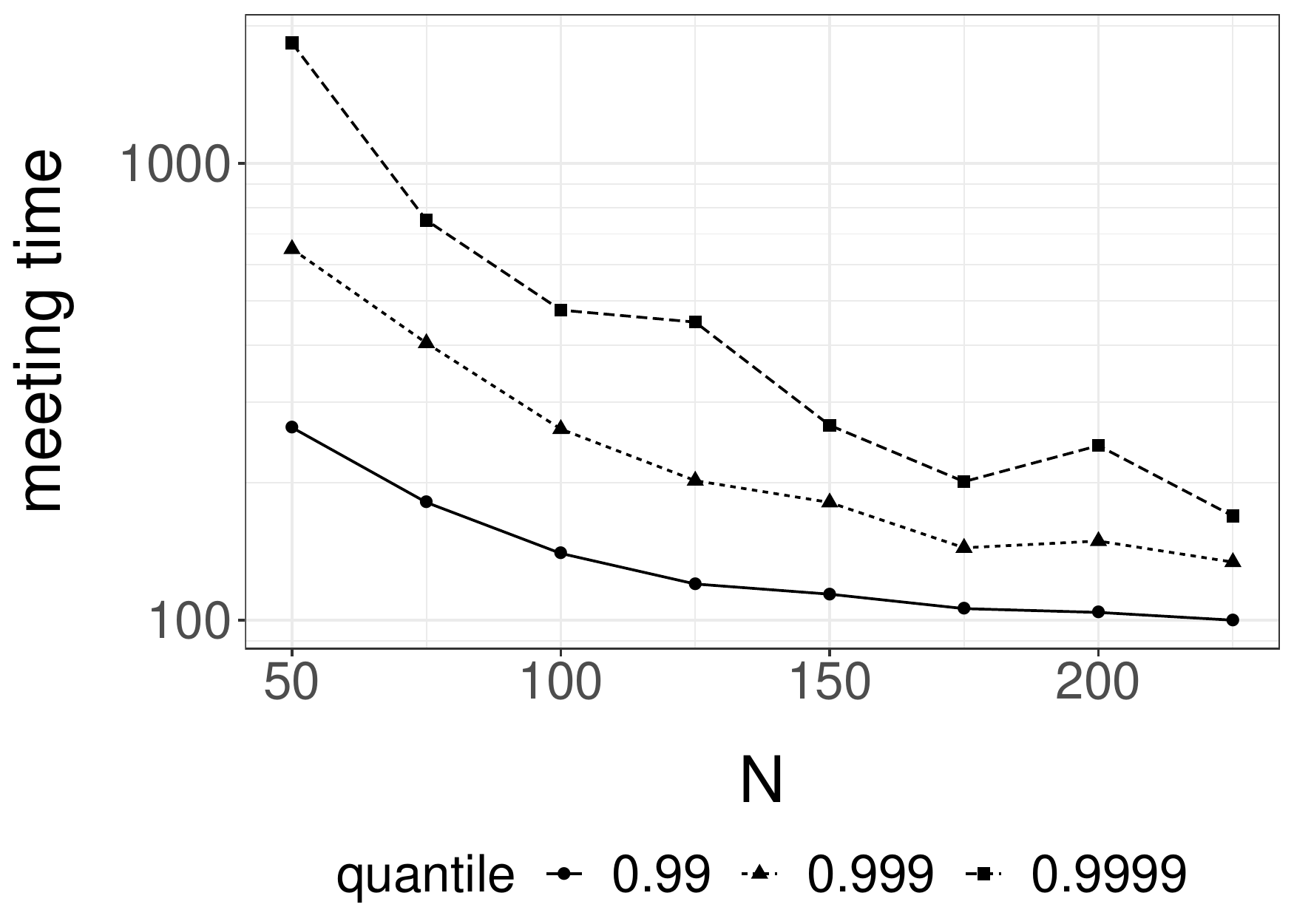}
\par\end{centering}
\label{fig:lgssmmt}}$\quad$\subfloat[]{\begin{centering}
\includegraphics[width=0.4\textwidth]{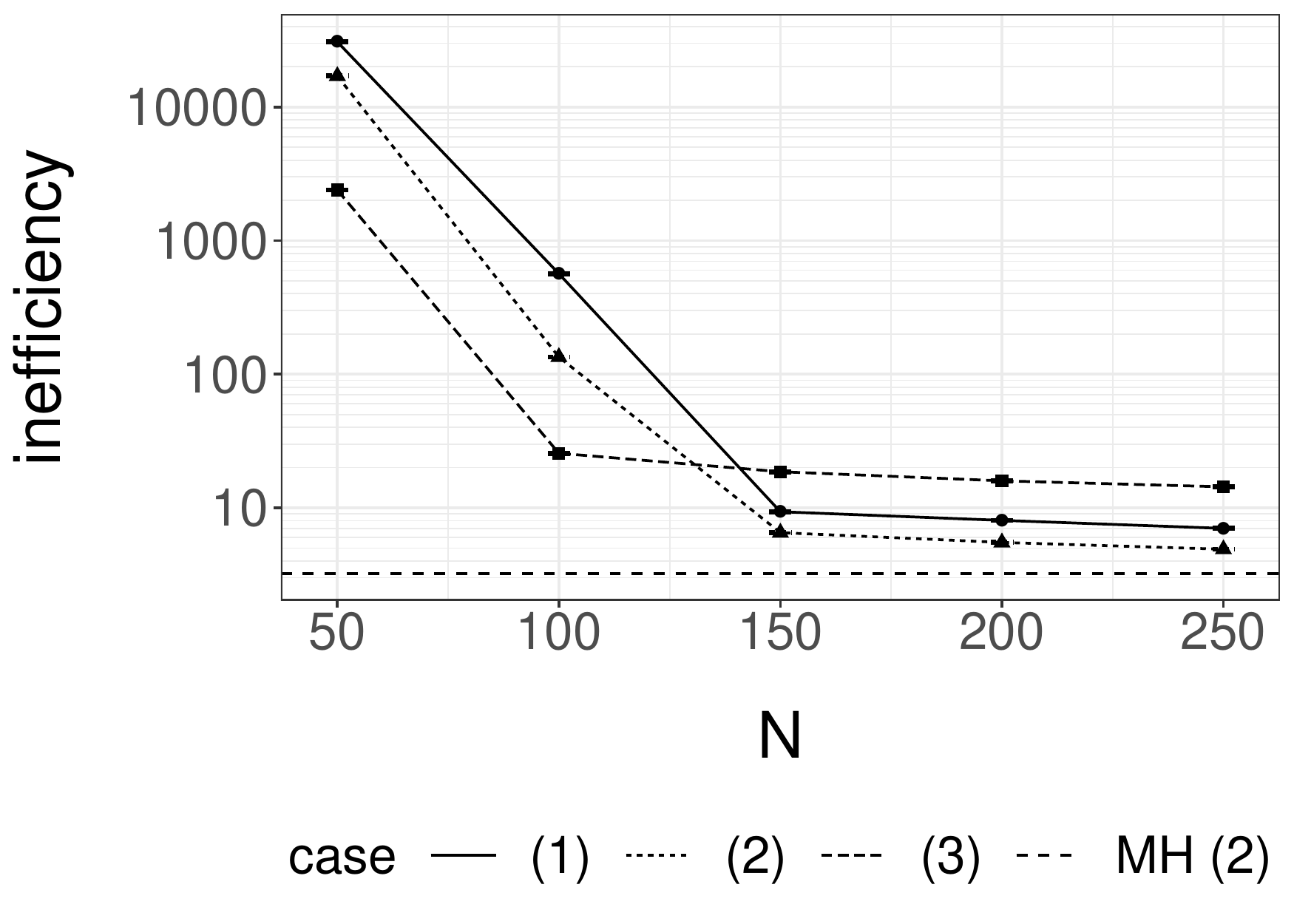}
\par\end{centering}
\label{fig:lgssmineff}}\caption{Coupled PMMH meeting times and inefficiency of estimators for a linear Gaussian state space model with $T=100$ observations and over a range of particles, $N$. Left: estimates of the quantiles of the meeting times. Right: inefficiencies for serial PMMH as a function of $N$, compared to the inefficiency of unbiased estimators obtained using coupled MH. }
\label{fig:lgssmcp}
\end{figure}

\begin{figure}[H]
\centering{}\subfloat[]{\begin{centering}
\includegraphics[width=0.4\textwidth]{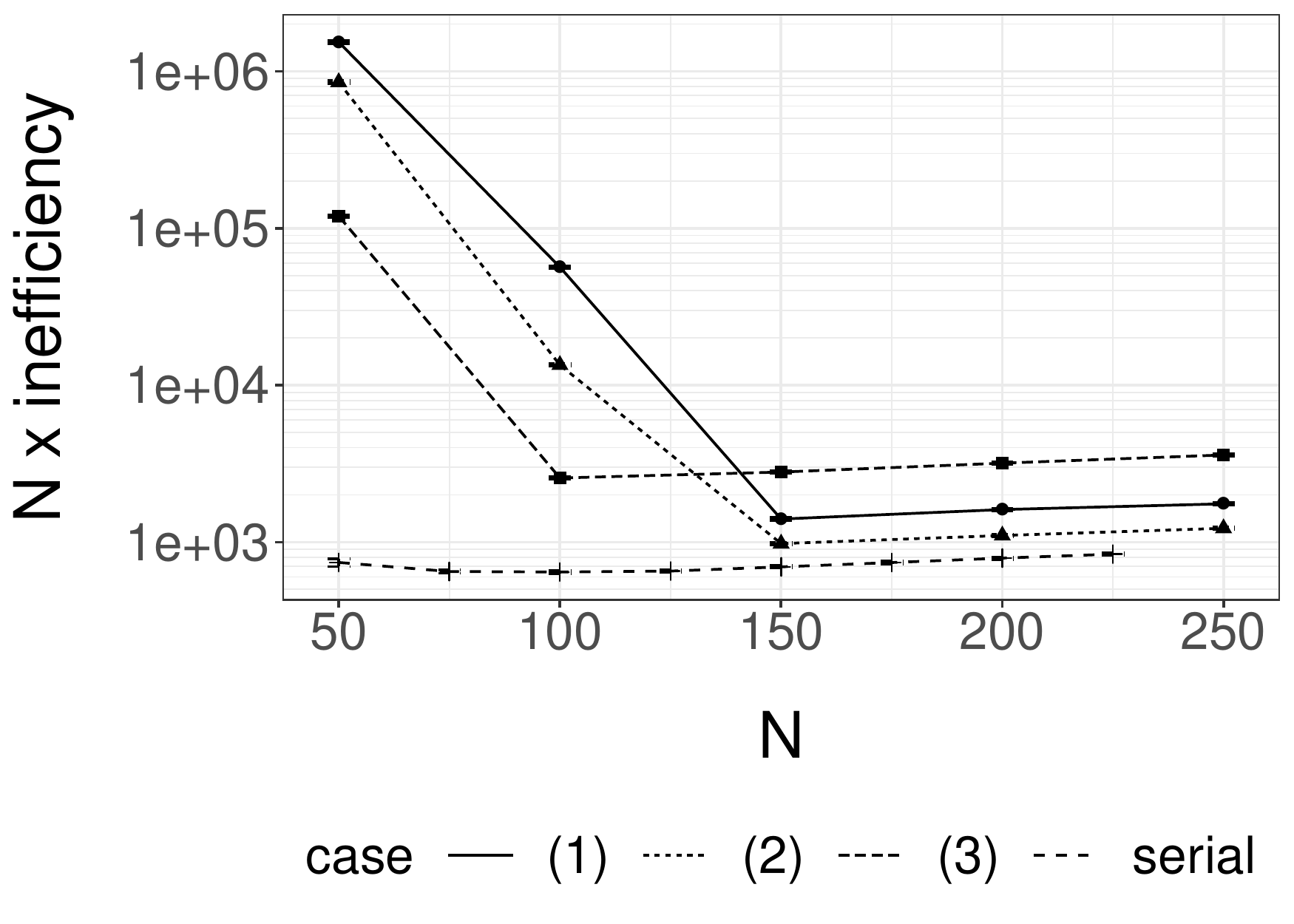}
\par\end{centering}
\label{fig:lgssm-ineffN}}$\quad$\subfloat[]{\begin{centering}
\includegraphics[width=0.4\textwidth]{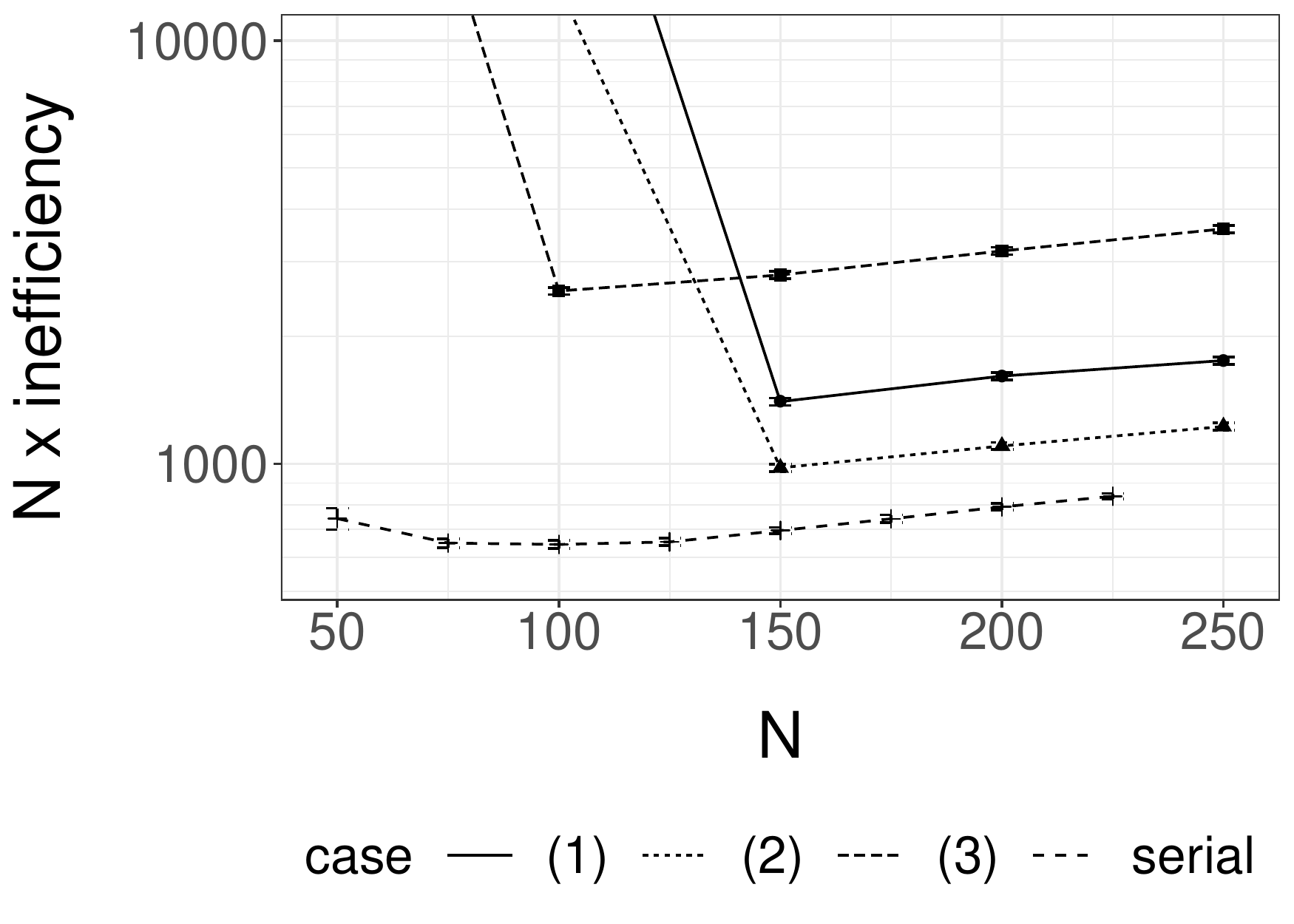}
\par\end{centering}
\label{fig:lgssm-ineffNzoom}}\caption{Inefficiencies weighted by $N$ for a linear Gaussian state space model, comparing directly the inefficiency of estimators obtained using the serial algorithm to those obtained using coupled PMMH. Left: inefficiencies weighted by $N$. Right: inefficiencies weighted by $N$ close to their optima.}
\label{fig:lgssm-ineffN-and-zoom}
\end{figure}

\subsubsection{Effect of the time horizon}

We investigate the distribution of meeting times as a function of
$T$, with $N$ scaling linearly with $T$. Such a scaling is motivated
through the guarantee that the variance of the log-likelihood estimates
obtained at each iteration are asymptotically constant \citep{berard2014lognormal,deligiannidis2015correlated,schmon2018large}.
For the model as before, we consider a grid of $T\in\{100,...,1000\}$,
using a single realisation of the data. Throughout the following,
we fix the proposal covariance to be $\frac{2^{2}}{T}I$, coinciding
with the proposal covariance in \ref{subsec:effectN} for $T=100$, providing an acceptable acceptance rate for the exact algorithm and where $1/T$ is motivated as a result of the variance of the posterior contracting at a rate proportional to $1/T$.

We consider two cases. Firstly, we examine how the distribution of
meeting time changes for a fixed initial distribution (the distribution
used previously of $\mathcal{U}[0,1]$ over $a$ and $\mathcal{U}[0,5]$
over $\sigma_{X}$); we refer to this as Scaling 1. Secondly, for
Scaling 2, we examine how the distribution of meeting times changes
if we also scale the initial distribution by setting $\pi_{0}=\mathcal{N}(\mu^{*},\frac{50}{T}I)$,
truncated to ensure it is dominated by the prior and where $\mu^{*}$
denotes the true parameter values.

In both cases we compare the distribution of meeting times for $N=T$
with the distribution of meeting times for the exact algorithm (i.e. $\bar{P}$ as in Algorithm \ref{alg:coupledMH}) with
likelihood evaluations performed using the Kalman filter. Figure \ref{fig:scalingT1}
and \ref{fig:scalingT2} show estimates of the $80^{th}$ and $99^{th}$
percentile over 1,000 repetitions for Scaling 1 and Scaling 2 respectively.
Firstly, it can be seen that in all cases the meeting times for coupled
PMMH are higher than the meeting times for coupled MH. Furthermore
the smaller difference between
the $80^{th}$ percentiles, compared to the difference between the $99^{th}$
percentiles, reflects a heavier tail of the distribution of the meeting time
in the case of PMMH. Finally, it can be seen that out of the two
scalings Scaling 2 appears to stabilise for larger values of $T$
whereas Scaling 1 exhibits an increase with $T$.

\begin{figure}[H]
\begin{centering}
\subfloat[]{\begin{centering}
\includegraphics[width=0.4\textwidth]{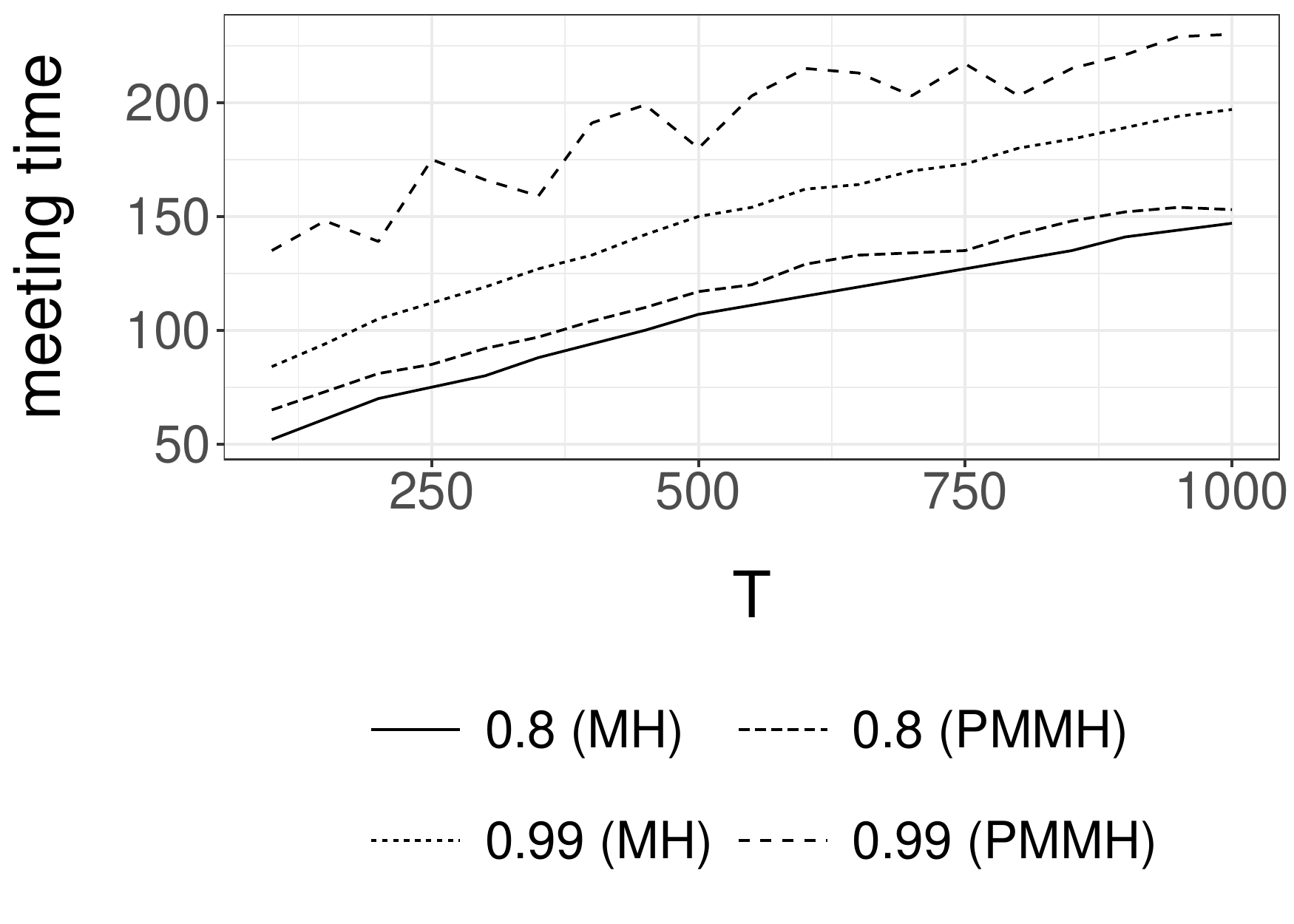}
\par\end{centering}
\label{fig:scalingT1}}$\qquad$\subfloat[]{\begin{centering}
\includegraphics[width=0.4\textwidth]{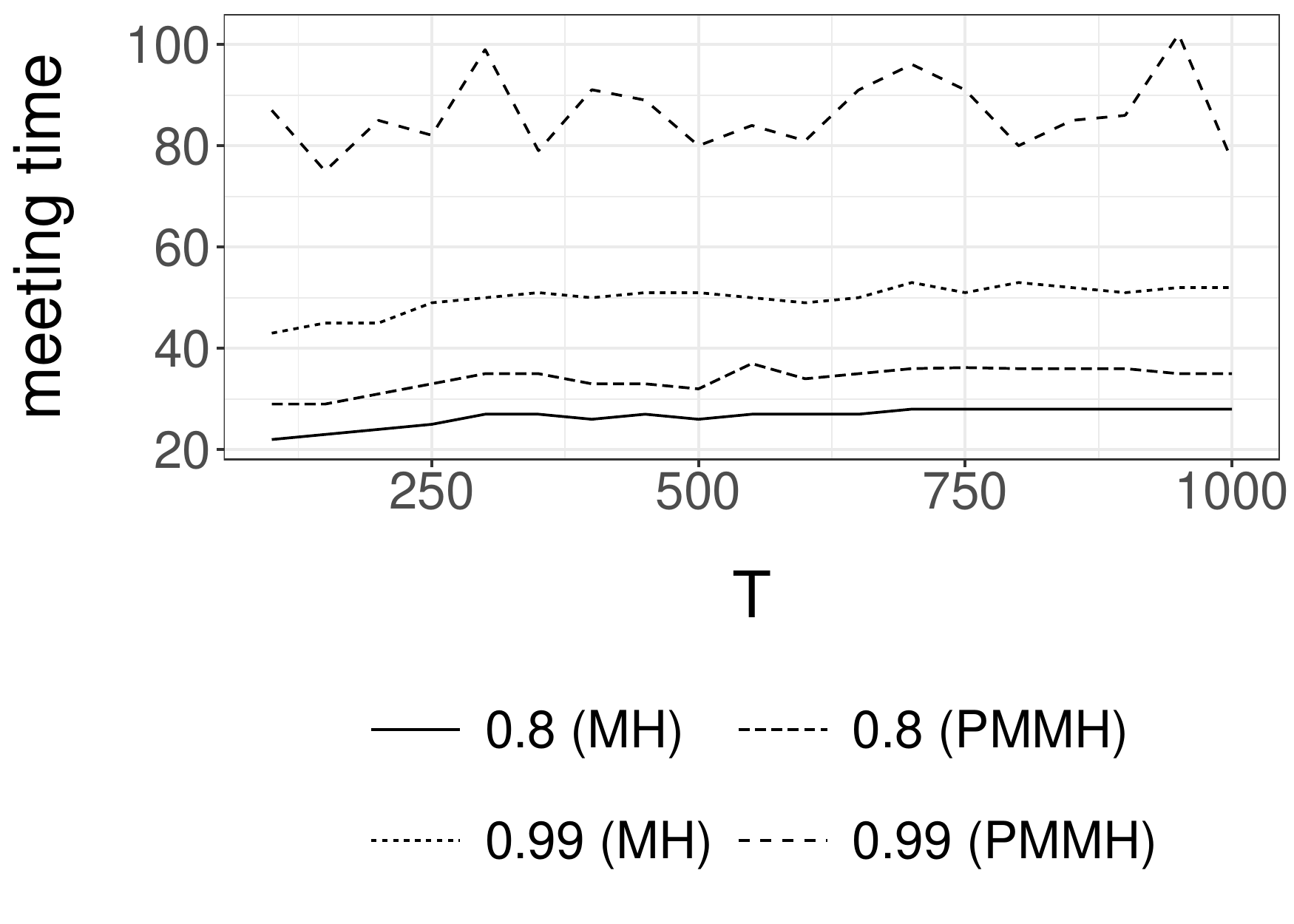}
\par\end{centering}
\label{fig:scalingT2}}
\par\end{centering}
\caption{Scaling quantiles of meeting times with $T$ over 1,000 repetitions. Left: fixing the initial distribution and scaling the proposals (Scaling 1). Right: scaling both the proposals and the initial distribution (Scaling 2).}
\label{fig:scaling1}
\end{figure}

\subsection{Neuroscience experiment \label{subsec:Neuroscience-experiment}}

We apply the proposed methodology to a neuroscience experiment described
in \citet{temereanca2008rapid}. The same data and model were used
to illustrate the controlled Sequential Monte Carlo (cSMC) algorithm
in \citet{heng2017controlled}.

\subsubsection{Model, data and target distribution}

The model aims at capturing the activation of neurons of rats as their
whiskers are being moved with a periodic stimulus. The experiment
involves $M=50$ repeated experiments, and $T=3000$ measurements
(one per millisecond) during each experiment. The activation of a
neuron is recorded as a binary variable for each time and each experiment.
These activation variables are then aggregated by summing over the
$M$ experiments at each time step, yielding a series of variables
$Y_{t}$ taking values between $0$ and $M$; see \citet{zhang2018estimating}
for an alternative analysis that avoids aggregating over experiments.
Letting $\text{Bin}(\cdot;n,p)$ denote the binomial distribution
for $n$ trials with success probability $p$, the model for neuron
activation is given by $X_{0}\sim\mathcal{N}(0,1)$ and, for $t\geq1$,
\begin{align*}
X_{t}|\{X_{t-1}=x\} & \sim\mathcal{N}(\cdot;ax,\sigma_{X}^{2}),\quad Y_{t}|\{X_{t}=x\}\sim\text{Bin}(\cdot;M,s(x))
\end{align*}
where $s(x):=(1+\exp(-s))^{-1}$. We focus on the task of estimating
$(a,\sigma_{X}^{2})$ from the data using the proposed method. Following
\citet{heng2017controlled} we specify a uniform prior on $[0,1]$
for $a$ and an inverse-Gamma prior on $\sigma_{X}^{2}$ with parameters
$(1,0.1)$, where the probability density function of an inverse-Gamma
with parameters $(a,b)$ is $x\,\mapsto\Gamma(a)^{-1}b^{a}x^{-a-1}\exp(-b/x)$.
The PMMH kernels employed below use a Gaussian random walk proposal.
The likelihood is estimated with cSMC with $N=128$ particles and
$3$ iterations, where the exact specification is taken from the appendix
of \citet{heng2017controlled}. Such cSMC runs take approximately
one second, on a 2015 desktop computer and a simple R implementation.
Figure \ref{fig:neuro:datalikelihood} presents the time series of
observations (\ref{fig:neuro:data}) and the estimated log-posterior
density (\ref{fig:neuro:likelihood}), obtained on a $500\times500$
grid of parameter values, and one cSMC likelihood estimate per parameter
value. In Figure \ref{fig:neuro:likelihood}, the upper right corner
presents small black circles, generated by the contour plot function,
which indicate high variance in the likelihood estimators for these
parameters. Thus we expect PMMH chains to have a lower acceptance
rate in that part of
the space. On the other hand, the maximum likelihood estimate (MLE)
is indicated by a black dot on the bottom right corner. The variance
of the log-likelihood estimators is of the order of $0.2$ around
the MLE, so that PMMH chains are expected to perform well there, as
was observed in \citet{heng2017controlled} where the chains were
initialized close to the MLE.

\begin{figure}
\begin{centering}
\subfloat[\label{fig:neuro:data}]{\includegraphics[width=0.4\textwidth]{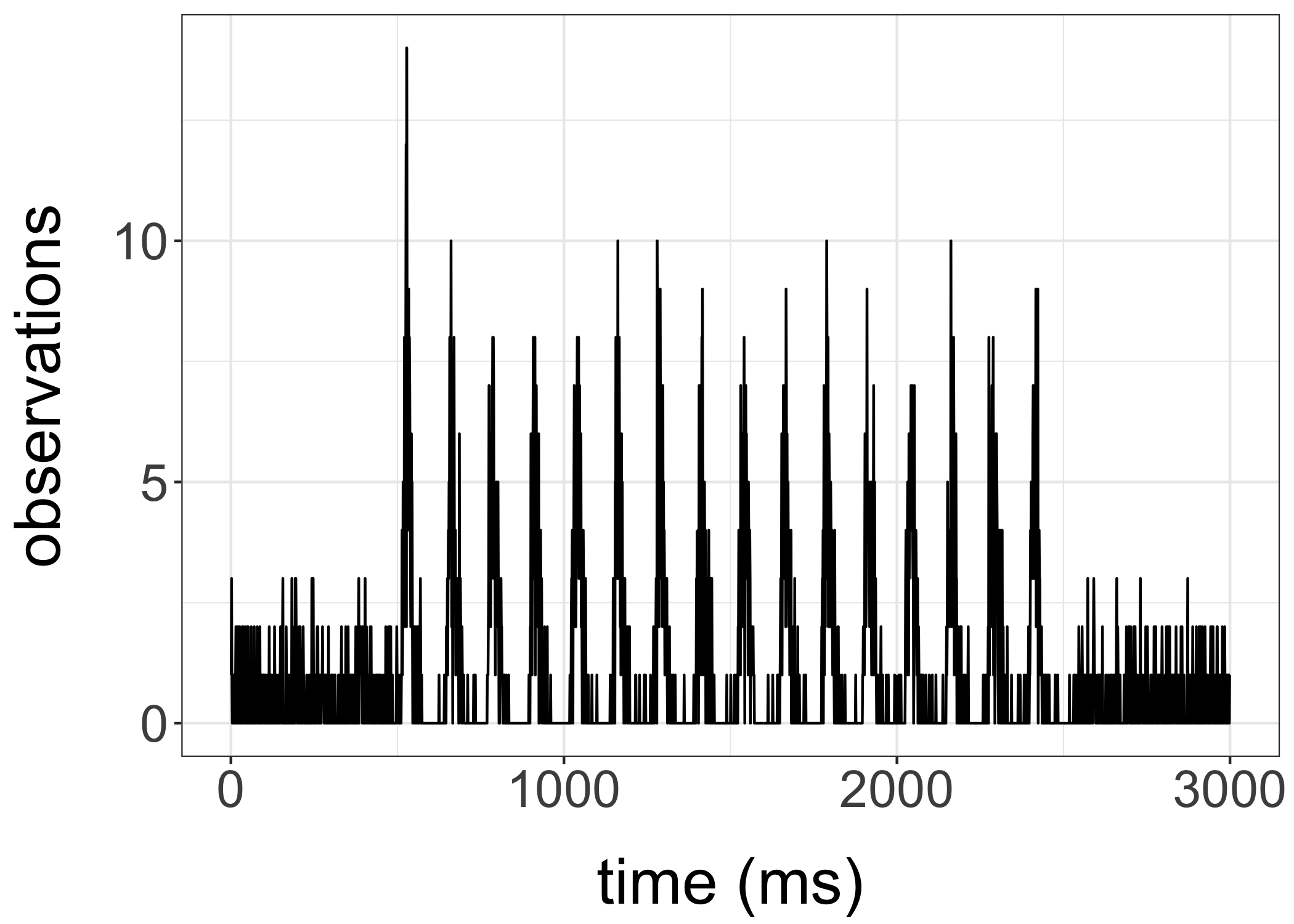}

}\hspace{1cm}\subfloat[\label{fig:neuro:likelihood}]{\includegraphics[width=0.4\textwidth]{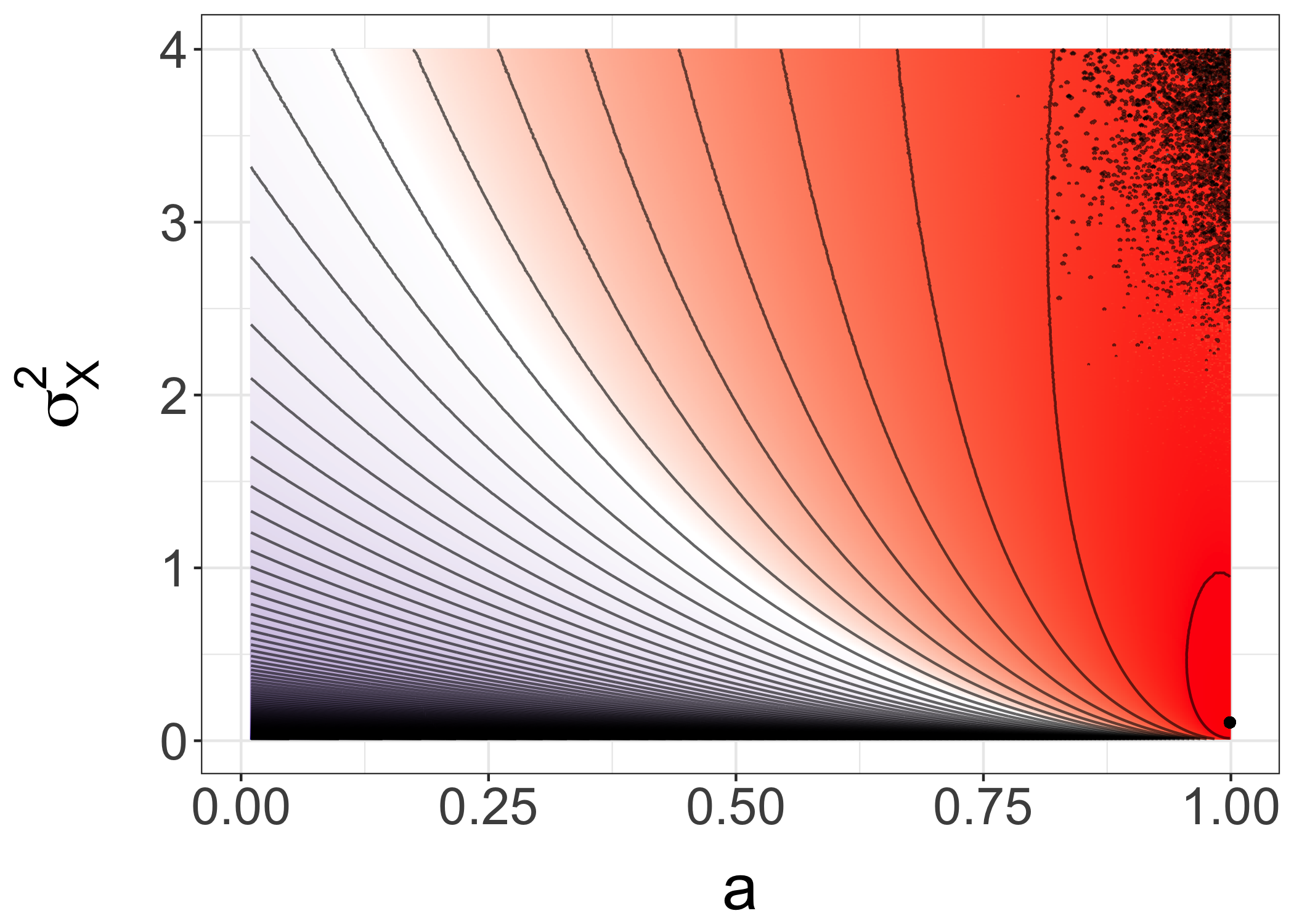}

}
\par\end{centering}
\caption{\label{fig:neuro:datalikelihood}Left: counts of neuron activation
in $50$ experiments, over a duration of three seconds. Right: estimated
log-posterior density in the neuroscience experiment of Section \ref{subsec:Neuroscience-experiment}. }
\end{figure}

\subsubsection{Standard deviation of the proposal}

Here, we initialize the chains from a uniform distribution on $[0,1]^{2}$,
and we investigate two choices of standard deviation for the random
walk proposals: the one used in \citet{heng2017controlled}, that
is $0.002$ for $a$ and $0.01$ for $\sigma_{X}^{2}$, and another
choice equal to $0.01$ for $a$ and $0.05$ for $\sigma_{X}^{2}$,
i.e. five times larger. For each choice, we can run pairs of chains
until they meet and record the meeting time; we can do so on $P$
processors in parallel (e.g. hundreds), and for a certain duration
(e.g. a few hours). Thus the number of meeting times produced by each
processor is a random variable. Following \citet{glynn1990bias},
if no meeting time was produced by a processor within the time budget,
the computation continues until one meeting time is produced, otherwise
on-going calculations are interrupted when the budget is reached.
This allows unbiased estimation of functions of the meeting time on
each processor via Corollary 7 of \citet{glynn1990bias}, and then
we can average across processors. In particular we use this strategy
to produce all histograms in the present section, as in Figure \ref{fig:neuro:meetingtimes}.

We observe that the meeting times are significatively larger when
using the smaller standard deviation (\ref{fig:meet:origsd}), with
a maximum value of $21,570$ over $1565$ realizations. With the larger
choice of standard deviation (\ref{fig:meet:largersd}), we observe
shorter meeting times, with a maximum of $928$ over $5572$ realizations.
This suggests that the values of $k$ and $m$ should be chosen very
differently in both cases.

\begin{figure}
\begin{centering}
\subfloat[\label{fig:meet:origsd}]{\includegraphics[width=0.4\textwidth]{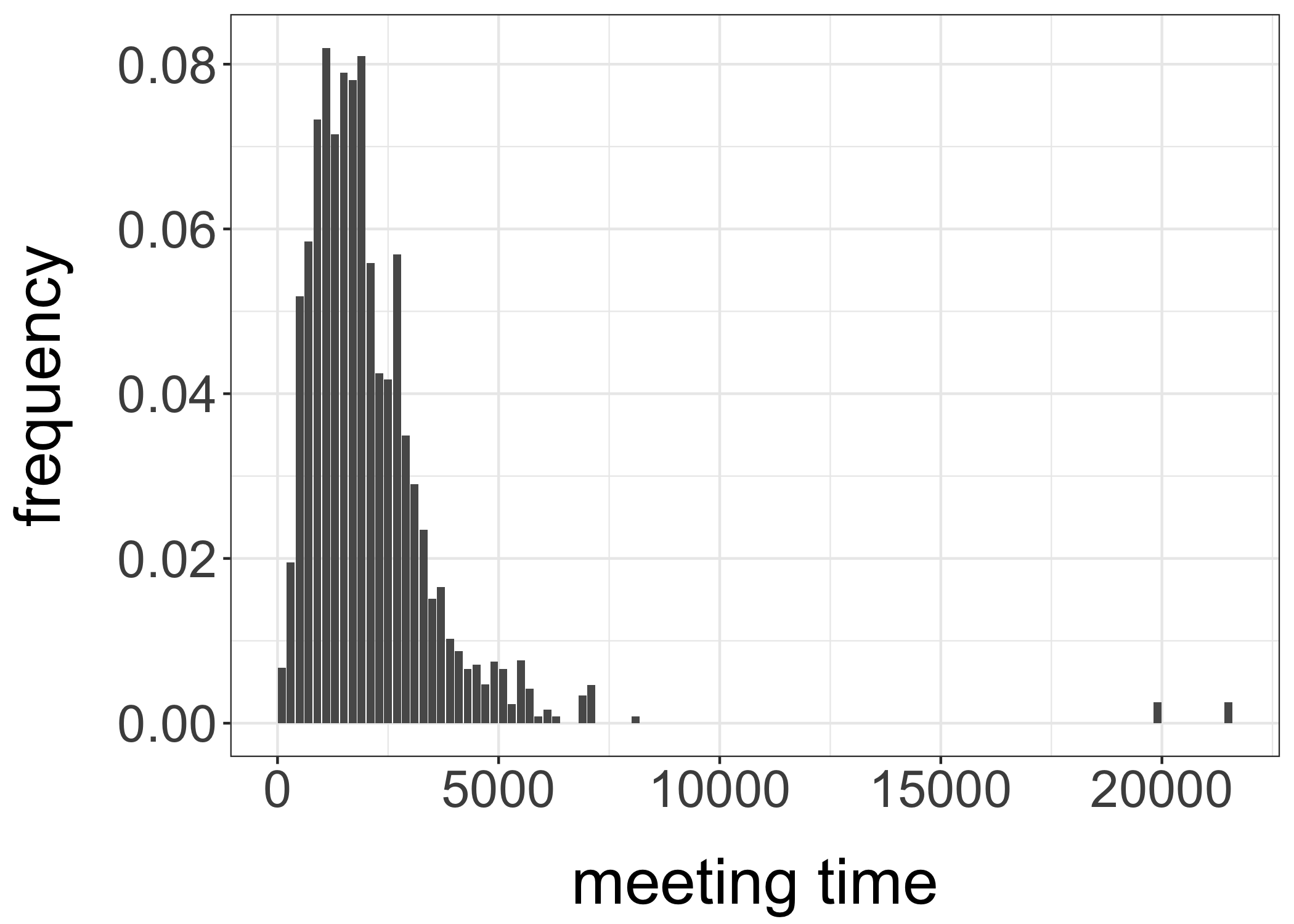}

}\hspace{1cm}\subfloat[\label{fig:meet:largersd}]{\includegraphics[width=0.4\textwidth]{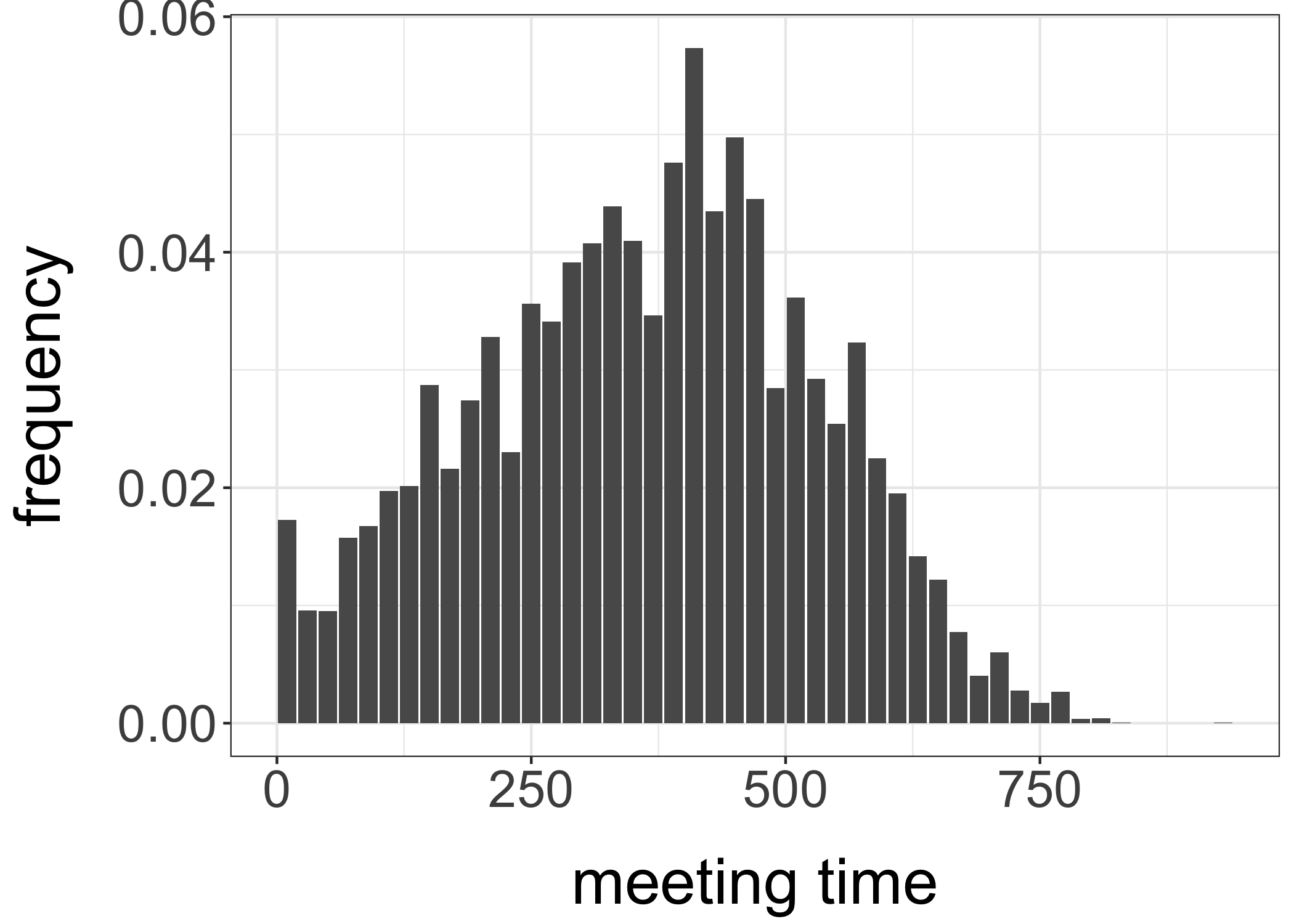}

}
\par\end{centering}
\caption{\label{fig:neuro:meetingtimes}Histograms of meeting times associated
with coupled PMMH chains, obtained with a standard deviation of the
random walk proposal of $0.002$ for $a$ and $0.01$ for $\sigma_{X}^{2}$
on the left, and with a larger standard deviation ($0.01$ on $a$
and $0.05$ on $\sigma_{X}^{2}$) on the right. In both cases, the
likelihood was estimated with controlled SMC, with $N=128$ particles,
$I=3$ iterations, in the neuroscience model of Section \ref{subsec:Neuroscience-experiment}. }
\end{figure}

To explain this difference we investigate the realization of the coupled
chains that led to the largest meeting time of $21,570$, in Figure
\ref{fig:neuro:longesttraces}. Figure \ref{fig:neuro:longesttraces:originalsd}
presents the trajectories of two chains overlaid with contours of
the target density function. The chains seem to follow approximately
the gradient of the density. Given the shape of this density, it means
that small starting values for component $a$ result in the chains
going to the region of high variance of the likelihood estimator,
in the top right corner of the plot. The marginal trace plots of one
of the two chains are shown in \ref{fig:neuro:longesttraces:traceplot}.
From the trace plots we see that most of the $21,570$ iterations
have been spent in that top right corner, where the chain got stuck,
approximately between iterations $2,000$ and $20,000$. The overall
acceptance rate is of $6\%$ for that chain, compared to $39\%$ for
the other chain shown in \ref{fig:neuro:longesttraces:originalsd}
Therefore the use of a larger proposal standard deviation seems to
have a very noticeable effect here on the ability of the Markov chain
to escape a region of high variance of the likelihood estimator.

\begin{figure}
\begin{centering}
\subfloat[\label{fig:neuro:longesttraces:originalsd}]{\includegraphics[width=0.4\textwidth]{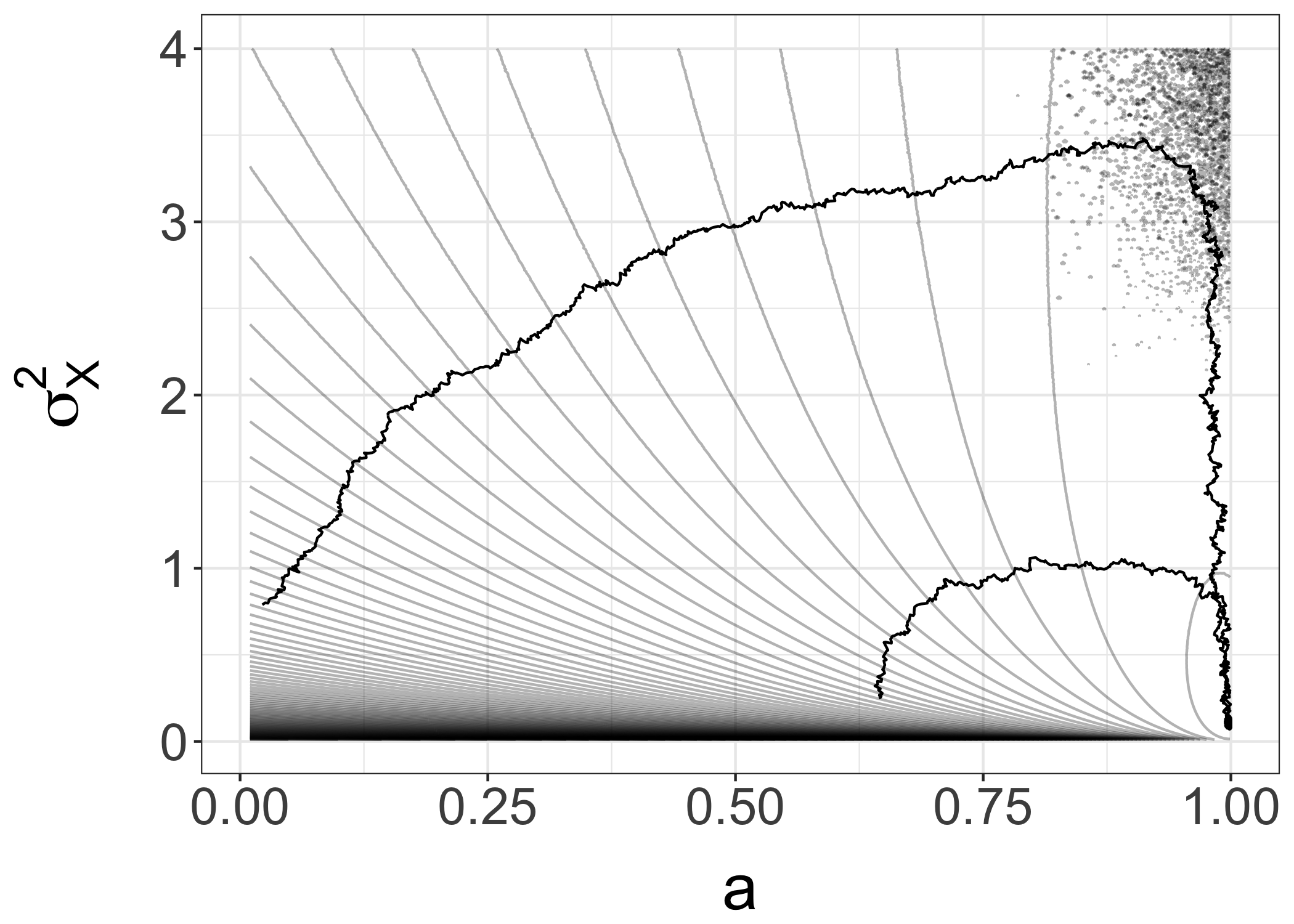}
}\hspace{1cm}\subfloat[\label{fig:neuro:longesttraces:traceplot}]{\includegraphics[width=0.4\textwidth]{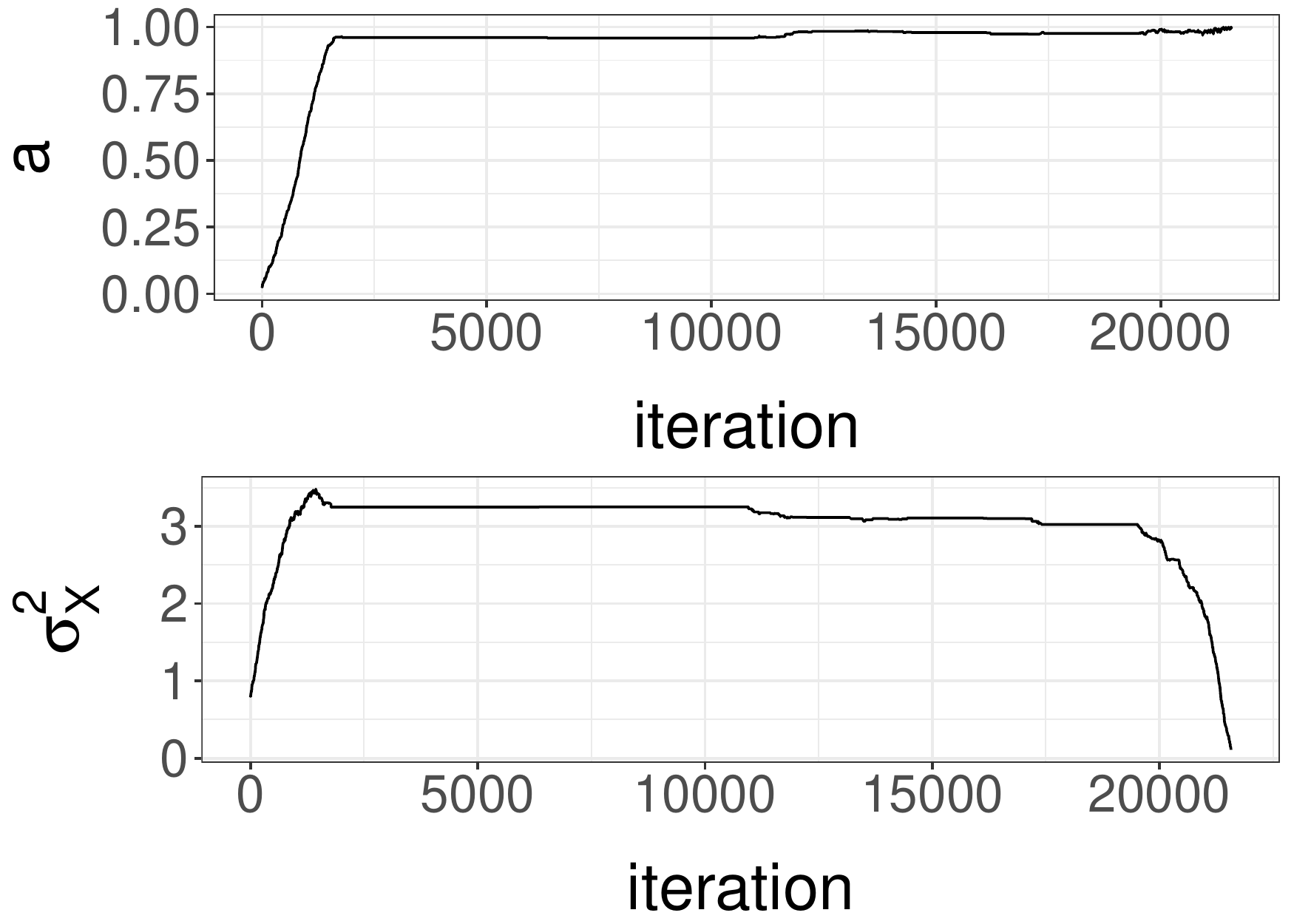}

}
\par\end{centering}
\caption{\label{fig:neuro:longesttraces}Traces of the chains corresponding
to the largest observed meeting time ($21,570$) obtained with a small
standard deviation of the random walk proposal ($0.002$ for $a$
and $0.01$ for $\sigma_{X}^{2}$), in form of a two-dimensional trajectory
on the left, and trace plots of one of the two chains on the right.
The likelihood is estimated with controlled SMC, with $N=128$ particles,
$I=3$ iterations, in the neuroscience experiment of Section \ref{subsec:Neuroscience-experiment}. }
\end{figure}

\subsubsection{Comparison with PMMH using bootstrap particle filters}

We use the larger choice of standard deviation ($0.01$ on $a$ and
$0.05$ on $\sigma_{X}^{2}$) hereafter, and compare meeting times
obtained with cSMC with those obtained with bootstrap particle filters,
with $N=4,096$ particles. This number is chosen so that the compute
times are comparable. Over 23 hours of compute time, the number of
meeting times obtained per processor varied between $4$ and $35$,
and a total of $7,776$ meeting times were obtained from $400$ processors.
The meeting times are plotted against the duration it took to produce
them in Figure \ref{fig:neuro:bpf:duration}. The compute time associated
with meeting times is not only proportional to the meeting times themselves,
but also varies across processors. This is partly due to hardware
heterogeneity across processors, and to concurrent tasks being executed
on the cluster during our experiments. The histogram in Figure \ref{fig:neuro:bpf:meeting}
shows that meeting times are larger, and heavier tailed, than when
using cSMC (see Figure \ref{fig:meet:largersd}). The maximum observed
value is $9,371$. From these plots, we see that to produce unbiased
estimators $H_{k:m}$ using BPF with a similar variance as when using
cSMC, we would have to choose larger values of $k$ and $m$,
and thus the cost per estimator would likely be higher\@.

\begin{figure}
\begin{centering}
\subfloat[\label{fig:neuro:bpf:duration}]{\includegraphics[width=0.4\textwidth]{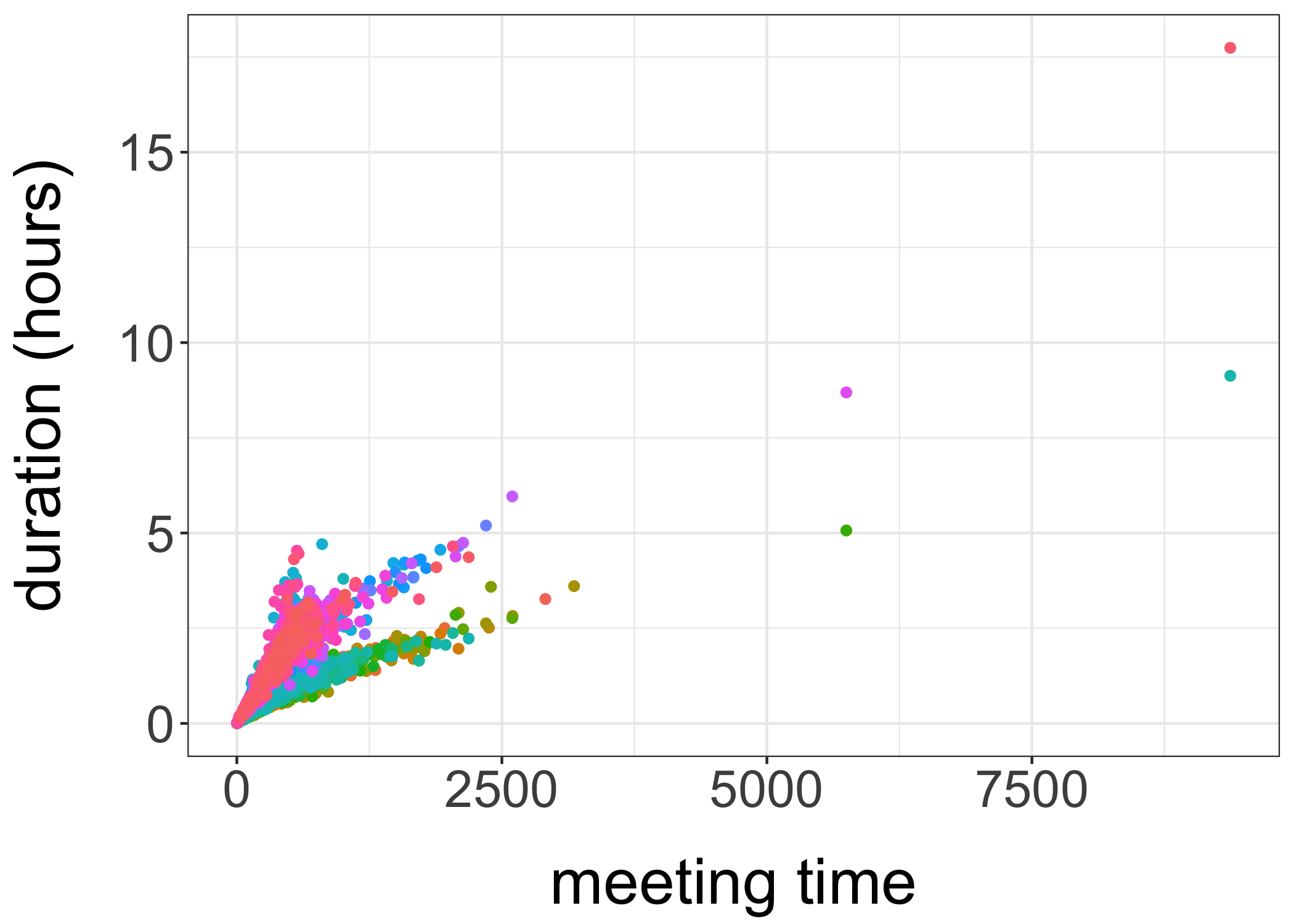}

}\hspace{1cm}\subfloat[\label{fig:neuro:bpf:meeting}]{\includegraphics[width=0.4\textwidth]{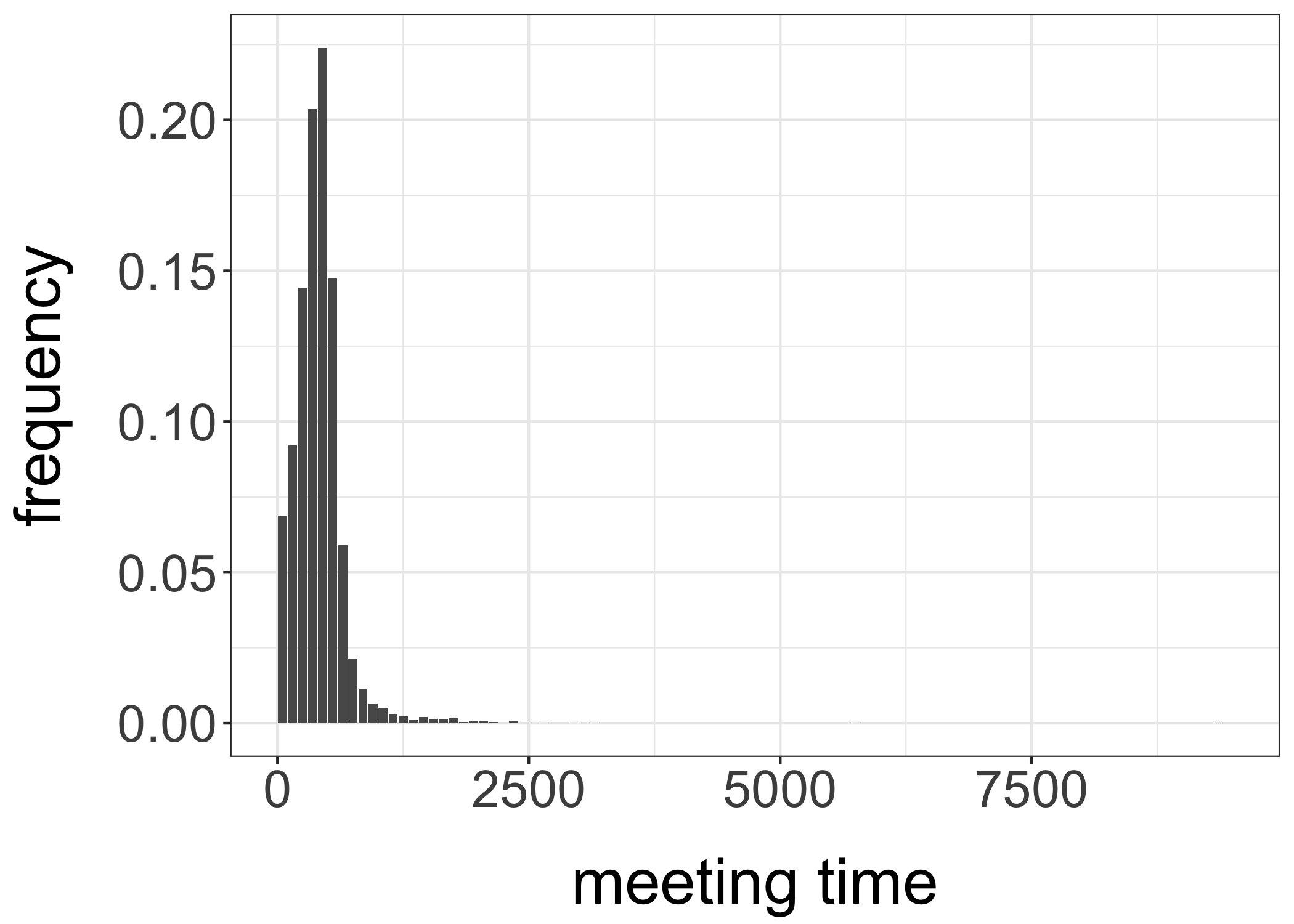}

}
\par\end{centering}
\caption{\label{fig:neuro:bpf:durationmeeting}Left: duration (in hours) versus
meeting times, using BPF with $N=4,096$ particles. Each color corresponds
to a different processor. Right: estimated histogram of the meeting
times, in the neuroscience experiment of Section \ref{subsec:Neuroscience-experiment}. }
\end{figure}

\subsubsection{Efficiency compared to the serial algorithm}

Using cSMC and the larger choice of standard deviation for the proposal,
we produce unbiased estimators $H_{k:m}$ with $k=1,000$ and $m=10,000$.
We run $100$ processors for a time budget of $23$ hours, and each
processor produced between $2$ and $9$ estimators, for a total of
$578$ estimators. The generation of samples for each processor is
represented chronologically in Figure \ref{fig:neuro:csmc:chronology}.
The variation among durations is due to the randomness of meeting
times and also to external factors such as concurrent tasks being
executed on the cluster. We produce histograms of the posterior marginals
in Figure \ref{fig:neuro:csmc:parameters}, with the result from a
long run of PMMH with cSMC ($250,000$ iterations) overlaid in red
lines.

\begin{figure}
\begin{centering}
\subfloat[\label{fig:neuro:csmc:chronology}]{\includegraphics[width=0.4\textwidth]{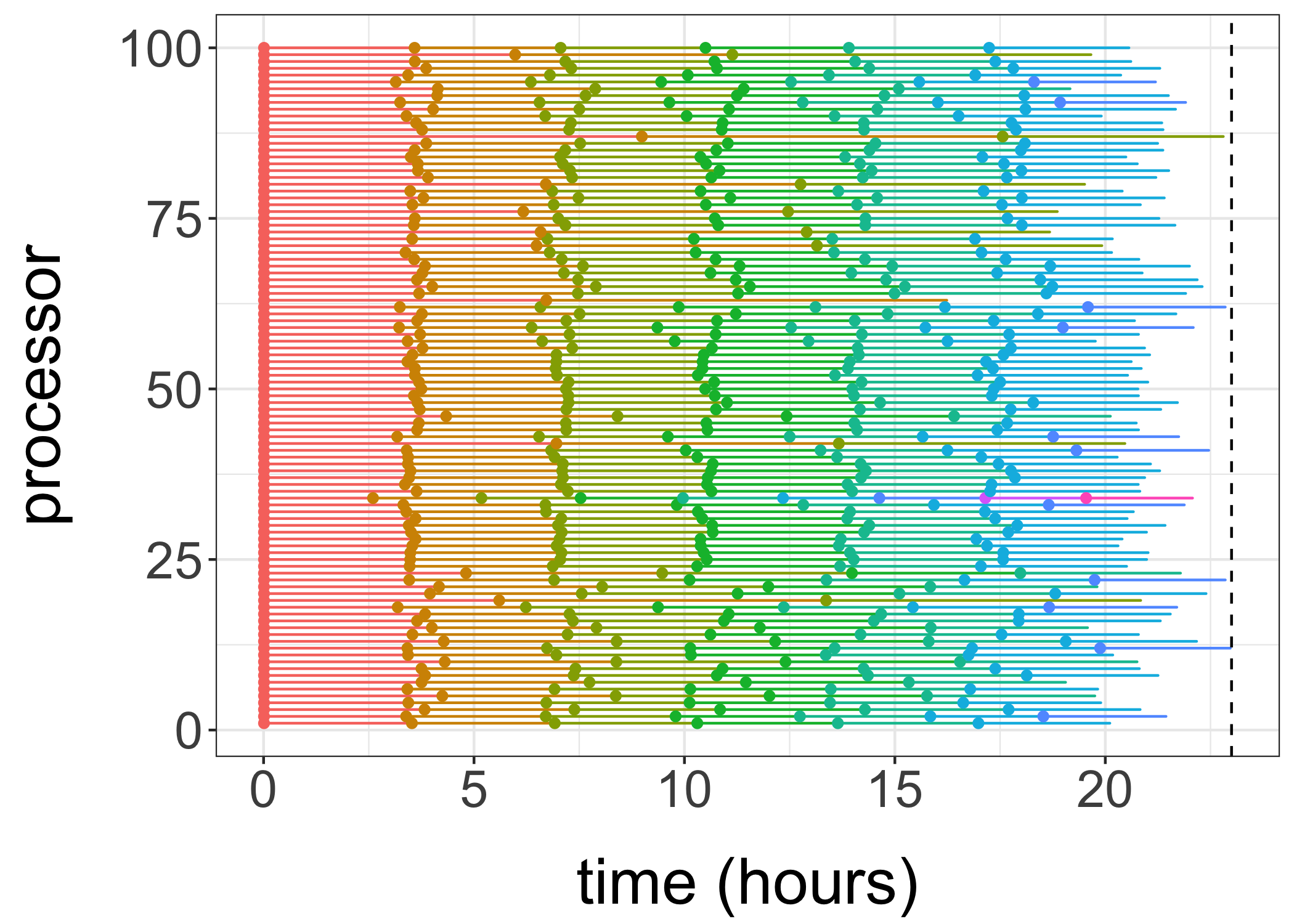}

}\hspace{1cm}\subfloat[\label{fig:neuro:csmc:parameters}]{\includegraphics[width=0.4\textwidth]{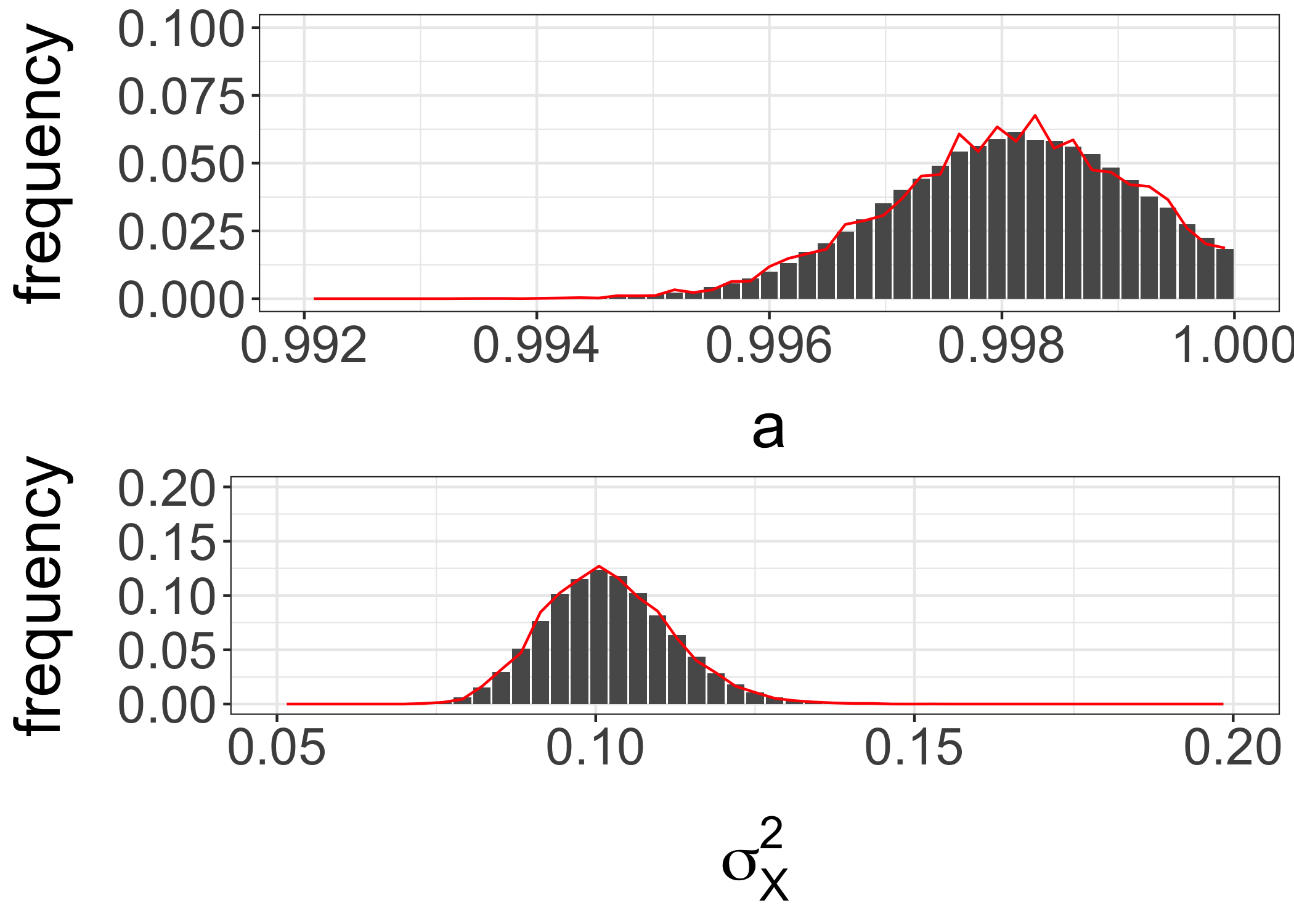}

}
\par\end{centering}
\caption{\label{fig:neuro:csmc:estimates}Left: start and end time for the
calculation of unbiased estimators on $100$ parallel processors,
for a budget of $23$ hours (dashed line), with cSMC, $N=128$ particles,
$I=3$ iterations and $k=1,000$, $m=10,000$. Right: estimated histograms
of both parameters $a$ and $\sigma_{X}^{2}$; the red lines correspond
to estimates obtained with $250,000$ iterations of PMMH and discarding
the first $10,000$ as burn-in; this is for the neuroscience experiment
of Section \ref{subsec:Neuroscience-experiment}. }
\end{figure}

We compute the loss of efficiency incurred by debiasing the PMMH chain
with the proposed estimators. We consider the test function $h:x\mapsto x_{1}+x_{2}+x_{1}^{2}+x_{2}^{2}$.
Along the PMMH chain of length $n_{\text{mcmc}}=250,000$, after a
burn-in of $n_{\text{burnin}}=10,000$ steps, and using the \texttt{spectrum0}
function of the \texttt{CODA} package, we find the asymptotic variance
associated with $h$ to be $V_{\text{as}}=7.53\cdot10^{-3}$. If we
measure computing cost in terms of MCMC iterations, we obtain an inefficiency
of $n_{\text{mcmc}}\times V_{\text{as}}/(n_{\text{mcmc}}-n_{\text{burnin}})\approx7.84\cdot10^{-3}$.
With the unbiased estimators $H_{k:m}$, if the cost of an estimator
is $2(\tau-1)+\max(1,m+1-\tau)$, then the average cost per processor
is $59,860$. The empirical variance of the unbiased estimators obtained
per processor is equal to $1.4\cdot10^{-7}$, thus we obtain an inefficiency
of $59860\times1.4\cdot10^{-7}\approx8.4\cdot10^{-3}$. This inefficiency
is slightly above $7.84\cdot10^{-3}$.

Next, we parameterize cost in terms of time (in seconds) instead of
number of MCMC steps. This accounts for the fact that running jobs
on a cluster involve heterogeneous hardware and concurrent tasks.
The serial PMMH algorithm was run on a desktop computer for $169,952$
seconds and thus the inefficiency might be measured as $169,952\times V_{\text{as}}/(n_{\text{mcmc}}-n_{\text{burnin}})\approx5.3\cdot10^{-3}$.
Note that each iteration took less than a second on average, because
parameter values proposed outside of the support of the prior were
rejected before running a particle filter; on the other hand the cost
of a cSMC run is above one second on average. For the proposed estimators,
the budget was set to $23$ hours and we obtained a variance across
processors of $1.4\cdot10^{-7}$; thus we can compute the inefficiency
as $1.16\cdot10^{-2}$, which is approximately twice the inefficiency
of the serial algorithm.

\section{Methodological extensions}

\label{sec:methodext}

The following provides two further examples of coupled MCMC algorithms
to perform inference when the likelihood function is intractable.
The associated estimators are not covered by our theoretical results.

\subsection{Block pseudo-marginal method}

Block pseudo-marginal methods have demonstrated significant computational
savings for Bayesian inference for random effects models over standard
pseudo-marginal methods \citep{tran2016block}. Such methods proceed
through introducing strong positive correlation between the current
likelihood estimate $\hat{p}(y|\theta)$ and the likelihood estimate
of the proposed parameter $\hat{p}(y|\theta')$ through only modifying
a subset of the auxiliary variables used to obtain the likelihood
estimate at each iteration. We here demonstrate the computational
benefits of such a scheme in obtaining unbiased estimators of posterior
expectations.

We focus here on random effects models, as defined in section \ref{sec:randomeffects}.
We recall that the likelihood estimate is given by $\hat{p}(y|\theta,U)=\prod_{t=1}^{T}\hat{p}(y_{t}|\theta,U_{t})$,
where $\{\hat{p}(y_{t}|\theta,U{}_{t})\}_{t=1,...,T}$ are $T$ independent
non-negative unbiased likelihood estimates of $\{p(y_{t}|\theta)\}_{t=1,...,T}$
when $U_{t}\sim m_{t}(\cdot)$. In the following, we provide a minor modification of the blocking
strategy proposed in \citet{tran2016block}, where instead of jointly
proposing a new parameter and a single block of auxiliary random variables,
a parameter update is performed, followed by sequentially iterating
through the auxiliary random variables used to construct the likelihood
estimate of observation $t$. For each data $t$, new values are proposed
according to $U_{t}'\sim m_{t}(\cdot)$ and accepted with probability
\begin{equation}
\alpha_{\mathrm{BPM},t}\left\{ \widehat{p}(y_{t}\mid\theta,U_{t}),\hat{p}(y_{t}|\theta,U'_{t})\right\} =\min\left\{ 1,\frac{\hat{p}(y_{t}|\theta,U'_{t})}{\hat{p}(y_{t}|\theta,U{}_{t})}\right\} .\label{eq:pmtaccept}
\end{equation}
As remarked in \citet{tran2016block}, such blocking strategies are
generally not applicable to particle filter inference in state space
models, whereby likelihood estimates for observation $t$ typically
depend on all auxiliary random variables generated up to and including
$t$. We provide pseudo-code for the proposed blocking strategy in
Algorithm \ref{alg:tblockedpm}. We denote by $U_{t,n}$ the set of
auxiliary variables $U_{t}$ at iteration $n$,

\begin{algorithm}[h]
\caption{\textbf{Sampling from the block pseudo-marginal kernel given $\left(\theta_{n-1},(U_{t,n-1})_{_{t\protect\geq1}}\right)$
}\label{alg:tblockedpm}}

\begin{enumerate}
\item \textsf{Sample $\theta'\sim q\left(\cdot|\theta_{n-1}\right)$ and
compute $\widehat{p}(y_{t}\mid\theta',U_{t,n-1})$ for $t=1,...,T$.}
\item \textsf{With probability $\alpha_{\mathrm{PM}}\left\{ (\theta_{n-1},\prod_{t=1}^{T}\widehat{p}(y_{t}\mid\theta_{n-1},U_{t,n-1})),(\theta',\prod_{t=1}^{T}\widehat{p}(y_{t}\mid\theta',U_{t,n-1}))\right\} $,
set $\theta_{n}=\theta'$}. \textsf{Otherwise, set }$\theta_{n}=\theta_{n-1}$.
\item \textsf{For $t=1,...,T$}
\begin{enumerate}
\item \textsf{Sample $U'_{t}\sim m{}_{t}\left(\cdot\right)$. }
\item \textsf{With probability $\alpha_{\mathrm{BPM,t}}\left\{ \widehat{p}(y_{t}\mid\theta_{n},U_{t,n-1}),\widehat{p}(y_{t}\mid\theta_{n},U'_{t})\right\} $,
set $U_{t,n}=U'_{t}$. Otherwise, set $U_{t,n}=U_{t,n-1}$.}
\end{enumerate}
\end{enumerate}
\end{algorithm}

\subsubsection{Coupled block pseudo-marginal method}
An algorithm to couple two block pseudo-marginal algorithms to construct unbiased estimators $H_{k:m}$ is provided in Algorithm \ref{alg:coupledblock}.
Denoting the two states of the chains at step $n\geq1$ by $(\theta_{n},(U_{t,n})_{_{t\geq1}})$
and $(\tilde{\theta}{}_{n-1},(\widetilde{U}_{t,n-1})_{_{t\geq1}})$,
Algorithm \ref{alg:coupledblock} describes how to obtain $(\theta_{n+1},(U_{t,n+1})_{_{t\geq1}})$
and $(\tilde{\theta}{}_{n},(\widetilde{U}_{t,n})_{_{t\geq1}})$; thus
it describes a kernel $\bar{P}$.

\begin{algorithm}[h]
\caption{\textbf{Sampling from the coupled block pseudo-marginal kernel given
$(\theta_{n},(U_{t,n})_{_{t\protect\geq1}},\tilde{\theta}{}_{n-1},(\widetilde{U}_{t,n-1})_{_{t\protect\geq1}})$}\label{alg:coupledblock}}

\begin{enumerate}
\item \textsf{Sample $(\theta',\tilde{\theta}')$ from the maximal coupling
of $q\left(\cdot|\theta_{n}\right)$ and $q(\cdot|\tilde{\theta}{}_{n-1})$.}
\item \textsf{Compute $\widehat{p}(y_{t}\mid\theta',U_{t,n})$ and $\widehat{p}(y_{t}\mid\tilde{\theta}',\widetilde{U}_{t,n-1})$
for $t=1,...,T$.}
\item \textsf{Sample $\mathit{\mathfrak{u}}\sim\mathcal{U}\left[0,1\right]$.}
\item \textsf{If $\mathit{\mathfrak{u}}<\alpha_{\mathrm{PM}}\left\{ \left(\theta_{n},\prod_{t=1}^{T}\widehat{p}(y_{t}\mid\theta_{n},U_{t,n})\right),(\theta',\prod_{t=1}^{T}\widehat{p}(y_{t}\mid\theta',U_{t,n}))\right\} $
then set $\theta_{n+1}=\theta'$. Otherwise, set }$\theta_{n+1}=\theta_{n}$.
\item \textsf{If $\mathit{\mathfrak{u}}<\alpha_{\mathrm{PM}}\left\{ (\tilde{\theta}_{n-1},\prod_{t=1}^{T}\widehat{p}(y_{t}\mid\tilde{\theta}{}_{n-1},\widetilde{U}_{t,n-1})),(\tilde{\theta}',\prod_{t=1}^{T}\hat{p}(y_{t}|\tilde{\theta}',\widetilde{U}_{t,n-1}))\right\} $
then set $\tilde{\theta}_{n}=\tilde{\theta}'$. Otherwise, set }$\tilde{\theta}_{n}=\tilde{\theta}_{n-1}$\textsf{. }
\item \textsf{For $t=1,...,T$}
\begin{enumerate}
\item \textsf{Sample $U'_{t}\sim m_{t}(\cdot)$. }
\item \textsf{Sample $\mathit{\mathfrak{u}}\sim\mathcal{U}\left[0,1\right]$.}
\item \textsf{If $\mathit{\mathfrak{u}}<\alpha_{\mathrm{BPM,t}}\left\{ \widehat{p}(y_{t}\mid\theta_{n+1},U_{t,n}),\widehat{p}(y_{t}\mid\theta_{n+1},U'_{t})\right\} $
then set $U_{t,n+1}=U'_{t}$. Otherwise, set $U_{t,n+1}=U{}_{t,n}$.}
\item \textsf{If $\mathit{\mathfrak{u}}<\alpha_{\mathrm{BPM,t}}\left\{ \widehat{p}(y_{t}\mid\tilde{\theta}_{n},\tilde{U}_{t,n-1}),\widehat{p}(y_{t}\mid\tilde{\theta}_{n},U'_{t})\right\} $
then set $\widetilde{U}_{t,n}=U'_{t}$. Otherwise, set $\widetilde{U}_{t,n}=\widetilde{U}_{t,n-1}$.}
\end{enumerate}
\end{enumerate}
\end{algorithm}

\subsubsection{Bayesian multivariate probit regression}

The following demonstrates the proposed algorithm for a latent variable
model applied to polling data and explores the possible gains when compared
to the unbiased estimators obtained using the coupled pseudo-marginal
algorithm. The data consists of polling data collected between February
2014 and June 2017 as part of the British Election Study \citep{https://doi.org/10.15127/1.293723}.
We use a multivariate probit model, which for $i\in\{1,..,T\}$ and
$j\in\{1,2,3\}$ can be expressed as $X_{ij}=\beta'\zeta_{ij}+\epsilon_{ij}$
and $Y_{ij}=\mathds{1}[X_{ij}>0]$ for observed binary response $Y_{ij}$,
latent state $X_{ij}$ and where $i$ indexes the $i^{th}$ participant,
$j$ indexes the $j^{th}$ wave of questions, $\beta$ is a vector
of regression coefficients (including an intercept) and $\zeta_{ij}$
is a vector of independent variables.

We use a random sample of $T=2,000$ participants over three waves
(one a year) in the run up to the United Kingdom's European Union
membership referendum on $23^{rd}$ June 2016, regressing the binary
outcome asking participants how they would vote in an EU referendum
against how they perceive the general economic situation in the UK
has changed over the previous 12 months (graded 1-5, with 1=`Got a
lot worse', 5=`Got a lot better'). A detailed description of the data
is provided in Appendix \ref{subsec:surveydescription}.

We allow for correlations between waves through modelling the perturbations
$(\epsilon_{i1},\epsilon_{i2},\epsilon_{i3})\sim\mathcal{N}(0,\Sigma_{\rho})$
with a generic correlation matrix $\Sigma_{\rho}$. In total, we have
five unknown parameters $\theta=(\beta_{1},\beta_{2},\rho_{2,1},\rho_{3,1},\rho_{3,2})$,
with $\beta_{1}$ denoting a regressor coefficient, $\beta_{2}$ a
constant offset and $\rho_{s,t}$ element $(s,t)$ of $\Sigma_{\rho}$.
We place independent priors on each parameter with $\beta_{1},\beta_{2}\stackrel{i.i.d.}{\sim}\mathcal{N}(0,10^{2})$ and $\rho_{s,t}\stackrel{i.i.d.}{\sim}\mathcal{U}[-1,1]$, where we additionally truncate the prior on $\Sigma_{\rho}$ to ensure support only on the manifold of positive definite matrices.

\paragraph{Inference}

For each observation $y_{i}:=(y_{i1},y_{i2},y_{i3})$, we obtain unbiased
estimates of the likelihood of $\theta$ using the sequential importance
sampling algorithm of Geweke, Hajivassiliou and Keane; see, e.g., \citet[5.6.3]{train2009discrete}
and references therein. We set the initial distribution $\pi_{0}=\mathcal{N}(\hat{\mu},0.01^{2}I)$,
supported only on areas of positive mass under the prior (we employ a simple rejection sampling algorithm to sample $\Sigma_{\rho}$ initially) and use a Normal
random walk proposal with covariance set to $\frac{2.38^{2}}{5}\hat{\Sigma}$, see \cite{roberts1997weak}, following
where $\hat{\mu}$ and $\hat{\Sigma}$ are an empirical estimate of
the posterior mean and covariance on a preliminary run of 10,000 iterations
of block pseudo-marginal with $N=40$ discarding the first 10\% as
burn-in.

We compare coupled block pseudo-marginal with coupled pseudo-marginal.
We examine values of $N$ for the latter that are close to the optimal
value of $N$ for the serial algorithm, estimated through ensuring
the variance of the log-likelihood estimates is between 1 and 2, as
per the guidance in \citet{doucet2015biometrika}. In this case we
consider $N\in\{600,700,800\}$, providing a corresponding variance
of the log-likelihood estimates given by $\{1.67,1.40,1.25\}$ (estimated
using 10,000 likelihood estimates at $\hat{\mu}$). For the block
pseudo-marginal algorithm we consider $N\in\{5,10,20\}$.

\paragraph{Meeting times}

Both algorithms were run continuously until coupling for half an hour
each on a 48 CPU Intel Xeon 2.4Ghz E5-4657L server, with the number
of estimators produced for block pseudo-marginal varying between 2,000
and 4,000 for the values of $N$ considered and between 160 and 180
for the pseudo-marginal. The meeting times are plotted in Figure \ref{fig:blockpmmt},
where it can be seen that despite the lower cost of the block pseudo-marginal
algorithm the absolute values of the meeting times are comparable
across algorithms.

Accordingly, we also plot the distribution of meeting times accounting
for the cost of running each algorithm, i.e. $N\tau$ for the pseudo-marginal
algorithm and $2N\tau$ for the block pseudo-marginal algorithm. The
additional factor of $2$ for the latter can be seen as an upper bound
on the additional computational cost of the block pseudo-marginal
algorithm, assuming twice the density evaluations per complete iteration
and less than twice the number of pseudo-random numbers generated.
Figure \ref{fig:blockpmcw} shows the results of this additional cost-weighting
where it can be seen that meeting times are between 1 and 2 orders
of magnitude larger for the pseudo-marginal over the block pseudo-marginal
algorithm.

\begin{figure}
\begin{centering}
\subfloat[]{\begin{centering}
\includegraphics[width=0.4\textwidth]{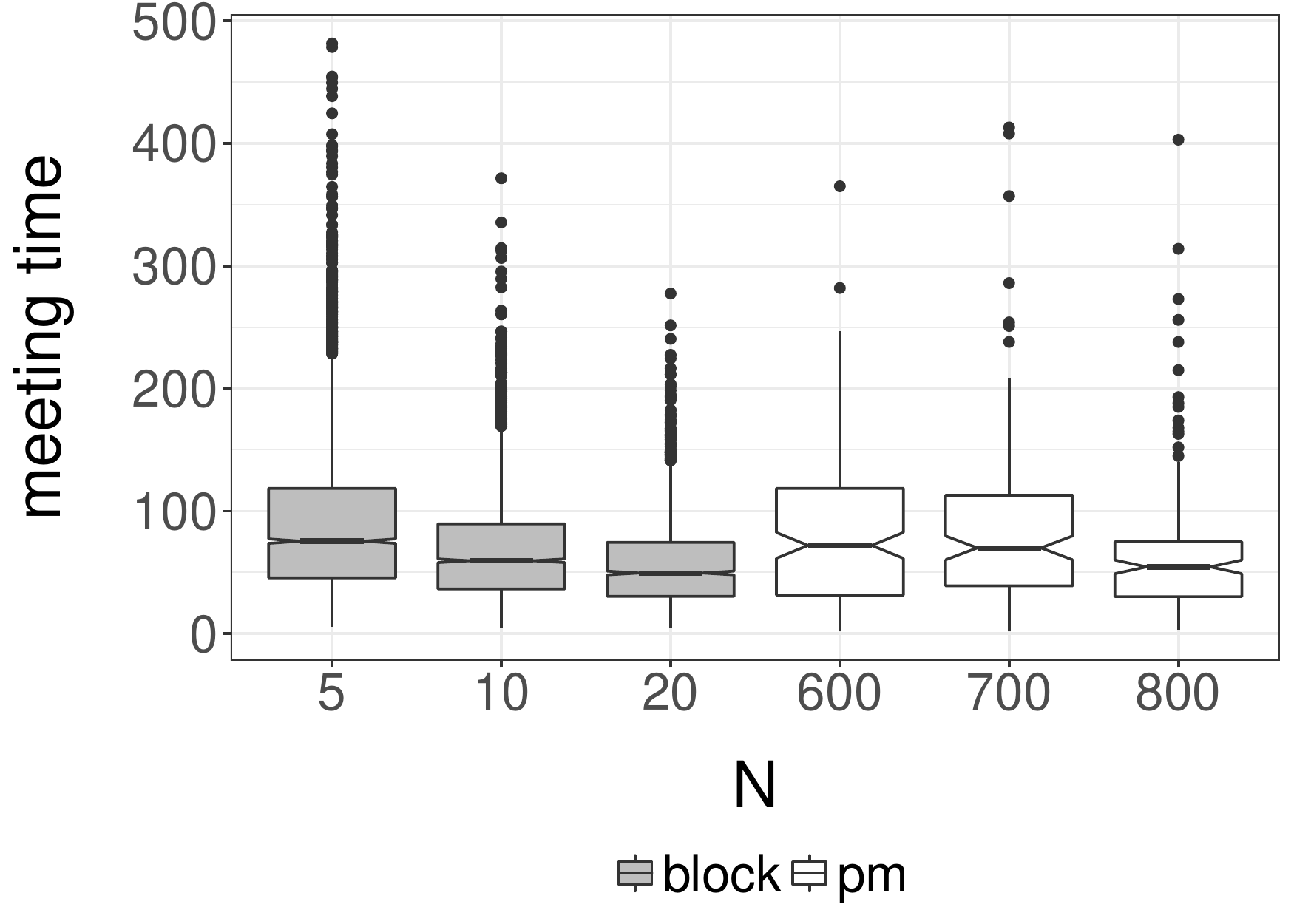}
\par\end{centering}
\label{fig:blockpmmt}}\subfloat[]{\begin{centering}
\includegraphics[width=0.4\textwidth]{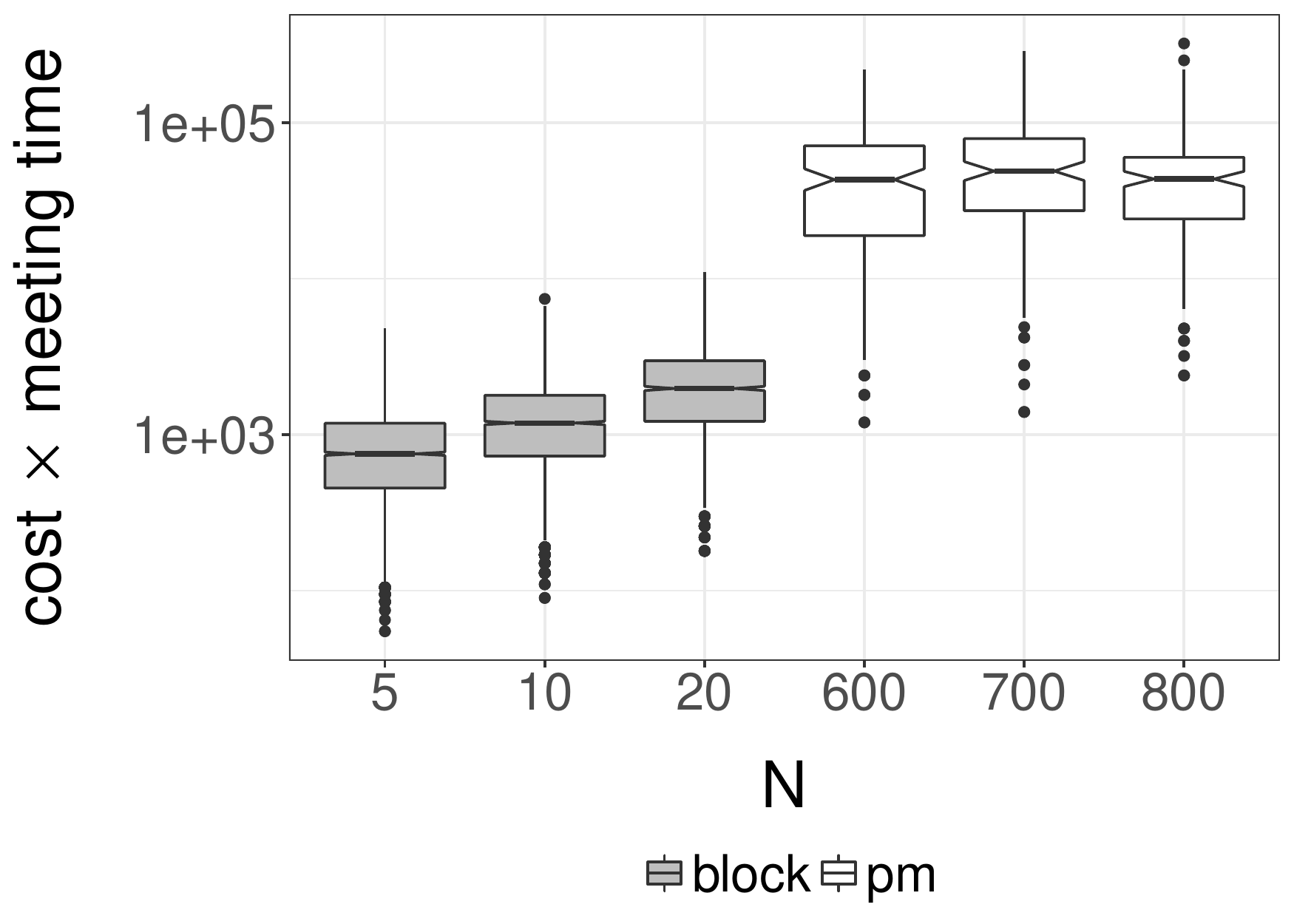}
\par\end{centering}
\label{fig:blockpmcw}}
\par\end{centering}
\caption{Meeting times for coupled block pseudo-marginal and coupled pseudo-marginal
algorithms. Left: raw meeting times for the two algorithms. Right: meeting times weighted by cost for the two algorithms, i.e. $\tau N$ for coupled pseudo-marginal and $2\tau N$ for coupled block pseudo-marginal.}
\end{figure}

\paragraph{Variance of estimators}

We estimate the increase in inefficiency of the coupled over the serial
algorithm for $N=10$, $k=500$ and $m=5,000$; the choice of $k$
is guided by the meeting times in Figure \ref{fig:blockpmmt}. Running
coupled block pseudo-marginal 200 times, we estimate the variance
using the test function $h:x\mapsto\sum_{i}(x_{i}+x_{i}^{2})$ to
be $1.05\cdot10^{-5}$. Estimating the cost of $2(\tau-1)+\max(1,m+1-\tau)$
to be 5121, implies an inefficiency of $5.36\cdot10^{-2}$. In comparison,
we estimate the inefficiency of the serial algorithm using \texttt{spectrum0.ar}
as before on runs of length 125,000 (discarding 10\% as burn-in and
averaging over 20 estimators) to be $n_{\text{mcmc}}\times V_{as}/(n_{\text{mcmc}}-n_{\text{burnin}})=4.82\cdot10^{-2}$
suggesting an increase in inefficiency of $11\%$ for the unbiased
estimators.

Finally, we compare the inefficiency of unbiased estimators generated with coupled block pseudo-marginal kernels
with those produced using standard coupled pseudo-marginal kernels with $N_{\text{PM}}=700$ particles.
For coupled pseudo-marginal, the variance of the unbiased estimator was estimated
to be $1.53\cdot10^{-5}$ and the expected cost was estimated to be
$5147$, implying an inefficiency of $7.86\cdot10^{-2}$. As a result
we estimate the improvement of inefficiency for the coupled block
pseudo-marginal by $\frac{N_{\text{PM}}}{2N}\times\frac{7.86\cdot10^{-2}}{5.36\cdot10^{-2}}$
to be approximately 51 times. Estimation of the asymptotic variance
of the serial pseudo-marginal algorithm was computationally infeasible
for this many particles, with a single iteration taking on average
six seconds on the aforementioned server, hence the choice of $N$
motivated by the guidance in \citet{doucet2015biometrika} instead.

\subsection{Exchange algorithm}

Problems where the likelihood function is only known only up to a
constant of proportionality occur frequently across Bayesian statistics;
see, e.g., \citet{park2018bayesian} for a recent account of current
methodology and applications. In this case, posterior distributions
$\pi(\theta)\propto p(y|\theta)p\left(\theta\right)$ are given by
\[
p(y|\theta)=\frac{f(y|\theta)}{\mathcal{Z}(\theta)},\quad\mathcal{Z}(\theta):=\int f(y|\theta)dy,
\]
where $f(y|\theta)$ can be evaluated pointwise but its parameter-dependent
normalizing constant $\mathcal{Z}(\theta)$ is intractable. %
This is a scenario common for undirected graphical models and spatial
point processes \citep{moller2006efficient,murray2006mcmc}. The exchange
method detailed in Algorithm \ref{alg:exchange} is an MCMC scheme
proposed by \citet{murray2006mcmc} to sample such distributions under
the assumption that, although $\mathcal{Z}(\theta)$ cannot be evaluated,
it is possible to simulate exactly artificial observations from $p(y|\theta)$.
This is indeed possible for a large class of spatial point processes
as well as the Ising and Potts models using perfect simulation procedures.

\begin{algorithm}[H]
\caption{\textbf{Sampling from the Exchange kernel given }$\theta_{n-1}$\label{alg:exchange}}

\begin{enumerate}
\item \textsf{Sample $\theta'\sim q\left(\cdot|\theta_{n-1}\right)$ and
}$Y'\sim p(\cdot|\theta')$.
\item \textsf{With probability }
\begin{equation}
\alpha_{\mathrm{EX}}\left(\theta_{n-1},\theta',Y'\right):=\min\left\{ 1,\frac{f(y|\theta')p(\theta')f(Y'|\theta_{n-1})q\left(\theta_{n-1}|\theta'\right)}{f(y|\theta_{n-1})p(\theta_{n-1})f(Y'|\theta')q(\theta'|\theta_{n-1})}\right\} ,\label{eq:Exchangeacceptproba}
\end{equation}
\textsf{set $\theta_{n}=\theta'.$ Otherwise, set }$\theta_{n}=\theta_{n-1}$.
\end{enumerate}
\end{algorithm}

\subsubsection{Coupled exchange algorithm}

An algorithm to couple two block pseudo-marginal algorithms to construct unbiased estimators $H_{k:m}$ is provided in Algorithm \ref{alg:exchange-coupled}.
Denoting the two states of the chains at step $n\geq1$ by $\theta_{n}$
and $\tilde{\theta}{}_{n-1}$, Algorithm \ref{alg:coupledPMMH} describes
how to obtain $\theta_{n+1}$ and $\tilde{\theta}{}_{n}$; thus it
describes a kernel $\bar{P}$.

\begin{algorithm}[H]
\caption{\textbf{Sampling from the coupled Exchange kernel given $(\theta_{n},\tilde{\theta}{}_{n-1})$}\label{alg:exchange-coupled}}

\begin{enumerate}
\item \textsf{Sample $\theta'$ and $\tilde{\theta}'$ from the maximal
coupling of $q\left(\cdot|\theta_{n}\right)$ and }$q(\cdot|\tilde{\theta}{}_{n-1})$.
\item \textsf{If the proposals couple, i.e. if $\theta'=\tilde{\theta}',$
then sample }$Y'\sim p(\cdot|\theta')$ \textsf{and set }$\tilde{Y}'=Y'$.
\item \textsf{If the proposals do not couple, sample $Y'\sim p(\cdot|\theta')$
and }$\tilde{Y}'\sim p(\cdot|\tilde{\theta}')$.
\item \textsf{Sample $\mathit{\mathfrak{u}}\sim\mathcal{U}\left[0,1\right]$.}
\item \textsf{If $\mathit{\mathfrak{u}}<\alpha_{\mathrm{EX}}\left(\theta_{n},\theta',Y'\right)$
then set $\theta_{n+1}=\theta'.$ Otherwise, set }$\theta_{n+1}=\theta_{n}.$
\item \textsf{If $\mathit{\mathfrak{u}}<\alpha_{\mathrm{EX}}(\tilde{\theta}_{n-1},\tilde{\theta}',\tilde{Y}')$
then set $\tilde{\theta}_{n}=\tilde{\theta}'.$ Otherwise, set }$\tilde{\theta}_{n}=\tilde{\theta}_{n-1}.$
\end{enumerate}
\end{algorithm}

\subsubsection{High temperature Ising model}

We examine the proposed algorithm for inference in a planar lattice
Ising model without an external field. The model comprises observations
$y_{i}\in\{-1,+1\}$ on a $L\times L$ square lattice such that $p(y|\theta)\propto\exp\left(\beta\sum_{i\sim j}y_{i}y_{j}\right)$
where $i\sim j$ denotes the neighbours $j$ of node $i$ and $\theta=\beta$
denotes the inverse temperature. We restrict interest to high temperature
models specifying a prior distribution $\beta\sim\mathcal{U}[0,\beta_{c}]$,
with $\beta_{c}=\frac{1}{2}\log(1+\sqrt{2})$ denoting the critical
temperature of the Ising model on the infinite lattice \citep{ullrich2013exact,onsager1944crystal}.
Here, perfect simulation can be performed using coupling from the
past techniques with simple heat bath dynamics developed by \citet{propp1996exact}.
We generate observations for $L=80$ and set the proposal covariance
to $10^{-4}I$, initialising the chains from the prior.

We obtain estimates of the distribution of meeting times using 1,000
repetitions of coupled exchange, with the results shown in Figure
\ref{fig:isingmt-1}. Based on this, we obtain unbiased estimates
of the expectation of $\beta$ under the posterior distribution\textbf{
}using $k=100$ and $m=10k$\textbf{ }over 1,000 repetitions. It is
noted that the clock time to obtain a single estimator (on the same
machine) varies significantly due to the variable computational cost
of performing coupling from the past, depending on $\theta$. We plot
a histogram of the clock times to obtain each $H_{k:m}$ in Figure
\ref{fig:isingclock}.

Based on the heterogeneity of times to produce a single unbiased estimator,
we compare the serial inefficiency with the inefficiency of coupled
exchange\textbf{ }based on the clock time to obtain a certain variance,
with the test function $h:x\mapsto x$. We estimate the asymptotic
variance with $n_{\text{mcmc}}=200,000$ iterations of the original
algorithm (discarding the first 10\% as burn-in, and using \texttt{spectrum0.ar}
as before) to be $V_{as}\approx4.22\cdot10^{-4}$, and the algorithm
taking in total $41,095$ seconds. As a result, we estimate the serial
inefficiency in terms of clock-time to be $41,095\times V_{\text{as}}/(n_{\text{mcmc}}-n_{\text{burnin}})\approx9.6\cdot10^{-5}$.
Comparatively, the mean time to return a single estimator $H_{k:m}$
was estimated to be $546$ seconds, with the variance of a single
$H_{k:m}$ estimated to be $4.78\cdot10^{-7}$ providing an estimated
inefficiency of $2.6\cdot10^{-4}$, implying a three-fold increase in inefficiency.

\begin{figure}
\centering{}\subfloat[]{\centering{}\includegraphics[width=0.4\textwidth]{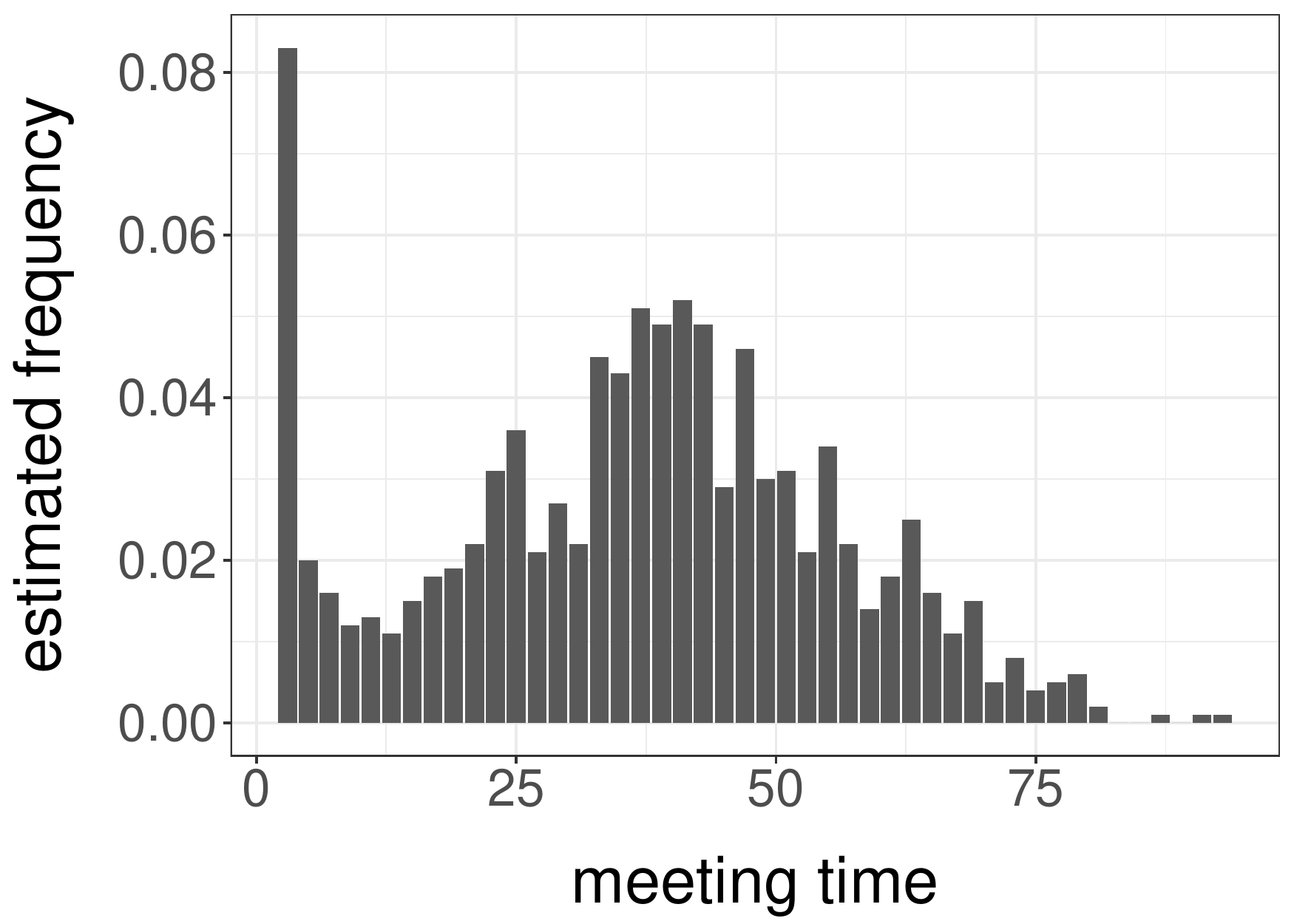}\label{fig:isingmt-1}}\subfloat[]{\centering{}\includegraphics[width=0.4\textwidth]{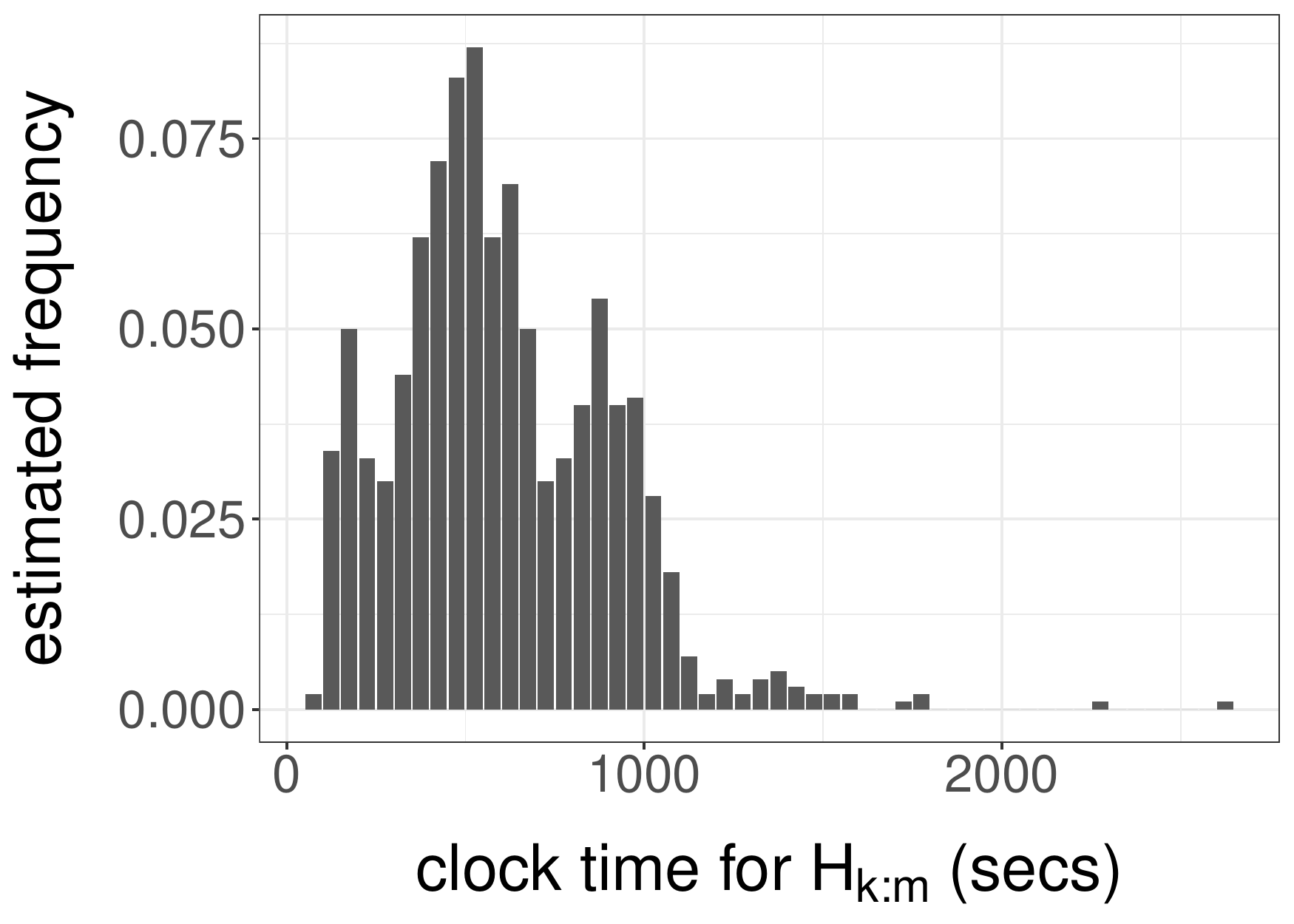}\label{fig:isingclock}}\caption{Coupled exchange algorithm for unbiased Bayesian inference with an $80\times80$ Ising model. Left: distribution of meeting times (1,000 runs). Right: clock time to obtain 1,000 unbiased estimators.}
\end{figure}

\section{Conclusion}

Markov chain Monte Carlo algorithms designed for scenarios where the target
density function is intractable can be coupled and utilized in the framework of
\citet{glynn2014exact,jacob2017unbiased}. The validity of the resulting unbiased estimators
can be related to polynomial drift conditions on the underlying Markov kernels. These estimators
open new ways of using parallel computing hardware to perform numerical integration in such scenarios.

In the context of state space models, in addition to parameter estimation, the proposed coupling strategy for PMMH
would additionally provide unbiased estimators with respect to the joint distribution over state and parameters.
This would enable unbiased smoothing under parameter uncertainty, instead of fixing the parameters as in \citet{jacob2017smoothing},
\citet{lee2018coupled} and \citet{middleton2018unbiasedpimh}.

%
%
%

\subsection*{Acknowledgement}

The authors are grateful to Jeremy Heng for very helpful discussions.
The data of Section \ref{subsec:Neuroscience-experiment} was kindly
shared by Demba Ba. The experiments of that section were performed
on the Odyssey cluster supported by the FAS Division of Science, Research
Computing Group at Harvard University. Pierre E. Jacob acknowledges
support from the National Science Foundation through grant DMS-1712872.

\bibliographystyle{abbrvnat}
\bibliography{unbiasedpmcmc}

\appendix

\section{Appendix}

In the rest of the paper we will often use the symbol $c$ to denote
a generic positive constant whose value may vary from line to line.

\subsection{Proof of Theorem \ref{thm:validity} \label{subsec:Proof-of-Theorem-Validity}}

The following provides a slight relaxation on Assumption 2.2 in \citet{jacob2017unbiased},
where geometric conditions were imposed on the tails of the distribution
of the meeting time. The following proof considers $H_{0}(Z,\tilde{Z})$
instead of $H_{k:m}$; one can first perform the same reasoning for
$H_{k}(Z,\tilde{Z})$ for all $k\geq0$, and then consider the finite
average $(m-k+1)^{-1}\sum_{\ell=k}^{m}H_{\ell}(Z,\tilde{Z})$ to obtain
the result for $H_{k:m}$.

By Assumption \ref{assumption:meetingtime}, it follows that $\mathds{E}\left[\tau\right]<\infty$.
This implies that the estimator $H_{0}$ can be computed in expected
finite time.
To show
that $H_{0}(Z,\tilde{Z})$ admits a finite variance, we proceed by
following \citet[~Proposition~3.1]{jacob2017unbiased}, adapting the
proof under the proposed weaker assumptions. We denote the complete
space of random variables with finite second moment by $L_{2}$. We
then construct a Cauchy sequence of random variables $H^{n}(Z,\tilde{Z})$
in $L_{2}$ converging to $H_{0}(Z,\tilde{Z})$, where $H^{n}(Z,\tilde{Z}):=\sum_{t=0}^{n}\Delta_{t}$
with $\Delta_{t}=h(Z_{t})-h(\tilde{Z}_{t-1})$ if $t>0$ and $\Delta_{t}=h(Z_{t})$
for $t=0$. As $\mathds{E}\left[\tau\right]<\infty$, we have $\mathds{P}(\tau<\infty)=1$
and $Z_{t}=\tilde{Z}_{t-1}$ for $t>\tau$. This implies that $H^{n}(Z,\tilde{Z})\rightarrow H_{0}(Z,\tilde{Z})$
almost surely. For positive integers $n,n'$ we have
\begin{align*}
\mathds{E}\left[\left(H^{n}(Z,\tilde{Z})-H^{n'}(Z,\tilde{Z})\right)^{2}\right] & =\sum_{s=n+1}^{n'}\sum_{t=n+1}^{n'}\mathds{E}[\Delta_{s}\Delta_{t}]\\
 & \le\sum_{s=n+1}^{n'}\sum_{t=n+1}^{n'}\mathds{E}[\Delta_{s}^{2}]^{1/2}\mathds{E}[\Delta_{t}^{2}]^{1/2}\\
 & =\left(\sum_{t=n+1}^{n'}\mathds{E}[\Delta_{t}^{2}]^{1/2}\right)^{2}.
\end{align*}
We note that $\mathds{E}[\Delta_{t}^{2}]=\mathds{E}[\Delta_{t}^{2}1_{\tau>t}]$.
Thus by Holder's inequality we obtain
\begin{align*}
\mathds{E}[\Delta_{t}^{2}] & \le\mathds{E}\left[\left|\Delta_{t}\right|^{2+\eta}\right]^{\frac{1}{1+\frac{\eta}{2}}}\mathds{E}[1_{\tau>t}]^{\frac{\eta}{2+\eta}}\\
 & \le c{}^{\frac{1}{1+\frac{\eta}{2}}}\mathds{P}(\tau>t){}^{\frac{\eta}{2+\eta}},
\end{align*}
where $\mathds{E}[\left|\Delta_{t}\right|^{2+\eta}]<c$ for all $t$
as $\mathds{E}[h(Z_{t})^{2+\eta}]<c$ by Assumption \ref{assumption:marginaldistributions}.
Consequently we have
\begin{align*}
\mathds{E}\left[\left(H^{n}(Z,\tilde{Z})-H^{n'}(Z,\tilde{Z})\right)^{2}\right] & \le\left(\sum_{t=n+1}^{n'}\left(c^{\frac{1}{1+\frac{\eta}{2}}}\mathds{P}(\tau>t){}^{\frac{\eta}{2+\eta}}\right)^{\frac{1}{2}}\right)^{2}\\
 & =c^{\frac{1}{1+\frac{\eta}{2}}}\left(\sum_{t=n+1}^{n'}\mathds{P}(\tau>t){}^{\frac{1}{2}\frac{\eta}{2+\eta}}\right)^{2}.
\end{align*}
With $\gamma=\frac{1}{2}\frac{\eta}{2+\eta}$, it follows from Assumption \ref{assumption:meetingtime} that $\mathds{P}(\tau>t)\le Kt^{-\kappa}$ for $\kappa>1/\gamma$
which yields
\begin{align*}
\sum_{t=n+1}^{\infty}\mathds{P}(\tau>t)^{\gamma} & \le K\sum_{t=n+1}^{\infty}\frac{1}{t^{\gamma\kappa}}\leq K\int_{n}^{\infty}\frac{dt}{t^{\gamma\kappa}}<\infty.
\end{align*}
We obtain $\lim_{n\rightarrow\infty}\sum_{t=n+1}^{\infty}\mathds{P}(\tau>t)^{\gamma}=0$.
This proves that $H_{n}(Z,\tilde{Z})$ is a Cauchy sequence in $L_{2}$.
We can thus conclude that the variance of $H_{0}(Z,\tilde{Z})$ is
finite and that its expectation is $\lim_{n\to\infty}\mathds{E}[H^{n}(Z,\tilde{Z})]=\lim_{n\to\infty}\mathds{E}[h(Z_{n})]=\pi(h)$.

\subsection{Proof of Theorem \ref{thm:pmmeetingtimes} \label{subsec:polyproof}}

The following establishes a bivariate drift condition that we will later use to bound moments of the hitting time to the diagonal set $\mathcal{D}$.
A similar statement
is provided in \citet[Lemma~1]{andrieu2015quantitative}.
\begin{lem}
Let $\bar{P}$ be a coupling of the Markov kernel $P$ with itself,
and $V$ be as in Assumption$\,$\ref{ass:minordrift}. Then the function
$\bar{V}(z,\tilde{z}):=V(z)+V(\tilde{z})-1$ satisfies
\begin{equation}
\bar{P}\bar{V}(z,\tilde{z})\leq\bar{V}(z,\tilde{z})-\epsilon_{b}\phi\circ\bar{V}(z,\tilde{z})+\bar{b}\mathds{1}_{\bar{C}}(z,\tilde{z}),\label{eq:driftbar}
\end{equation}
for all $(z,\tilde{z})\in\mathcal{Z\times Z}$, where $\bar{b}:=2b_{V}+\epsilon_{b}\phi(1)$
and $\bar{C}=C\times C$.
\end{lem}
\begin{proof}
For $(z,\tilde{z})\notin\bar{C}$ we have
\begin{align*}
\bar{P}\bar{V}(z,\tilde{z}) & =PV(z)+PV(\tilde{z})-1\\
 & \leq V(z)+V(\tilde{z})-1-\phi\circ V(z)-\phi\circ V(\tilde{z})+b_{V}\left(\mathds{1}_{C}(z)+\mathds{1}_{C}(\tilde{z})\right)\\
 & \leq V(z)+V(\tilde{z})-1-\phi\circ V(z)-\phi\circ V(\tilde{z})+b_{V}\\
 & =V(z)+V(\tilde{z})-1-\epsilon_{b}\left[\phi\circ V(z)+\phi\circ V(\tilde{z})\right]\\
 &\qquad-(1-\epsilon_{b})\left[\phi\circ V(z)+\phi\circ V(\tilde{z})\right]+b_{V}.
\end{align*}
Since $(z,\tilde{z})\notin\bar{C}$ then at least one of $z,\tilde{z}$
is not in $C$, and $\phi\circ V\geq0$, so

\begin{align*}
 & \leq V(z)+V(\tilde{z})-1-\epsilon_{b}\left[\phi\circ V(z)+\phi\circ V(\tilde{z})\right]-(1-\epsilon_{b})\inf_{z\notin C}\phi\circ V(z)+b_{V}\\
 & \leq V(z)+V(\tilde{z})-1-\epsilon_{b}\left[\phi\circ V(z)+\phi\circ V(\tilde{z})\right]-b_{V}+b_{V}\\
 & =\bar{V}(z,\tilde{z})-\epsilon_{b}\left[\phi\circ V(z)+\phi\circ V(\tilde{z})\right],
\end{align*}
where we used (\ref{eq:drift2-2}) in Assumption$\,$\ref{ass:minordrift}.
By two applications of the mean value theorem, we have that for any
$t\geq s\geq1$ there exist $r\in[t,t+s-1]$ and $r^{\ast}\in[1,s]$
such that
\[
\phi(t+s-1)-\phi(t)=\phi'(r)\left(s-1\right),\qquad\phi(s)-\phi(1)=\phi'(r^{\ast})(s-1).
\]
By concavity, since $t\geq s$ implies that $r\geq r^{\ast}$, it
follows that $\phi'(r)\leq\phi'(r^{\ast})$ and thus
\[
\phi(t+s-1)-\phi(t)\leq\phi(s)-\phi(1),
\]
or equivalently
\[
\phi(t+s-1)+\phi(1)\leq\phi(t)+\phi(s).
\]
Therefore, with $t=\max\{V(z),V(\tilde{z})\}$ and $s=\min\{V(z),V(\tilde{z})\}$
we get
\begin{equation}
\phi\circ\bar{V}(z,\tilde{z})+\phi(1)\leq\phi\circ V(z)+\phi\circ V(\tilde{z}),\label{eq:phi_concave_lyapunov}
\end{equation}
whence
\begin{align}
\bar{P}\bar{V}(z,\tilde{z}) & \leq\bar{V}(z,\tilde{z})-\epsilon_{b}\left[\phi\circ V(z)+\phi\circ V(\tilde{z})\right]\nonumber \\
 & \leq\bar{V}(z,\tilde{z})-\epsilon_{b}\left[\phi\left(\bar{V}(z,\tilde{z})\right)+\phi(1)\right]\nonumber \\
 & \leq\bar{V}(z,\tilde{z})-\epsilon_{b}\phi\circ\bar{V}(z,\tilde{z}).\label{eq:inequality1}
\end{align}
For $(z,\tilde{z})\in\bar{C}$ we get by Assumption \ref{ass:minordrift},
\begin{align}
\bar{P}\bar{V}(z,\tilde{z}) & =PV(z)+PV(\tilde{z})-1\nonumber \\
 & \leq V(z)-\phi\circ V(z)+b_{V}+V(\tilde{z})-\phi\circ V(\tilde{z})+b_{V}-1\nonumber \\
 & =\bar{V}(z,\tilde{z})-\phi\circ V(z)-\phi\circ V(\tilde{z})+2b_{V}.\label{eq:inequality2}
\end{align}
Combining (\ref{eq:inequality1}) and (\ref{eq:inequality2}), (\ref{eq:phi_concave_lyapunov})
and the fact that $\phi\geq0$, we have\textbf{ }for any $(z,\tilde{z})$
\begin{align*}
\bar{P}\bar{V}(z,\tilde{z}) & \leq\bar{V}(z,\tilde{z})-\epsilon_{b}\phi\circ\bar{V}(z,\tilde{z})\mathds{1}_{\bar{C}^{\mathsf{C}}}(z,\tilde{z})\\
 & \qquad-\left[\phi\circ V(z)+\phi\circ V(\tilde{z})-2b_{V}\right]\mathds{1}_{\bar{C}}(z,\tilde{z})\\
 & \leq\bar{V}(z,\tilde{z})-\epsilon_{b}\phi\circ\bar{V}(z,\tilde{z})\mathds{1}_{\bar{C}^{\mathsf{C}}}(z,\tilde{z})\\
 & \qquad-\left[\phi\circ\bar{V}(z,\tilde{z})+\phi(1)-2b_{V}\right]\mathds{1}_{\bar{C}}(z,\tilde{z})\\
 & =\bar{V}(z,\tilde{z})-\epsilon_{b}\phi\circ\bar{V}(z,\tilde{z})\mathds{1}_{\bar{C}^{\mathsf{C}}}(z,\tilde{z})\\
 & \qquad-\left[\epsilon_{b}\phi\circ\bar{V}(z,\tilde{z})+(1-\epsilon_{b})\phi\circ\bar{V}(z,\tilde{z})+\phi(1)-2b_{V}\right]\mathds{1}_{\bar{C}}(z,\tilde{z})\\
 & \leq\bar{V}(z,\tilde{z})-\epsilon_{b}\phi\circ\bar{V}(z,\tilde{z})+\left[2b_{V}-\phi(1)\right]\mathds{1}_{\bar{C}}(z,\tilde{z})\\
 & \leq\bar{V}(z,\tilde{z})-\epsilon_{b}\phi\circ\bar{V}(z,\tilde{z})+\left[2b_{V}+\phi(1)\right]\mathds{1}_{\bar{C}}(z,\tilde{z}).
\end{align*}

\end{proof}

The proof of Theorem \ref{thm:pmmeetingtimes} then follows through making use of \citet[Proposition~2.1]{douc2004practical},
which we provide below for the reader's convenience, noting that the
exact statement is taken from \citet[Proposition~4]{andrieu2015quantitative}.
We borrow the following definitions from \citet{andrieu2015quantitative}.
For any non-decreasing concave function $\psi:\left[1,\infty\right)\rightarrow\left(0,\infty\right)$,
let
\begin{equation}
H_{\psi}(v):=\int_{1}^{v}\frac{dx}{\psi(x)},\quad\label{eq:hphi}
\end{equation}
Let $H_{\psi}^{-1}:[0,\infty)\to[1,\infty)$ be its inverse. For $k\in\mathbb{N}$,
$n\geq0$, $\upsilon\geq1$,
let
\begin{equation}
\begin{split}  r_{\psi}(n)&:=\frac{\psi\circ H_{\psi}^{-1}(n)}{\psi(1)}\\
H_{k}(\upsilon) & :=H_{\psi}^{-1}\left(H_{\psi}\left(\upsilon\right)+k\right)-H_{\psi}^{-1}\left(k\right).
\end{split}
\label{eq:definitions}
\end{equation}

\begin{prop}
(Proposition~2.1 from \citet{douc2004practical}). \label{propn:4}
Assume that $P$ is a Markov kernel such that for some function $V\geq1$
we have
\[
PV(z)\leq V(z)-\psi\circ V(z)+b\mathds{1}_{C}\left(z\right),
\]
where $\psi:[1,\infty)\mapsto(0,\infty)$ is a nondecreasing concave
function. Let $r_{\psi}$ and $H_{\psi}$ be defined as in (\ref{eq:definitions}).
Then we have for $V_{k}:=H_{k}\circ V$
\[
PV_{k+1}(z)\leq V_{k}(z)-\psi(1)r_{\psi}(k)+br_{\psi}(k+1)\mathds{1}_{C}(z),\qquad k\geq0.
\]
\end{prop}
Equipped with the above results we proceed to the proof of Theorem \ref{thm:pmmeetingtimes}.
Applying Proposition \ref{propn:4} with $\bar{P}$, $\bar{V}$, $\psi=\epsilon_{b}\phi$ and
and $b=\bar{b}$, then letting
\[
r(n) :=\frac{\phi\circ H_{\phi}^{-1}(\epsilon_{b}n)}{\phi(1)}
\]
we have the sequence of drift conditions
\[
\bar{P}\bar{V}_{k+1}(z,\tilde{z})\leq\bar{V}_{k}(z,\tilde{z})-\epsilon_{b}\phi(1)r(k)+\bar{b}r(k+1)\mathds{1}_{\bar{C}}(z,\tilde{z}),\quad k\geq0,
\]
where $\bar{V}_{k}:=H_{k}\circ\bar{V}$. Letting $\tilde{V}_{k}=\bar{V}_{k}+1\geq1$
we obtain
\begin{equation}
\bar{P}\tilde{V}_{k+1}(z,\tilde{z})\leq\tilde{V}_{k}(z,\tilde{z})-\epsilon_{b}\phi(1)r(k)+\bar{b}r(k+1)\mathds{1}_{\bar{C}}(z,\tilde{z}),\quad k\geq0.\label{eq:driftsequence}
\end{equation}
To proceed we follow the proof of \citet[Proposition~2.5]{douc2004practical},
specifically the steps leading up to \citet[Equation~(2.6)]{douc2004practical}.
Notice that by Assumption~\ref{ass:pbarproperties} the diagonal $\mathcal{D}$ is an accessible set, since clearly $\pi_{\mathcal{D}}(\mathcal{D})=1>0$. Therefore by Dynkin's formula we have
\begin{align*}
\epsilon_b \phi(1)\mathds{E}_{z,\tilde{z}}
\left[ \sum_{k=0}^{\tau_\mathcal{D}-1} r(k)\right]
&\leq \tilde{V}_0(z,\tilde{z})
+ \bar{b} \mathds{E}_{z,\tilde{z}}\left[
		\sum_{k=0}^{\tau_\mathcal{D}-1} r(k+1) \mathds{1}_{\bar{C}} (\Xi_k)	\right],
\end{align*}
where in the above, $\mathds{E}_{z,\tilde{z}}$ denotes expectation with
respect to the probability measure under which the joint chain $\Xi_n:=\left(Z_{n},\tilde{Z}_{n-1}\right)$
is initialized at $\left(z,z'\right)$ and evolves according to the
transition kernel $\bar{P}$, $\tau_{\mathcal{D}}:=\inf\left\{ n\ge1:\Xi_{n}\in \mathcal{D}\right\} $,
$c_{1},c_{2}$ are positive constants depending on the set $B$ and
the various constants in the drift condition, but not on $(z,\tilde{z})$.
Notice that by Assumption~\ref{ass:pbarproperties} we have that for all
$(z,\tilde{z})\in \bar{C}$, and $\rho\in(0,1)$
$$K_\rho\left((z,\tilde{z}), \mathcal{D} \right)
:= \sum_{i=0}^\infty \rho^i \bar{P}^i\left((z,\tilde{z}), \mathcal{D} \right)
\geq \rho^{n_0} \epsilon.$$
In particular it easily follows that
$$\mathds{1}_{\bar{C}}\left((z,\tilde{z})\right) \leq (\rho^{n_0} \epsilon)^{-1}K_\rho \left((z,\tilde{z}), \mathcal{D} \right),$$
and therefore continuing from above
\begin{align*}
&\epsilon_b \phi(1)\mathds{E}_{z,\tilde{z}}
\left[ \sum_{k=0}^{\tau_\mathcal{D}-1} r(k)\right] \\
&\qquad\leq \tilde{V}_0(z,\tilde{z})
+ \frac{\bar{b}}{\rho^{n_0}\epsilon} \mathds{E}_{z,\tilde{z}}\left[
		\sum_{k=0}^{\tau_\mathcal{D}-1} r(k+1) K_\rho(\Xi_k, \mathcal{D})	\right]\\
&\qquad= \tilde{V}_0(z,\tilde{z})
+ \frac{\bar{b}}{\rho^{n_0}\epsilon} \sum_{i=0}^\infty \rho^i \mathds{E}_{z,\tilde{z}}\left[
		\sum_{k=0}^{\tau_\mathcal{D}-1} r(k+1) \bar{P}^i(\Xi_k, \mathcal{D})	\right]\\	
&\qquad= \tilde{V}_0(z,\tilde{z})
+ \frac{\bar{b}}{\rho^{n_0}\epsilon}  \sum_{i=0}^\infty \rho^i \sum_{k=0}^\infty\mathds{E}_{z,\tilde{z}}\Big[
		\mathds{1}\{k\leq \tau_\mathcal{D}-1\} r(k+1) \mathds{1}_{\mathcal{D}}(\Xi_{k+i})	\Big].
\end{align*}
A careful look above reveals that the integrand will be non-zero only for $k$ such that
$\tau_\mathcal{D}\leq k+i$ and $k\leq \tau_\mathcal{D}-1$. There are at most $i$ such values of $k$, and since $r(\cdot)$ is non-decreasing for each one of these values we will have $r(k+1)\leq r(\tau_\mathcal{D})$. Therefore
\begin{align*}
\epsilon_b \phi(1)\mathds{E}_{z,\tilde{z}}
\left[ \sum_{k=0}^{\tau_\mathcal{D}-1} r(k)\right]
&\leq \tilde{V}_0(z,\tilde{z})
+ \frac{\bar{b}}{\rho^{n_0}\epsilon} \sum_{i=0}^\infty \rho^i i \times
	\mathds{E}_{z,\tilde{z}}\Big[ r(\tau_\mathcal{D})\Big].
\end{align*}
Similarly to the proof of \citet[Proposition~2.5]{douc2004practical}, using the fact that $r(\cdot)$ grows sub-geometrically we can find for any $\delta>0$ a constant $c(\delta)>0$ such that
$$r(k)\leq \delta \sum_{j=0}^{k-1}r(j) + c(\delta),$$ and therefore conclude that for some constants $c_1, c_2$, independent of $(z,\tilde{z})$,  we have
\[
\mathds{E}_{z,\tilde{z}}\left[\sum_{k=0}^{\tau_{\mathcal{D}}-1}r(k)\right]\leq\frac{\tilde{V}{}_{0}(z,\tilde{z})+c_{1}}{c_{2}}.
\]
From the definition
of $\phi(y)$ we have that
\[
r(n)=\left[d(1-\alpha)\epsilon_{b}n+1\right]^{\alpha/(1-\alpha)}\geq cn^{\alpha/(1-\alpha)},
\]
where recall that $c$ denotes a generic constant whose value may
change from line to line. Thus for any $N$
\begin{align*}
\sum_{k=0}^{N}r(k) & \geq c\sum_{k=0}^{N}k^{\alpha/(1-\alpha)}\geq c\int_{x=0}^{N}x^{\alpha/(1-\alpha)}dx=cN^{1/(1-\alpha)},
\end{align*}
hence we obtain
\begin{align*}
\mathds{E}_{z,\tilde{z}}\left[\tau_{\mathcal{D}}^{1/(1-\alpha)}\right] & \leq c\mathds{E}_{z,\tilde{z}}\left[\sum_{k=0}^{\tau_{\mathcal{D}}-1}r(k)\right]\leq c\frac{\tilde{V}{}_{0}(z,\tilde{z})+c_{1}}{c_{2}}.
\end{align*}
We have that the chain $\left(Z_{n},\tilde{Z}_{n-1}\right)$ is initialised at $n=1$ under $\pi_{0}P\otimes\pi_{0}$.
Recalling the definition of $\tilde{V}_{0}$ we have that $\tilde{V}_{0}(z,\tilde{z})\le V(z)+V(\tilde{z})$ and as $\pi_{0}$ is compactly supported, $\pi_{0}\left(V\right)<\infty$.
Similarly by Assumption \ref{ass:minordrift} we have that $\pi_{0}P(V)<\infty$, in which case it follows that $\mathds{E}_{\pi_{0}P\otimes\pi_{0}}\left[\tau_{\mathcal{D}}^{1/(1-\alpha)}\right]<\infty$.
An application of Markov's inequality completes the proof
\[
\mathds{P}_{\pi_{0}P\otimes\pi_{0}}\left[\tau_{\mathcal{D}}\geq t\right]\leq\frac{\mathds{E}_{\pi_{0}P\otimes\pi_{0}}[\tau_{\mathcal{D}}^{1/(1-\alpha)}]}{t^{1/(1-\alpha)}}\leq\frac{c}{t^{1/(1-\alpha)}}.
\]

\subsection{Proof of Proposition \ref{prop:george33}}

\label{subsec:proofprop33}

To fix notation, we have that for any measurable functions $W:\mathcal{Z}\to[1,\infty)$,
$g:\mathcal{Z}\to\mathds{R}$, and a finite signed measure $\mu$
on $\mathcal{X}$, we write
\[
|g|_{W}:=\sup_{z\in\mathcal{Z}}\frac{|g(z)|}{W(z)},\qquad\|\mu\|_{W}:=\sup_{f:\|f\|_{W}\leq1}|\mu(f)|.
\]

Our starting point is Assumption \ref{ass:minordrift} which we restate
here
\begin{equation}
PV(z)\leq V(z)-dV^{\alpha}(z)+b_{V}\mathds{1}_{C}\left(z\right),\label{eq:drift_new}
\end{equation}
for some function $V:\mathcal{Z}\to[1,\infty)$, some $\alpha\in(0,1)$,
constants $b_{V},d>0$ and a small set $C$. As before we assume that
$(Z_{n},\tilde{Z}_{n-1})$ evolves according to $\bar{P}$, and that
marginally the components $Z_{n}$ and $\tilde{Z}_{n}$ evolve according
to $P$. Notice that we write $\mathds{E}$ for the measure with the
chains started from $\pi_{0}$ and $\mathds{E}_{\pi}$ for the measure
with the chains initialized at $\pi$.

By \citet[Lemma~3.5]{jarner2002polynomial} for any $\eta\in(0,1)$
there exist $b',d'>0$ such that
\begin{equation}
PV^{\gamma}(z)\leq V^{\gamma}(z)-d'V^{\alpha+\gamma-1}(z)+b'\mathds{1}_{C}\left(z\right).\label{eq:modified_drift}
\end{equation}
With $\gamma \in(1-\alpha,1)$ as in the statement of Proposition$\,$\ref{prop:george33},\textbf{
}we have that $\alpha+\gamma-1\in(0,1)$. Under this assumption, from
\eqref{eq:modified_drift}, \citet[Theorem~14.0.1]{meyn2009} applied
with $f=V^{\alpha+\gamma-1}$ and the fact that $\pi$ is a maximal
irreducibility measure (see \citet[Proposition~10.1.2]{meyn2009}),
it follows that $\pi(S_{V})=1$, with $S_{V}$ as defined in the statement
of Proposition \ref{prop:george33}. From this we conclude that $V$
is $\pi$-a.e.$\:$finite. Also from \citet[Theorem~14.0.1]{meyn2009},
since $\pi(V^{\gamma})\leq\pi(V^{4\gamma})^{1/4}<\infty$ by assumption,
we have that for all $\pi$-a.e.\ $z\in\mathcal{Z}$ there exists
a finite constant $c$ such that
\begin{equation}
\sum_{n=0}^{\infty}\|P^{n}(z,\cdot)-\pi\|_{V^{\alpha+\gamma-1}}\leq c(1+V^{\gamma}(z)).\label{eq:f_convergence}
\end{equation}
Since by assumption $|h|_{V^{\alpha+\gamma-1}}<\infty$,\textbf{ }we
have
\begin{align*}
\sum_{n=0}^{\infty}\left|P^{n}[h-\pi(h)](z)\right| & \leq\|h\|_{V^{\alpha+\gamma-1}}\sum_{n=0}^{\infty}\|P^{n}(z,\cdot)-\pi\|_{V^{\alpha+\gamma-1}}\\
 & \leq c\|h\|_{V^{\alpha+\gamma-1}}(1+V^{\gamma}(z))<\infty,
\end{align*}
for $\pi$-almost all $z$. Therefore the function
\[
g(z):=\sum_{j=0}^{\infty}P^{j}\left[h-\pi(h)\right](z)
\]
is well-defined and satisfies $|g|_{V^{\gamma}}<\infty$, $\pi(g^{2})<\infty$,
where the second property follows from $\pi(V^{4\gamma})<\infty$. In
particular it follows that $g-Pg=h-\pi(h)$, and therefore $g$ is
the solution to the Poisson equation with respect to $P$ and $h$.
We continue with the calculation in the proof of \citet[Proposition~3.3]{jacob2017unbiased}.
Let
\begin{align*}
S_{j}^{(N)}:=\mathds{1}\{\tau_{\mathcal{D}}>j\}\sum_{t=j}^{N\wedge\tau_{\mathcal{D}}-1}b_{t}\left[h(Z_{t})-h(\tilde{Z}_{t-1})\right],
\end{align*}
where $(b_{t})_{t\geq0}$ is an arbitrary bounded sequence. Writing
$\mathbf{Z}_{t}:=(Z_{t},\tilde{Z}_{t-1})$, $\bar{g}(x,y)=g(x)-g(y)$
and $\bar{P}$ for the transition kernel of $\mathbf{Z}_{t}$ we then
have
\begin{align*}
h(Z_{t})-h(\tilde{Z}_{t-1}) & =\left[h(Z_{t})-\pi(h)\right]-\left[h(\tilde{Z}_{t-1})-\pi(h)\right]\\
 & =\left[g(Z_{t})-Pg(Z_{t})\right]-\left[g(\tilde{Z}_{t-1})-Pg(\tilde{Z}_{t-1})\right]\\
 & =\left[g(Z_{t})-g(\tilde{Z}_{t-1})\right]-\left[Pg(Z_{t})-Pg(\tilde{Z}_{t-1})\right]\\
 & =\bar{g}(\mathbf{Z}_{t})-\bar{P}\bar{g}(\mathbf{Z}_{t}),
\end{align*}
where we used the fact that, by construction of $\bar{P}$, we have
$\bar{P}\bar{g}(z,\tilde{z})=Pg(z)-Pg(\tilde{z})$.

Then from \citet[Equation~(A.3)]{jacob2017unbiased} we have
\begin{align*}
\mathds{E}\left\{ \left[S_{j}^{(N)}\right]^{2}\right\} &\leq4\sum_{t=j}^{N-1}b_{t}^{2}\mathds{E}\left\{ \left[\bar{g}(\mathbf{Z}_{t+1})-\bar{P}\bar{g}(\mathbf{Z}_{t})\right]^{2}\mathds{1}\{\tau_{\mathcal{D}}>t\}\right\}  \\
 &\qquad +4b_{j}^{2}\mathds{E}\left[\bar{g}^{2}(\mathbf{Z}_{j})\mathds{1}\{\tau_{\mathcal{D}}>j\}\right]+4b_{N}^{2}\mathds{E}\left[\bar{g}^{2}(\mathbf{Z}_{N})\mathds{1}\{\tau_{\mathcal{D}}>N\}\right] \\
 &\qquad+4\left\{ \sum_{t=j}^{N-1}|b_{t+1}-b_{t}|\mathds{E}^{1/2}\left[\bar{g}^{2}(\mathbf{Z}_{t+1})\mathds{1}\{\tau_{\mathcal{D}}>t+1\}\right]\right\} ^{2},
\end{align*}
and we proceed to bound these terms. Letting $\mathcal{F}_{t}:=\sigma\left(\boldsymbol{Z_{s}};0\leq s\leq t\right)$,
notice that
\begin{align*}
\lefteqn{} & \mathds{E}\left\{ \left[\bar{g}(\mathbf{Z}_{t+1})-\bar{P}\bar{g}(\mathbf{Z}_{t})\right]^{2}\mathds{1}\{\tau_{\mathcal{D}}>t\}\right\} \\
& \mathds{E}\left\{\mathds{E}\left[ \left. \left(\bar{g}(\mathbf{Z}_{t+1})-\bar{P}\bar{g}(\mathbf{Z}_{t})\right)^{2}\mathds{1}\{\tau_{\mathcal{D}}>t\} \right| \mathcal{F}_t \right] \right\} \\ 
 & =\mathds{E}\left\{ \bar{g}(\mathbf{Z}_{t+1})^{2}\mathds{1}\{\tau_{\mathcal{D}}>t\}\right\} -\mathds{E}\left\{ \bar{P}\bar{g}(\mathbf{Z}_{t})^{2}\mathds{1}\{\tau_{\mathcal{D}}>t\}\right\} \\
 & \leq\mathds{E}\left\{ \bar{g}(\mathbf{Z}_{t+1})^{2}\mathds{1}\{\tau_{\mathcal{D}}>t\}\right\} \leq|g|_{V^{\gamma}}^{2}\mathds{E}\left\{ \left[V^{\gamma}(Z_{t})+V^{\gamma}(\tilde{Z}_{t-1})\right]^{2}\mathds{1}\{\tau_{\mathcal{D}}>t\}\right\} .
\end{align*}
We next bound the last quantity using the fact that $(a+b)^{2}\leq2a^{2}+2b^{2}$
and the Cauchy-Schwarz inequality
\begin{align*}
\mathds{E}&\left\{ \left[V^{\gamma}(Z_{t})+V^{\gamma}(\tilde{Z}_{t-1})\right]^{2}\mathds{1}\{\tau_{\mathcal{D}}>t\}\right\}   \leq2\mathds{E}\left\{ \left[V^{2\gamma}(Z_{t})+V^{2\gamma}(\tilde{Z}_{t-1})\right]\mathds{1}\{\tau_{\mathcal{D}}>t\}\right\} \\
 & \leq c\left[\mathds{E}\left\{ V^{4\gamma}(Z_{t})\right\} +\mathds{E}\left\{ V^{4\gamma}(\tilde{Z}_{t-1})\right\} \right]^{1/2}\mathds{P}\left(\tau_{\mathcal{\mathcal{D}}}>t\right)^{1/2}.
\end{align*}
Finally notice that since $V$ is non-negative
\begin{align*}
\mathds{E}\left\{ V^{4\gamma}(Z_{t})\right\}  & \leq\left\Vert \frac{d\pi_{0}}{d\pi}\right\Vert _{\infty}\mathds{E}_{\pi}\left\{ V^{4\gamma}(Z_{t})\right\} \\
& \leq\left\Vert \frac{d\pi_{0}}{d\pi}\right\Vert _{\infty}\mathds{E}_{\pi}\left\{ V^{4\gamma}(Z_{0})\right\} =c\pi(V^{4\gamma})<\infty,
\end{align*}
where we used the fact that when started from $\pi$ and evolved through
$\bar{P}$, the Markov chain $\{Z_{t}\}_{t\geq0}$ is stationary.
From the above and Theorem$\,$\ref{thm:pmmeetingtimes} we conclude
that there exists a positive constant $c<\infty$ such that
\begin{align*}
\mathds{E}\left\{ \left[\bar{g}(\mathbf{Z}_{t+1})-\bar{P}\bar{g}(\mathbf{Z}_{t})\right]^{2}\mathds{1}\{\tau_{\mathcal{D}}>t\}\right\} \leq\frac{c}{t^{\kappa/2}}.
\end{align*}
On the other hand for terms of the form $\mathds{E}[\bar{g}^{2}(\mathbf{Z}_{t})\mathds{1}\{\tau_{\mathcal{D}}>t\}]$,
using the same techniques we have
\begin{align*}
\mathds{E}\left[\bar{g}^{2}(\mathbf{Z}_{t})\mathds{1}\{\tau_{\mathcal{D}}>t\}\right]\leq|g|_{V^{\gamma}}^{2}\mathds{E}\left\{ \left[V^{\gamma}(Z_{t})+V^{\gamma}(\tilde{Z}_{t-1})\right]^{2}\mathds{1}\{\tau_{\mathcal{D}}>t\}\right\} \leq\frac{c}{t^{\kappa/2}}.
\end{align*}
Overall we thus have that
\begin{align*}
\mathds{E}\left\{ \left[S_{j}^{(N)}\right]^{2}\right\}  & \leq c\left[\frac{b_{j}^{2}}{j^{\kappa/2}}+\frac{b_{N}^{2}}{N^{\kappa/2}}+\sum_{t=j}^{N-1}\frac{b_{t}^{2}}{t^{\kappa/2}}+\left(\sum_{t=j}^{N-1}\frac{|b_{t+1}-b_{t}|}{t^{\kappa/2}}\right)^{2}\right],\\
\mathds{E}\left\{ S_{j}^{2}\right\}  & \leq c\left[\frac{b_{j}^{2}}{j^{\kappa/2}}+\sum_{t\geq j}\frac{b_{t}^{2}}{t^{\kappa/2}}+\left(\sum_{t=j}^{\infty}\frac{|b_{t+1}-b_{t}|}{t^{\kappa/2}}\right)^{2}\right],
\end{align*}
where $S_{j}:=\lim_{N\to\infty}S_{j}^{\left(N\right)}$ is the limit
in the $L^{2}$ sense as in \citet[Proposition~3.1]{jacob2017unbiased}.
Setting $b_{j}=0$, $b_{t}:=(t-j)/(m-j+1)$ for $j<t<m+1$ and $b_{t}:=1$
for $t>m+1$ we then obtain
\begin{align*}
\mathds{E}\left[S_{j}^{2}\right] & \leq c\left[\sum_{t=j+1}^{m}\frac{(t-j)^{2}}{(m-j+1)^{2}t^{\kappa/2}}+\sum_{t=m+1}^{\infty}\frac{1}{t^{\kappa/2}}+\left(\sum_{t=j}^{m+1}\frac{1}{(m-j+1)t^{\kappa/2}}\right)^{2}\right].
\end{align*}
For the first term notice that, after changing variables $r=t-j$
and writing $M=m-j+1$ we have
\begin{align*}
\sum_{t=j+1}^{m+1}\frac{(t-j)^{2}}{(m-j+1)^{2}t^{\kappa/2}} & =\frac{1}{M^{2}}\sum_{r=1}^{M}\frac{r^{2}}{(r+j)^{\kappa/2}}\\
 & \leq\frac{c}{M^{2}}\int_{x=1}^{M}\frac{x^{2}}{(x+j)^{\kappa/2}}dx\\
 & =\frac{c}{M^{2}j^{\kappa/2}}\int_{x=1}^{M}\frac{x^{2}}{\left(x/j+1\right)^{\kappa/2}}dx\qquad\text{(changing \ensuremath{z=x/j})}\\
 & \leq\frac{c}{M^{2}j^{\kappa/2}}\int_{z=1/j}^{M/j}\frac{j^{3}z^{2}}{(z+1)^{\kappa/2}}dz\\
 & =\frac{c}{M^{2}j^{\kappa/2-3}}\int_{z=1/j}^{M/j}\frac{z^{2}}{(z+1)^{\kappa/2}}dz\leq\frac{c}{M^{2}j^{\kappa/2-3}},
\end{align*}
since by assumption $\kappa=1/(1-\alpha)>6$. Finally we get
\begin{align*}
\mathds{E}\left[S_{j}^{2}\right] & \leq c\left[\frac{1}{\left(m-j+1\right)^{2}j^{\kappa/2-3}}+\frac{1}{m^{\kappa/2-1}}+\frac{1}{(m-j+1)^{2}j^{\kappa-2}}\right]\\
 & =c\left[\frac{1}{m^{\kappa/2-1}}+\frac{1}{(m-j+1)^{2}}\left(\frac{1}{j^{\kappa/2-3}}+\frac{1}{j^{\kappa-2}}\right)\right].\\
 & \leq c\left[\frac{1}{m^{\kappa/2-1}}+\frac{1}{(m-j+1)^{2}}\frac{1}{j^{\kappa/2-3}}\right]
\end{align*}
as $\kappa/2-3\leq\kappa-2$. With our choice of sequence $\left(b_{t}\right)_{t\geq0}$,
$S_{j}$ coincides with $\mathrm{BC}_{j:m}$ in the notation of the
statement of the proposition which thus follows from the above.%

\subsection{Proof of Proposition \ref{prop:resultsforpseudo}}

First we want to prove the minorization condition \eqref{ass:marginal_minor}
for the set $C=B\left(0,M\right)\times[\underline{w},\overline{w}]$,
where $M,\underline{w},\overline{w}>0$ are given and fixed. That
is, we want to establish that there exist $\epsilon_{0}>0$
and a probability measure $\nu$ such that
\begin{align*}
P\left(\left(\theta,w\right),d\theta',dw'\right)\geq\epsilon_{0}\nu\left(d\theta',dw'\right)
\end{align*}
for all $(\theta,w)\in C$. We have
\begin{align*}
P\left(\left(\theta,w\right),d\theta',dw'\right)\geq & \overline{g}_{\theta'}\left(w'\right)\min\left\{ q\left(\theta,\theta'\right),\frac{\pi\left(\theta'\right)}{\pi\left(\theta\right)}q(\theta',\theta)\right\} \min\left\{ 1,\frac{w'}{w}\right\} d\theta'dw'\\
\geq\varepsilon_{\pi} & \mathbb{I}\left(\theta'\in B(0,M)\right)\min\left\{ q\left(\theta,\theta'\right),q(\theta',\theta)\right\}\\
&\qquad\min\left\{ \overline{g}_{\theta'}\left(w'\right),\overline{g}_{\theta'}\left(w'\right)\frac{w'}{\overline{w}}\right\} d\theta'dw',
\end{align*}
where
\[
\varepsilon_{\pi}:=\frac{\inf_{\theta:\left|\theta\right|\leq M}\pi\left(\theta'\right)}{\sup_{\theta:\left|\theta\right|\leq M}\pi\left(\theta\right)}>0,
\]
by the assumption that $\pi$ is bounded from above, and bounded away from zero
on all compact sets. Since the proposal $q$ is bounded away from
zero on compact sets we also have that $\min\left\{ q\left(\theta,\theta'\right),q(\theta',\theta)\right\} \geq\varepsilon_{q}$
for $|\theta'-\theta|<2M$ which ensures that
\[
P\left(\left(\theta,w\right),d\theta',dw'\right)\geq\varepsilon_{q}\varepsilon_{\pi}\min\left\{ \overline{g}_{\theta'}\left(w'\right),\overline{g}_{\theta'}\left(w'\right)\frac{w'}{\overline{w}}\right\} d\theta'dw'.
\]
This can be rewritten as
\[
P\left(\left(\theta,w\right),d\theta',dw'\right)\geq\varepsilon_{q}\varepsilon_{\pi}\mathbb{I}\left(\theta'\in B(0,M)\right)Z\left(\theta'\right)\widetilde{g}_{\theta'}\left(w\right)d\theta'dw'
\]
with
\[
Z\left(\theta\right):=\intop\overline{g}_{\theta}\left(w\right)\min\left\{ 1,\frac{w}{\overline{w}}\right\} dw\leq1,
\]
and
\[
\widetilde{g}_{\theta}\left(w\right)=Z{}^{-1}\left(\theta\right)\overline{g}_{\theta}\left(w\right)\min\left\{ 1,\frac{w}{\overline{w}}\right\} .
\]

Suppose now that for fixed $M,\overline{w}$ we have
\[
\inf_{\theta:\left|\theta\right|\leq M}Z\left(\theta\right)=0,
\]
which implies that there is a sequence $\theta_{n}\in B\left(0,M\right)$
such that $\lim_{n\to\infty}Z(\theta_{n})=0$. Since $B(0,M)$ is
compact we can extract a convergent subsequence $\theta_{n_{k}}\to\bar{\theta}\in B(0,M)$
such that $\lim_{k\to\infty}Z(\theta_{n_{k}})=0$. By weak convergence,
since $w\mapsto\min\left\{ 1,w/\overline{w}\right\} $ is bounded
and continuous, we also have that
\[
0=\lim_{k\to\infty}Z(\theta_{n_{k}})=\lim_{k\to\infty}\intop\overline{g}_{\theta_{n_{k}}}\left(w\right)\min\left\{ 1,\frac{w}{\overline{w}}\right\} dw=\intop\overline{g}_{\overline{\theta}}\left(w\right)\min\left\{ 1,\frac{w}{\overline{w}}\right\} dw.
\]
Since $w\mapsto\min\left\{ 1,w/\overline{w}\right\} $ is strictly
positive for $w>0$, this implies that the support of $\overline{g}_{\overline{\theta}}$
is $\left\{ 0\right\} $ which is a contradiction, since in that case
necessarily $\intop\overline{g}_{\overline{\theta}}\left(w\right)w\,dw=0\neq1$.
Therefore we conclude that for all finite $M,\overline{w}>0$,
there exists \emph{$\varepsilon_{Z}(M,\overline{w})>0$} such that
$Z\left(\theta\right)>\varepsilon_{Z}\left(M,\overline{w}\right)$,
for all $\theta\in B\left(0,M\right)$.

Therefore we obtain
\[
P\left(\left(\theta,w\right),d\theta',dw'\right)\geq\varepsilon_{Z}\varepsilon_{q}\varepsilon_{\pi}\mathbb{I}\left(\theta'\in B(0,M)\right)\widetilde{g}_{\theta'}\left(w\right)d\theta'dw',
\]
which proves the result for $\epsilon_{0}=\varepsilon_{Z}\varepsilon_{q}\varepsilon_{\pi}\mathrm{vol\{}B(0,M)\}$
and with minorising measure $\nu(d\theta',dw')=\mathcal{U}\left(\theta'\in B(0,M)\right)\widetilde{g}_{\theta'}\left(w\right)$.

Next we establish that the minorization condition \eqref{ass:joint_minor}
holds for $\bar{P}$, the coupled transition kernel defined by Algorithm \ref{alg:coupledPMMH}, and $C$ as defined above. Let the current
states be $z:=\left(\theta,w\right),\tilde{z}:=(\tilde{\theta},\tilde{w})\in C$
respectively. According to Algorithm \ref{alg:coupledPMMH} the next parameter states $\theta',\tilde{\theta}'$
will be sampled from $\mathfrak{Q}((\theta,\tilde{\theta}),\mathrm{d}\theta',\mathrm{d}\tilde{\theta}')$,
the $\gamma$-coupling of $q\left(\cdot\left|\theta\right.\right)$
and $q(\cdot|\tilde{\theta})$. This is the
maximal coupling generated by the rejection sampler described in \citet{jacob2017unbiased}. If the coupling is successful, that is $\theta'=\tilde{\theta}'$,
then the algorithm samples $w'\sim\bar{g}_{\theta'}\left(\cdot\right)$,
sets $\tilde{w}'=w'$ in which case we know by definition that $((\theta',w'),(\widetilde{\theta}',\widetilde{w}'))\in\mathcal{D}$
if the proposal $\left(\theta',w'\right)$ is accepted since the same
uniform is used in both acceptance steps. Therefore under the
coupled transition kernel $\bar{P}$ , writing $z:=\left(\theta,w\right),\tilde{z}:=\left(\tilde{\theta},\tilde{w}\right)$ and letting $\mathcal{D}_{\theta}:=\left\{ \left(\theta,\tilde{\theta}\right):\theta=\tilde{\theta}\right\} $
be the diagonal of $\Theta\times\Theta$,
we have for $\left(z,z'\right)\in C\times C$ that
\begin{align*}
\bar{P}\left(\left(z,z'\right),\mathcal{D}\right) & \geq\bar{P}\left(\left(z,z'\right),\mathcal{D}\cap\left(B(0,M)\times\mathds{R^{+}}\right)^{2}\right)\\
 & =\iint_{\mathcal{D_{\theta}}\cap B(0,M)^{2}}\mathfrak{Q}\left(\left(\theta,\tilde{\theta}\right),\mathrm{d}\theta',\mathrm{d}\tilde{\theta}'\right)\int_{\mathds{R^{+}}}\bar{g}_{\theta'}\left(w'\right)\\
 & \qquad\int_{u=0}^{1}\mathbb{I}\left[u\leq\min\left\{ 1,\frac{\pi\left(\theta'\right)}{\pi\left(\theta\right)}\frac{w'}{w}\right\} \right]\mathbb{I}\left[u\leq\min\left\{ 1,\frac{\pi\left(\theta'\right)}{\pi\left(\tilde{\theta}\right)}\frac{w'}{\tilde{w}}\right\} \right]\mathrm{d}u\mathrm{d}w'\\
 & =\iint_{\mathcal{D_{\theta}}\cap B(0,M)^{2}}\mathfrak{Q}\left(\left(\theta,\tilde{\theta}\right),\mathrm{d}\theta',\mathrm{d}\tilde{\theta}'\right)
\\
&\qquad\int_{\mathds{R^{+}}}\bar{g}_{\theta'}\left(w'\right)\int_{u=0}^{1}\mathbb{I}\left[u\leq\min\left\{ 1,\frac{\pi\left(\theta'\right)}{\pi\left(\theta\right)}\frac{w'}{w},\frac{\pi\left(\theta'\right)}{\pi\left(\tilde{\theta}\right)}\frac{w'}{\tilde{w}}\right\} \right]\mathrm{d}u\mathrm{d}w'\\
 & =\iint_{\mathcal{D_{\theta}}\cap B(0,M)^{2}}\mathfrak{Q}\left(\left(\theta,\tilde{\theta}\right),\mathrm{d}\theta',\mathrm{d}\tilde{\theta}'\right)\\
 &\qquad\int_{\mathds{R^{+}}}\bar{g}_{\theta'}\left(w'\right)\min\left\{ 1,\frac{\pi\left(\theta'\right)}{\pi\left(\theta\right)}\frac{w'}{w},\frac{\pi\left(\theta'\right)}{\pi\left(\tilde{\theta}\right)}\frac{w'}{\tilde{w}}\right\} \mathrm{d}w',
\end{align*}
where we also used the fact that the proposal is symmetric by assumption.
Continuing from the above inequality, letting $\varepsilon_{\pi},\varepsilon_{q}$
and $\varepsilon_{Z}$ be as above, we have that
\begin{align*}
\bar{P}\left(\left(z,z'\right),\mathcal{D}\right) & \geq\iint_{\mathcal{D_{\theta}}\cap B(0,M)^{2}}\mathfrak{Q}\left(\left(\theta,\tilde{\theta}\right),\mathrm{d}\theta',\mathrm{d}\tilde{\theta}'\right)\\
&\qquad\int_{\mathds{R^{+}}}\bar{g}_{\theta'}\left(w'\right)\min\left\{ 1,\varepsilon_{\pi}\frac{w'}{\overline{w}},\varepsilon_{\pi}\frac{w'}{\overline{w}}\right\} \mathrm{d}w'\\
 &\geq \iint_{\mathcal{D_{\theta}}\cap B(0,M)^{2}}\mathfrak{Q}\left(\left(\theta,\tilde{\theta}\right),\mathrm{d}\theta',\mathrm{d}\tilde{\theta}'\right)\\
 &\qquad\int_{\mathds{R^{+}}}\bar{g}_{\theta'}\left(w'\right)\min\left\{ 1,\varepsilon_{\pi}\right\} \min\left\{ 1,\frac{w'}{\overline{w}}\right\} \mathrm{d}w'\\
 & =\min\left\{ 1,\varepsilon_{\pi}\right\} \iint_{\mathcal{D_{\theta}}\cap B(0,M)^{2}}\mathfrak{Q}\left(\left(\theta,\tilde{\theta}\right),\mathrm{d}\theta',\mathrm{d}\tilde{\theta}'\right)Z\left(\theta'\right)\int_{\mathds{R^{+}}}\widetilde{g}_{\theta'}\left(w'\right)\mathrm{d}w'\\
 & \geq\varepsilon_{Z}\varepsilon_{\pi}\iint_{\mathcal{D_{\theta}}\cap B(0,M)^{2}}\mathfrak{Q}\left(\left(\theta,\tilde{\theta}\right),\mathrm{d}\theta',\mathrm{d}\tilde{\theta}'\right)\\
 &\geq \varepsilon_{Z}\varepsilon_{\pi}\int_{B(0,M)}\min\left\{ q\left(\theta'\left|\theta\right.\right),q\left(\theta'\left|\tilde{\theta}\right.\right)\right\} \mathrm{d}\theta'\\
 & \geq\varepsilon_{Z}\varepsilon_{\pi}\int_{B(0,M)}\varepsilon_{q}\mathrm{d}\theta'\\
 & =\varepsilon_{Z}\varepsilon_{\pi}\varepsilon_{q}\mathrm{vol}\left(B\left(0,M\right)\right)>0,
\end{align*}
where we used the fact that in the $\gamma-$coupling, conditionally
on the coupling succeeding, the variables are sampled from a density
proportional to the minimum of their respective densities. This establishes
that  condition \eqref{ass:joint_minor} holds with $C=B\left(0,M\right)\times[\underline{w},\overline{w}]$
for any $M,\underline{w},\overline{w}$.

Next we establish that $\bar{P}$ is $\pi_{\mathcal{D}}-$irreducible.
Let $A\subset\mathcal{D}$ such that $\pi_{\mathcal{D}}\left(A\right)>0$.
For sets $A\subset\mathcal{D}$ we will write $A^{(1)}$ for the projection
onto its first coordinate, that is if $A\subset\mathcal{D}$ then $A=A^{(1)}\times A^{(1)}$.
We need to show that for any $z,\tilde{z}\in\mathcal{Z}$ there exists
$n\geq1$ such that $\bar{P}^{n}\left(\left(z,\tilde{z}\right),A\right)>0$.
Notice that by construction if $\left(z,\tilde{z}\right)\in\mathcal{D}$
then $\bar{P}\left(\left(z,\tilde{z}\right),\mathrm{d}z',\mathrm{d}\tilde{z}'\right)=P\left(z,\mathrm{d}z'\right)\delta_{z'}\left(\mathrm{d}\tilde{z}'\right)$,
that is the chain couples automatically from the diagonal and proceeds
as the pseudo-marginal kernel $P$.
Letting $z,\tilde{z}\in\mathcal{Z}$ and $n\geq1$ we have
\begin{align*}
\bar{P}^{n+1}&\left(\left(z,\tilde{z}\right),A\right)  \geq\iint_{\mathcal{D}}\bar{P}\left(\left(z,\tilde{z}\right),\mathrm{d}z',\mathrm{d}\tilde{z}'\right)\bar{P}^{n}\left(\left(z',\tilde{z}'\right),A\right)\\
 & =\iint_{\mathcal{D_{\theta}}}\mathfrak{Q}\left(\left(\theta,\tilde{\theta}\right),\mathrm{d}\theta',\mathrm{d}\tilde{\theta}'\right)\\
 &\qquad\int_{\mathds{R^{+}}}\bar{g}_{\theta'}\left(w'\right)\min\left\{ 1,\frac{\pi\left(\theta'\right)}{\pi\left(\theta\right)}\frac{w'}{w},\frac{\pi\left(\theta'\right)}{\pi\left(\theta\right)}\frac{w'}{\tilde{w}}\right\} \mathrm{d}w'\int P^{n}\left(\left(\theta',w'\right),A\right),
\end{align*}
where we have provided a lower bound by considering the event where
the joint chain couples in the first step and then moves to the set
$A$ in $n$ steps. Continuing we have
\begin{align*}
\bar{P}^{n+1}&\left(\left(z,\tilde{z}\right),A\right)  \geq\int_{\Theta}\min\left\{ q\left(\theta,\theta'\right),q\left(\tilde{\theta},\theta'\right)\right\} \mathrm{d}\theta'\int_{\Theta}\frac{\min\left\{ q\left(\theta,\theta'\right),q\left(\tilde{\theta},\theta'\right)\right\} }{\int_{\Theta}\min\left\{ q\left(\theta,\theta'\right),q\left(\tilde{\theta},\theta'\right)\right\} \mathrm{d}\theta'}\\
 & \qquad\int_{\mathds{R^{+}}}\bar{g}_{\theta'}\left(w'\right)\min\left\{ 1,\frac{\pi\left(\theta'\right)}{\pi\left(\theta\right)}\frac{w'}{w},\frac{\pi\left(\theta'\right)}{\pi\left(\theta\right)}\frac{w'}{\tilde{w}}\right\} \mathrm{d}w'\int P^{n}\left(\left(\theta',w'\right),A^{(1)}\right)\mathrm{d}\theta'\\
 & =\int_{\Theta}\min\left\{ q\left(\theta,\theta'\right),q\left(\tilde{\theta},\theta'\right)\right\}\\
 &\qquad \int_{\mathds{R^{+}}}\bar{g}_{\theta'}\left(w'\right)\min\left\{ 1,\frac{\pi\left(\theta'\right)}{\pi\left(\theta\right)}\frac{w'}{w},\frac{\pi\left(\theta'\right)}{\pi\left(\theta\right)}\frac{w'}{\tilde{w}}\right\} \mathrm{d}w'\int P^{n}\left(\left(\theta',w'\right),A^{(1)}\right)\mathrm{d}\theta'.
\end{align*}
Therefore we have that
\begin{align*}
\sum_{n=0}^{\infty}& 2^{-(n+1)}\bar{P}^{n+1}\left(\left(z,\tilde{z}\right),A\right)  \geq\int_{\Theta}\min\left\{ q\left(\theta,\theta'\right),q\left(\tilde{\theta},\theta'\right)\right\} \\
&\qquad \int_{\mathds{R^{+}}}\bar{g}_{\theta'}\left(w'\right)\min\left\{ 1,\frac{\pi\left(\theta'\right)}{\pi\left(\theta\right)}\frac{w'}{w},\frac{\pi\left(\theta'\right)}{\pi\left(\theta\right)}\frac{w'}{\tilde{w}}\right\} \mathrm{d}w'
 \\
 &\qquad \sum_{n=0}^{\infty}2^{-(n+1)}\int P^{n}\left(\left(\theta',w'\right),A^{(1)}\right)\mathrm{d}\theta'.
\end{align*}
By Assumption \ref{assu:The-posterior-density} and \citet[Theorem 2.2][]{roberts1996geometric}
it easily follows that the exact algorithm is $\pi-$irreducible and
aperiodic. Since by assumption we have $\varrho_{\mathrm{PM}}\left(\theta,w\right)<1$,
we deduce from \citet[Theorem 1][]{andrieu2009pseudo} that the kernel
$P$ is irreducible, hence $\pi-$irreducible. This further implies
that
\[
\sum_{n=0}^{\infty}2^{-(n+1)}\int P^{n}\left(\left(\theta',w'\right),B\right)>0,
\]
for all $\left(\theta',w'\right)$ and sets $B$ such that $\pi(B)>0$
by \citet[Proposition 4.2.1][]{meyn2009}. Since by assumption $\pi\left(A^{(1)}\right)=\pi_{\mathcal{D}}(A)>0$
the integrand above will be strictly positive on a set of non-vanishing
Lebesgue measure whence $\bar{P}$ is $\pi_{\mathcal{D}}-$irreducible.
Finally to establish aperiodicity first notice that by assumption
and continuity of the measures defined by the densities $\bar{g}_{\theta}(\cdot)$,
we have that $\pi(C)>0$. Letting $\mathcal{D}_{C}:=\left\{ \left(z,\tilde{z}\right)\in\mathcal{\mathcal{D}}:z\in C\right\} $
and following the steps proving Equation (\ref{ass:pbarproperties}) we can establish
that for some $\epsilon'>0$
\[
\inf_{z\in\mathcal{D_{C}}}\bar{P}\left(z,\mathcal{D}\right)\geq\epsilon',
\]
and since $\pi(\mathcal{D_{C}})>0$ this proves the aperiodicity of $\bar{P}$.

\subsection{Proof of Proposition \ref{prop:propmoments} \label{sec:proofofmoments}}

We proceed to bound the moments of the likelihood
estimate. For $y=1$, letting
\[
\bar{\mathcal{Z}}:=\frac{\text{B}(\alpha,\beta(1+\epsilon))}{\text{B}(\alpha,\beta)}\frac{\alpha+\beta}{\alpha+\beta(1+\epsilon)}
\]
 then we have for $c'\in\mathds{R}$,
\begin{align*}
\mathds{E}_{q_{\theta}}\left[\bar{\omega}(X,1)^{c'}\right] & =\int_{[0,1]}\bar{\omega}(x,1)^{c'}\text{Beta}(x;1+\alpha,\beta(1+\epsilon))dx\\
 & =\bar{\mathcal{Z}}^{c'}\int_{[0,1]}(1-x)^{-\epsilon\beta c'}\text{Beta}(x;1+\alpha,\beta(1+\epsilon))dx\\
 & =\bar{\mathcal{Z}}^{c'}\int_{[0,1]}x^{-\epsilon\beta c'}\text{Beta}(x;\beta(1+\epsilon),1+\alpha)dx\\
 & \le\frac{\bar{\mathcal{Z}}^{c'}}{\text{B}(\beta(1+\epsilon),1+\alpha)}\int_{[0,1]}x^{\beta(1+\epsilon(1-c'))-1}dx,
\end{align*}
where the third equality exploits symmetry properties of the Beta
distribution. We wish to show that there exists $c'$ such that
\[
\sup_{\beta}\mathds{E}_{q_{\theta}}\left[\bar{\omega}(X,1)^{c'}\right]<\infty.
\]
Firstly, we note that $\sup_{\beta\in\Theta}\bar{\mathcal{Z}}<\infty$
and $\sup_{\beta\in\Theta}\text{B}(\beta(1+\epsilon),1+\alpha)<\infty$
as $\Theta$ is compact. Secondly, we see that if $c'<0$ then the
integral is finite, thereby proving the second part of the Proposition.

For the first part of the proposition consider $c'$ such that
\[
0<c'-1\le\frac{1}{\epsilon}\left(1-\frac{\delta}{\underline{\beta}}\right),
\]
for some $0<\delta<\underline{\beta}$. This implies
that $\underline{\beta}(1-\epsilon(c'-1))\ge\delta$ and as a result
we have

\begin{align*}
\frac{\bar{\mathcal{Z}}^{c'}}{\text{B}(\beta(1+\epsilon),1+\alpha)}\int_{[0,1]}x^{\beta(1+\epsilon(1-c'))-1}dx & \le\frac{\bar{\mathcal{Z}}^{c'}}{\text{B}(\beta(1+\epsilon),1+\alpha)}\int_{[0,1]}x^{\delta-1}dx<\infty.
\end{align*}
Furthermore, for fixed $c'>1$ we see that $\underline{\beta}(1-\epsilon(c'-1))\ge\delta$
is also equivalent to requiring that $\epsilon\le\frac{1-\frac{\delta}{\underline{\beta}}}{c'-1}$
which can be satisfied for $\epsilon$ sufficiently small enough thereby
proving the final part of the proposition. 

Repeating the above argument
for $y=0$ we have that for $\bar{\mathcal{Z}}':=\frac{\text{B}(\alpha(1+\epsilon),\beta)}{\text{B}(\alpha,\beta)}\frac{\alpha+\beta}{\alpha(1+\epsilon)+\beta}$,
\begin{align*}
\mathds{E}_{q_{\theta}}\left[\bar{\omega}(X,0)^{c'}\right] & =\bar{\mathcal{Z}}'^{c'}\int_{[0,1]}x^{-\epsilon\alpha c'}\text{Beta}(x;\alpha(1+\epsilon),1+\beta)dx\\
 & \le\frac{\bar{\mathcal{Z}}'^{c'}}{\text{B}(\alpha(1+\epsilon),1+\beta)}\int_{[0,1]}x^{\delta'-1}dx<\infty,
\end{align*}
for $0<\delta'<\alpha$ .

\subsection{Description of referendum survey data}

\label{subsec:surveydescription}

We use data from the 13 wave internet survey study (as of June 2018)
\citep{https://doi.org/10.15127/1.293723} comprising 68,625 respondents
in total, with the number of respondents varying between waves. We
first subset the data into those in the four annual waves 1, 4, 7
and 11 occurring between February and May of 2014-2017. Of these 7,729
answered either `Stay/remain in the EU' or `Leave the EU' to the question
`If you do vote in the referendum on Britain\textquoteright s membership
of the European Union, how do you think you will vote?' in each wave.
We filter out those that answered `Don't know' to the question `How
do you think the general economic situation in this country has changed
over the last 12 months?' reducing the sample by 5 respondents. Finally,
we perform inference only on waves 1, 4, and 7, the waves prior to
the EU referendum on $23^{rd}$ June 2016. For simplicity, we do not
take into account respondent weighting.

\end{document}